\documentclass[aps,prl,english,superscriptaddress,twocolumn,tightenlines,showpacs,longbibliography,floatfix,amsfonts]{revtex4-2}
\usepackage[T1]{fontenc}
\usepackage[utf8]{inputenc}
\usepackage[normalem]{ulem}
\usepackage{textcomp}
\usepackage{graphicx,float}
\usepackage[dvipsnames,svgnames]{xcolor}
\usepackage{mathrsfs,mathtools,xfrac}
\usepackage{amsmath,amssymb,amsthm,amscd,bm,bbm}
\usepackage{txfonts} % better font, but put after ams math to avoid conflicts
\usepackage{afterpackage,multirow,comment}
\usepackage{array,booktabs,minibox}
\newcolumntype{P}[1]{>{\centering\arraybackslash}p{#1}} % control the alignment and width
\usepackage[colorlinks,bookmarks=false,citecolor=blue,linkcolor=blue,urlcolor=blue,filecolor=blue]{hyperref}
\usepackage{wasysym}
\newtheorem{thm}{Theorem}

\newtheorem{rem}{Remark}
\newtheorem{lem}{Lemma}

\newtheorem{cor}{Corollary}

\usepackage{babel}

\newtheorem{ex}{Example}

\newcommand{\im}[1]{{\operatorname{im} \, #1}}
\newcommand{\Ann}[1]{{\operatorname{Ann} \, #1}}

\newcolumntype{C}{>{$}c<{$}}
\AtBeginDocument{
	\heavyrulewidth=.08em
	\lightrulewidth=.05em
	\cmidrulewidth=.03em
	\belowrulesep=.65ex
	\belowbottomsep=0pt
	\aboverulesep=.4ex
	\abovetopsep=0pt
	\cmidrulesep=\doublerulesep
	\cmidrulekern=.5em
	\defaultaddspace=.5em
}

\begin{document}

	\title{Anyon Theory and Topological Frustration of High-Efficiency
		\\Quantum Low-Density Parity-Check  Codes}

	\author{Keyang Chen}
	\affiliation{Institute of Theoretical Physics, Chinese Academy of Sciences, Beijing 100190, China}
	\affiliation{Hefei National Research Center for Physical Sciences at the Microscale and School of Physical Sciences, University of Science and Technology of China, Hefei 230026, China}
	\affiliation{School of Physical Sciences, University of Chinese Academy of Sciences, Beijing 100049, China.}
	
	\author{Yuanting Liu}
	\affiliation{Institute of Theoretical Physics, Chinese Academy of Sciences, Beijing 100190, China}
	\affiliation{School of Physical Sciences, University of Chinese Academy of Sciences, Beijing 100049, China.}
	
	\author{Yiming Zhang}
	\affiliation{Hefei National Research Center for Physical Sciences at the Microscale and School of Physical Sciences, University of Science and Technology of China, Hefei 230026, China}
	\affiliation{Shanghai Research Center for Quantum Science and CAS Center for Excellence in Quantum Information and Quantum Physics, University of Science and Technology of China, Shanghai 201315, China}

	\author{Zijian Liang}
	\affiliation{International Center for Quantum Materials, School of Physics, Peking University, Beijing 100871, China}
	
	\author{Yu-An Chen}
	\affiliation{International Center for Quantum Materials, School of Physics, Peking University, Beijing 100871, China}
	
	\author{Ke Liu}
	\email{ke.liu@ustc.edu.cn}
	\affiliation{Hefei National Research Center for Physical Sciences at the Microscale and School of Physical Sciences, University of Science and Technology of China, Hefei 230026, China}
	\affiliation{Shanghai Research Center for Quantum Science and CAS Center for Excellence in Quantum Information and Quantum Physics, University of Science and Technology of China, Shanghai 201315, China}

	\author{Hao Song}
	\email{songhao@itp.ac.cn}
	\affiliation{Institute of Theoretical Physics, Chinese Academy of Sciences, Beijing 100190, China}
	
	\date{\today}
	
	\begin{abstract}
		Quantum low-density parity-check (QLDPC) codes offer a promising path to low-overhead fault-tolerant quantum computation but lack systematic strategies for exploration.
		In this Letter, we establish a topological framework for studying the bivariate-bicycle codes, a prominent class of QLDPC codes tailored for real-world quantum hardware.
		Our framework enables the investigation of these codes through universal properties of topological orders.
		In addition to efficient characterizations using Gr\"obner bases, we also introduce a novel algebraic-geometric approach based on the Bernstein--Khovanskii--Kushnirenko theorem.
		This approach allows us to analytically determine how the topological order varies with the generic choices of bivariate-bicycle codes under toric layouts. 
		Novel phenomena are unveiled, including \emph{topological frustration}, where ground-state degeneracy on a torus deviates from the total anyon number, and \emph{quasi-fractonic mobility}, where anyon movement violates energy conservation.
		We demonstrate their intrinsic link to symmetry-enriched topological orders and derive an efficient method for generating finite-size codes. 
		Furthermore, we extend the connection between anyons and logical operators using Koszul complex theory. 
		Our Letter provides a rigorous theoretical basis for exploring the fault tolerance of QLDPC codes and deepens the interplay among topological order, quantum error correction, and advanced algebraic structures.
	\end{abstract}
	
	\maketitle
	
	\emph{Introduction}---Quantum error correction (QEC) is essential for building scalable quantum computers on noisy hardware~\cite{Shor95,Steane96,Knill97, Kitaev03,Gottesman97,Dennis02}.
	The standard QEC methods have predominantly relied on topological codes~\cite{Dennis02,Terhal15,Campbell17,Fowler12}, such as surface code~\cite{Kitaev03, Bravyi98} and color code~\cite{Bombin06, Bombin07}, which have been successfully demonstrated on diverse quantum platforms~\cite{Satzinger21,Krinner22,USTC22,USTC23,Google24a,Quantinuum24,QuEra24,Monz21,Blatt14}.
	Nevertheless, these approaches require millions of physical qubits to tackle problems of practical interest~\cite{Litinski19,Gidney21,Alexeev21,Beverland21,Beverland22}, primarily due to the low encoding efficiency.
	This challenge has fueled the exploration of quantum low-density parity-check (QLDPC) codes for low-overhead quantum computation~\cite{Gottesman14,MacKay04,Kovalev13,Tillich14,Bravyi14}.
	The promising recent advances in the asymptotic regimes~\cite{Hastings20,Panteleev21,Panteleev22,Panteleev22b,Breuckmann21a,Breuckmann21b,Leverrier22} has spurred enormous research activities to probe their physical relevance~\cite{Rakovszky23,Rakovszky24,DeRoeck24,Yin24,liang2024extracting, liang2024operator,Zhu25} and possible hardware implementations~\cite{Tremblay22,Bravyi24,Xu24,Nguyen24,Tamiya24}.
	In particular, growing efforts are devoted to developing QLDPC codes that are implementable on existing and near-term quantum architectures~\cite{Tremblay22, Wang22, Bravyi24, Xu24, Hong24, Lin24, Pecorari25, Ruiz25, Scruby24, Zhang25}. 
	Among them, the bivariate-bicycle (BB) codes~\cite{Bravyi24} introduced by IBM have gained significant prominence and sparked intense interest~\cite{Bravyi22, Berthusen25, Poole24, Gong24, Shaw24, Wang24, Sayginel24, eberhardt2024pruning,Cross24,Cowtan24,Eberhardt24,Williamson24}.
	These codes offer the potential for a ten-fold reduction in qubit overheads compared to standard methods, while demanding only moderate long-range resources.
	
	However, unlike topological codes, 
	the exploration of systematic frameworks for universal fault tolerance of BB codes remains in early stages~\cite{Cross24,Cowtan24,He25Extractors,Eberhardt24,Williamson24}.
	Extending the scope of topological order to QLDPC codes may bridge this gap and unlock opportunities to optimize both resource efficiency and fault tolerance.
	Owing to their diverse forms and unclear scalability, it is nevertheless uncertain whether a generic topological correspondence for BB codes can be established. 
	Even if achievable, it would prompt further crucial questions regarding their nature and the applicability of conventional topological principles.

	In this Letter, we develop a topological framework to explore the universal properties of BB codes.
	Going beyond case-specific bases, we prove that, except for a single pathological case, \emph{all} BB codes with the toric layout are topological, and establish a global characterization of their topological orders.
	Moreover, we uncover a distinctive \emph{topological frustration} that induces a strong size-dependency in ground-state degeneracy (GSD) and logical dimensions.
	It deviates sharply from the general expectation of constant GSD on a torus, underscoring a key departure from standard topological codes.
	Finally, we show that the full characterization of topological orders in BB codes leads to a symmetry-enriched topological perspective,  deepening the connection between quantum codes and topological order.

	In addition, we also make notable methodological advancements.
	We creatively integrate the Bernstein--Khovanskii--Kushnirenko theorem from algebraic geometry and a dual property of Koszul complexes in homological algebra in the study of QLDPC codes.
	Efficient algorithms for characterizing topological structures are also introduced.
	Our Letter represents a fruitful interplay between various research areas and holds broad multi-disciplinary interest.

	\emph{Code Hamiltonian}---The BB codes with the toric layout can be defined on a square lattice with qubits living on links.
	Each of their stabilizer generators (checks) involves six qubits, including two remote ones that lead to long-range connections in syndrome measurements~\cite{Bravyi24}.
	Similar to the toric code, the $X$ and $Z$ checks can be ambiguously labeled by vertices ($v$) and plaquettes ($p$) of the lattice, respectively.
	The code Hamiltonian is given by summing over these checks,
	\begin{align}\label{eq:Hamiltonian}
		H = -\sum_{v} h_{X,v} - \sum_{p} h_{Z,p},
	\end{align}
	where each $h_{X,v}$ ($h_{Z,p}$) is a product of six Pauli $X$ ($Z$) operators. Lattice translation symmetry is implied.
	
	\begin{figure}[t]
		\includegraphics[width=1\columnwidth]{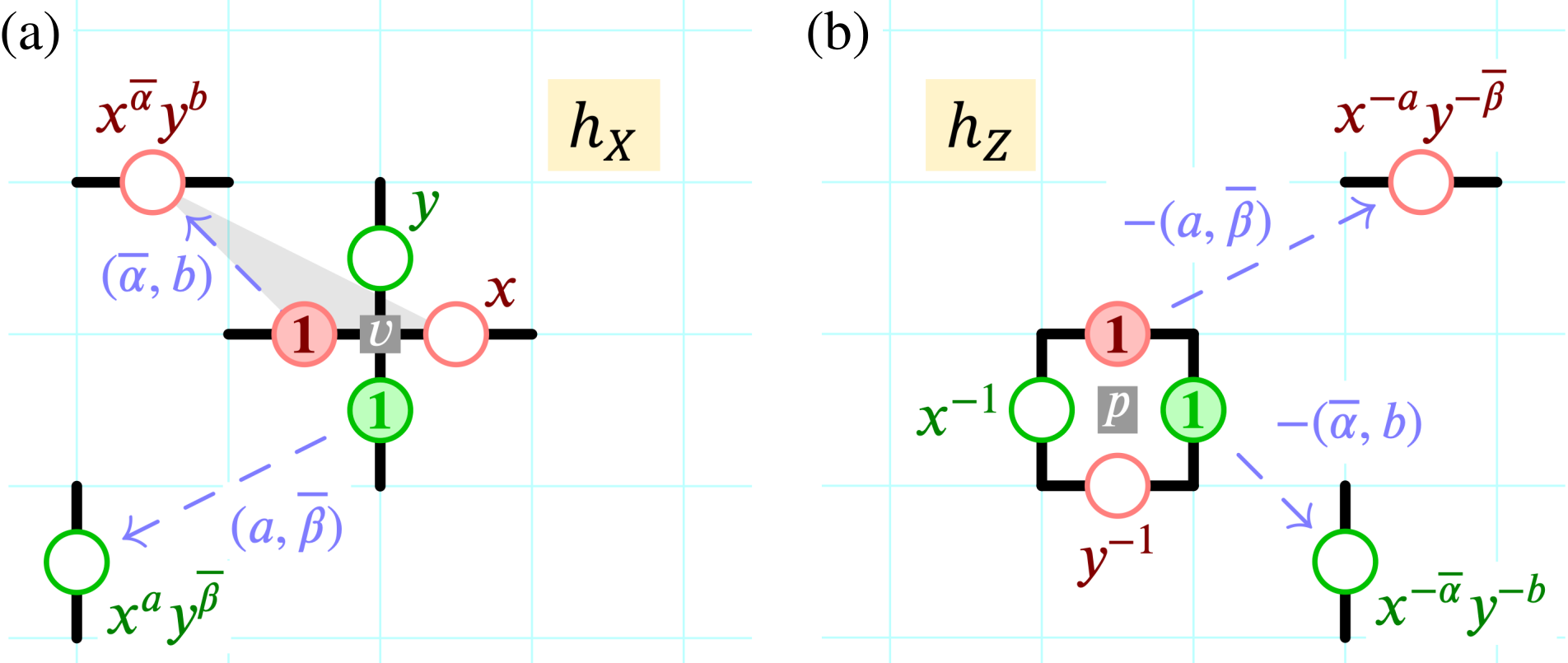}
		\caption{Code Hamiltonian of BB codes with the toric layout.
			Each stabilizer generator, $h_{X}$ (a) or $h_{Z}$ (b), involves six Pauli $X$ or $Z$ operators, respectively, including two remote ones compared to the toric code.
			Locations of the involved qubits are explicitly labeled; e.g.,
			$x^{\overline{\alpha}} y^b$ labels a qubit on the link at coordinates $(\overline{\alpha}, b)$ relative to the reference link $x^0y^0=1$, where horizontal and vertical links are viewed as two sublattices, indicated by red and green cycles, respectively.
			The dashed arrows indicate the relative positions of remote qubits with respect to their corresponding origins (filled circles).
			The shaded triangular region in (a) represents a Newton polytope later used for computing $Q$.}
		\label{fig:model}
	\end{figure}
	
	Explicit forms of $h_X$ and $h_Z$ can be conveniently specified using a polynomial representation~\cite{Haah13}.
	Concretely, we associate a lattice site at $(i,j)$ with a monomial $x^i y^j$, and use the Laurent polynomial $f(x,y) = \sum_{\{i, j\}} x^i y^j$ to denote collections of sites.
	The $h_X$ and $h_Z$ terms are expressed as
	\begin{align}\label{eq:checks}
		h_X = \begin{pmatrix} 
			f \\  g
		\end{pmatrix}, 
		\quad
		h_Z = \begin{pmatrix} 
			g^\ast \\  f^\ast
		\end{pmatrix},
	\end{align}
	with
	\begin{gather}\label{eq:generator}
		\begin{aligned}
			f = 1+x + x^{\overline{\alpha}} y^b, \\
			g = 1+y + y^{\overline{\beta}} x^a.
		\end{aligned}
	\end{gather}
	Here, the first (second) component of 
	$\left(\begin{smallmatrix}
		\square \\ \square 
	\end{smallmatrix}\right)$
	corresponds to qubits living on horizontal (vertical) links of the lattice.
	Polynomials $f$ and $g$ specify the six qubits involved in an $h_X$ check.
	Their spatial inversions, denoted as  $f^\ast (x,y) \coloneqq f(x^{-1}, y^{-1})$ and 
	$g^\ast (x,y) \coloneqq g(x^{-1}, y^{-1})$, accordingly define an $h_Z$ check.
	See Fig.~\ref{fig:model} for an illustration.

	The $x^{\overline{\alpha}} y^b$ and $y^{\overline{\beta}} x^a$ terms in Eq.~\eqref{eq:generator} specify the locations of remote qubits in the interactions.
	Each combination $(\overline{\alpha}, \overline{\beta}, a, b)$ corresponds to a different Hamiltonian.
	
	We denote the model as $\mathbb{BB}(\overline{\alpha}, \overline{\beta}, a, b)$ code. Without loss of generality, we set $\overline{\alpha}\equiv-\alpha$ and $\overline{\beta}\equiv-\beta$ with $\alpha \geq \beta\geq0$. 
	Other parameter ranges can be obtained by redefining the coordinate system~\cite{SM1};  
	all codes highlighted in Ref.~\cite{Bravyi24} are encompassed.

	\emph{Topological condition and topological order of BB codes}---A code is topological if its thermodynamic limit exhibits a gapped topological order.
	Formally, this requires the following chain complex to be \emph{exact}:
	\begin{equation}\label{eq:chain}
		R\xrightarrow{\; h_{X}={f \choose g}\;}R^{2}\xrightarrow{h_{Z}^{\dagger}=\left(g\ f\right)}R,
	\end{equation}
	namely, $\text{im}\,h_{X}=\ker h_{Z}^{\dagger}$.
	This is known as the topological condition and guarantees that no finite product of Pauli operators forms a nontrivial logical operator, hence the code distance is macroscopic~\cite{Haah13,Haah13_thesis}.
	Here, $R\coloneqq\mathbb{F}_{2}[x^{\pm},y^{\pm}]$ is the Laurent polynomial ring over $\mathbb{F}_{2}=\left\{ 0,1\right\} $, while $\text{im}$ and $\ker$ denote the image and kernel of a map, respectively.

	Justifying whether the topological condition generically holds for BB codes is a challenging task due to the infinite number of possible combinations of $(\overline{\alpha}, \overline{\beta}, a, b)$.
	Despite this complexity, we transcend case-specific verification and establish a general proof demonstrating that all BB codes defined in Eq.~\eqref{eq:generator} are topological for arbitrary $\alpha,\beta\in\mathbb{Z}_{\geq0}$ and $a,b\in\mathbb{Z}$, with the sole exception of $(\overline{\alpha}, \overline{\beta}, a, b) = (0,0,1,1)$. 
	
	Remarkably, we identify an equivalent relation between the exactness criterion and the vanishing of annihilators $\mathrm{Ann}_{R/\left(g\right)}\left(f\right)$ and $\mathrm{Ann}_{R/\left(f\right)}\left(g\right)$ for $f,g \neq 0$.
	This means $f$ and $g$ should share no common factor other than a monomial.
	Then, by generalized Eisenstein's criterion, we can show $f$ and $g$ with a toric layout are always irreducible.
	Finally, using this irreducibility, we determine those $\left(\alpha,\beta,a,b\right)$ values satisfying $\mathrm{Ann}_{R/\left(g\right)}\left(f\right)=\mathrm{Ann}_{R/\left(f\right)}\left(g\right)=0$ and conclude the topological condition is met for all $(\overline{\alpha}, \overline{\beta}, a, b) \neq (0,0,1,1)$.
	See SM~\cite{SM1} for the formal proof.

	Having the topological nature established, the next essential question is to determine the specific topological orders that BB codes can realize. This is achieved by identifying and enumerating the distinct anyon types they support. Specifically, all the anyon types in a BB code can be labeled by
	\begin{equation} \label{eq:C}
		\mathscr{C}\coloneqq\frac{R}{\left(f,g\right)}\oplus\frac{R}{(f^\ast,g^\ast)}.
	\end{equation}
	Here, the quotient spaces $R/(f,g)$ and $R/(f^\ast,g^\ast)$ correspond to the equivalence classes of $h_Z$ and $h_X$ excitations, respectively, defined modulo all possible local moves. They can be interpreted as electric ($e$) and magnetic ($m$) charges in the emergent gauge theory.

	The total number of anyon types is then given by
	\begin{equation}\label{eq:GSD}
		\left|\mathscr{C}\right|=2^{2Q} = 2^{k_{\max}},
		\qquad\text{with}\;Q \coloneqq \dim_{\mathbb{F}_{2}}\frac{R}{\left(f,g\right)}.
	\end{equation}
	using the $e$-$m$ duality $\dim_{\mathbb{F}_2}R/(f,g) = \dim_{\mathbb{F}_2} R/(f^\ast,g^\ast)$.
	Thus, the topological order of a BB code corresponds to a $\mathbb{Z}_2^Q$ gauge theory, with $Q$ being the \emph{topological index}.
	It will be clear that $k_{\max} =2Q$ gives the \emph{maximal} number of logical qubits a BB code can encode.
	
	A standard approach to proceed would be computing $Q$ case by case using the Gr\"obner-basis technique~\cite{Haah13}, which we shall adopt later when  examining specific examples.
	Nevertheless, this method is insufficient and impractical for revealing the generic behavior of $Q$, as the family of BB codes is unbounded.
	Instead, we introduce a more elegant and efficient approach to proceed.
	By leveraging the Bernstein--Khovanskii--Kushnirenko (BKK) theorem from algebraic geometry, we relate the topological index $Q$ to the mixed volume of Newton polytopes~\cite{SM1}.
	This enables us to derive analytical expressions for the generic dependence of $Q$ on all possible combinations of $(\overline{\alpha}, \overline{\beta}, a, b)$.

	\begin{figure}[t]
		\includegraphics[width=1\columnwidth]{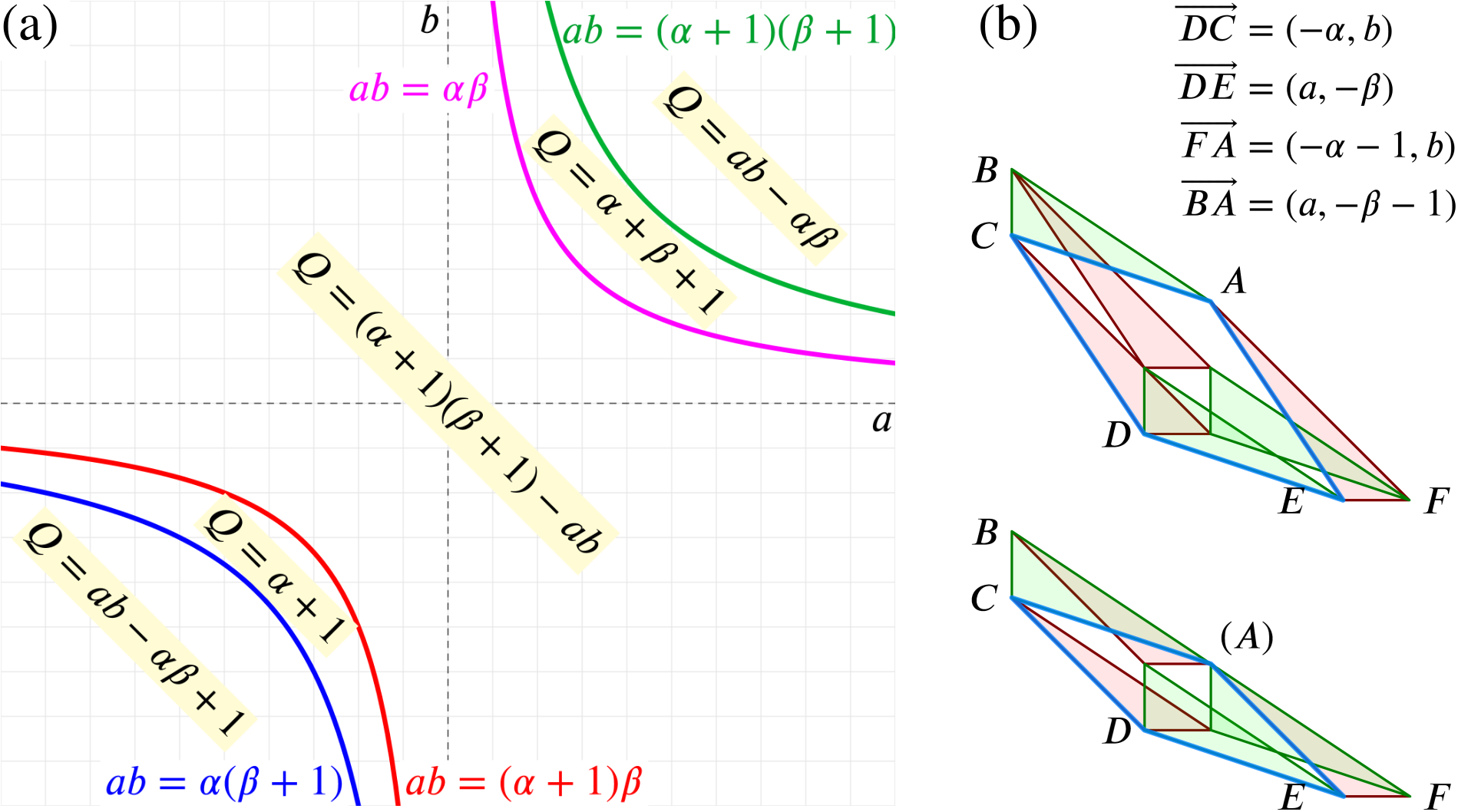}
		\caption{(a) The value of topological index $Q$ for $\mathbb{BB}(\overline{\alpha},\overline{\beta},a,b)$ codes.
			There are five distinct regimes separated by four hyperbolic curves.  The expression of $Q$ for each regime and separating boundary is indicated.
			(b) Illustration of the mixed volume.
			The top-right regime (top) and its boundary (bottom) in (a) are shown for example.
			Each red (cyan) triangle represents, up to translation, the Newton polytope associated with $f$ (respectively, $g$). 
			The mixed volume is given by hexagon $ABCDEF$ subtracting triangles $ABC$ and $AEF$, resulting in the parallelogram outlined in blue. 
		}
		\label{fig:anyonnum}
	\end{figure}
	
	The results of generic BB codes under the toric layout can be distilled into a single plot, as presented in Fig.~\ref{fig:anyonnum}.
	We organize them into five regimes, separated by four hyperbolic curves.
	A neat formula of $Q(\overline{\alpha}, \overline{\beta}, a, b)$ is discovered for each regime and each separating boundary.
	The value of $Q$, so the topological order, of an associated BB code can then be determined conveniently.
	For instance, the $\mathbb{BB}(\overline{1},\overline{1},3,3)$ code belongs to the middle regime of Fig.~\ref{fig:anyonnum}, and one reads out its $Q=8$.
	
	These results also prove that BB codes can encode arbitrarily many logical qubits despite moderate resources of long-range connectivity.
	The analytical expressions in Fig.~\ref{fig:anyonnum} manifest that $k_{\max}=2Q$ grows either linearly or quadratically with the distance of the remote interactions in $f$ and $g$.

	\emph{Topological frustration and quasi-fractonic mobility}---In practical QEC experiments, codes are inherently finite in size.
	It is not \textit{a priori} clear whether or how their finite-size properties relate to the thermodynamic limit.
	Remarkably, we show that the topological framework remains capable of capturing code properties for all sizes.
	Nevertheless, the relationship between finite-size features and the anyon perspective reveals an unconventional and intriguing behavior.

	Specifically, consider codes living on a torus of size $\ell \times m$, where $x^\ell=y^m=1$.
	In analogy to Eq.~\eqref{eq:GSD}, the size-dependent logical dimension is given by
	\begin{equation}\label{eq:GSDtorus}
		k(\ell, m)\coloneqq 2\dim_{\mathbb{F}_{2}}\frac{R}{\mathcal{I}_{\ell m}} \leq k_{\max},
	\end{equation}
	with PBCs encoded into the ideal $\mathcal{I}_{\ell m} \coloneqq \big(f,g,x^\ell-1,y^m-1\big)$ of $R$.
	A rigorous proof of this formula is provided in SM~\cite{SM1}. For each $\mathbb{BB}(\overline{\alpha}, \overline{\beta}, a, b)$ code, one can determine $k(\ell, m)$ by computing the Gr\"obner bases $\mathcal{G}_{\ell m}$ of the ideal $\mathcal{I}_{\ell m}$.
	Gr\"obner bases are special generating sets of ideals and can be found efficiently utilizing Buchberger's algorithm~\cite{SM1}.

	For concreteness, we again take the $\mathbb{BB}(\overline{1},\overline{1},3,3)$ code as an example, while the analysis applies to generic BB codes.
	We compute $k(\ell, m)$ for \emph{all} finite lattices $(l, m)$, with details kept in SM~\cite{SM1}, and find their values can be subsumed into four classes,
	\begin{align}\label{eq:k_33}
		k(\ell, m) = \begin{cases}
			16, & \gcd(\ell,m) \in 12\mathbb{Z}, \\
			12, & \gcd(\ell,m) \in 6\mathbb{Z} -12\mathbb{Z}, \\
			8, & \gcd(\ell,m) \in 3\mathbb{Z} - 6\mathbb{Z},\\
			0, & \text{others}. 
		\end{cases}
	\end{align}

	Notably, $k(\ell, m)$ matches the anyon count $k_{\max} = \log\left|\mathscr{C}\right|=16$ for system sizes where the greatest common divisor $\gcd(\ell,m) \in 12\mathbb{Z}$. 
	However, it does not saturate $k_{\max}$ in most cases, but instead strongly oscillates with $(l, m)$.
	This sharply violates the general notion that $\text{GSD}=2^k$ on a torus should be constant and equal to the number of anyon types.
	We therefore refer to this unsaturation of logical dimension as \emph{topological frustration}.
	
	In fact, several frequently referenced examples of BB codes, including the $[[72,12,6]]$, $[[108,8,10]]$, $[[144,12,12]]$ codes~\cite{Bravyi24,Poole24,Berthusen25,Cross24,Gong24,Shaw24}, are just different finite-size instantiations of the $\mathbb{BB}(\overline{1},\overline{1},3,3)$ code Hamiltonian.
	The first and the third one, with $k=12$, fall into the same class in Eq.~\eqref{eq:k_33}, while the second one, with $k=8$, belongs to a different class.
	The considerable gaps between their $k$'s values and $k_{\max} = 16$ are precisely due to the topological frustration.

	\begin{figure}[t]
		\includegraphics[width=1\columnwidth]{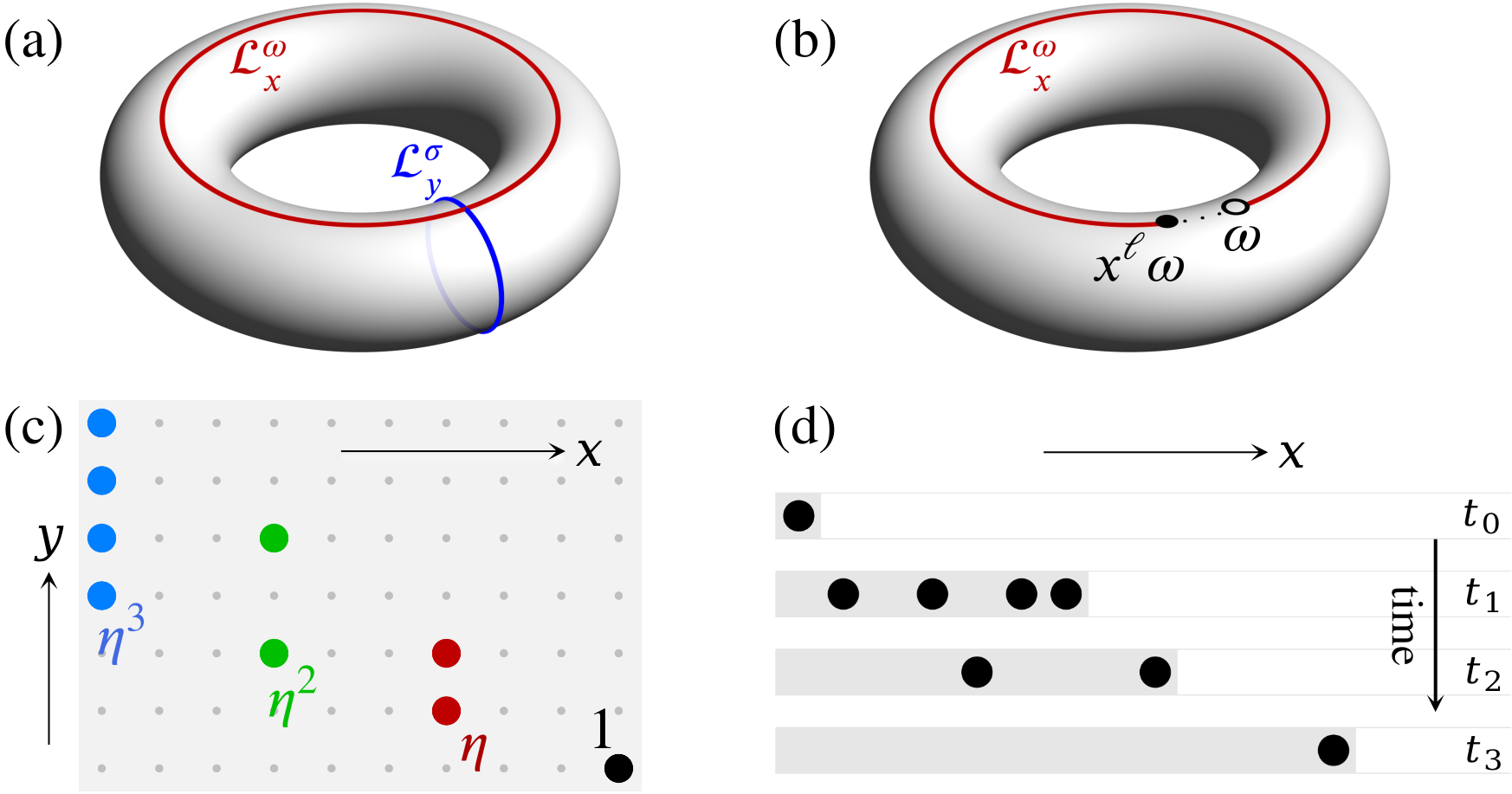}
		\caption{(a) Illustration of non-contractible loop operators $\mathcal{L}_{x}^{\omega}$
			and $\mathcal{L}_{y}^{\sigma}$ on a torus for anyons $\omega$ and $\sigma$. 
			(b) Loop operator $\mathcal{L}_{x}^{\omega}$
			corresponds to the propagation of $\omega$ around the system in the
			$x$ direction. (c $\&$ d) Anyon movements and excitation patterns of the $\mathbb{BB}(\overline{1},\overline{1},3,3)$ code.
			(c) Small gray dots denote $h_{Z}$-checks, while large colored dots (black, red, green, blue) represent $h_{Z}$-excitations.
			The excitation patterns illustrate a sequence of local moves $1\rightarrow\eta\rightarrow\eta^{2}\rightarrow\eta^{3}$
			(right to left), with $\eta=x^{-3}(y+y^{2})$. 
			(d) Hopping process of $h_{Z}$-excitations (black dots) along the $x$ axis and at four time moments.
			The single anyon at $t_{0}$ is mapped to its original form at $t_{3}$ after a finite number of hops.
			The bars indicate the hopping distance.}
		\label{fig:anyon_hop}
	\end{figure}

	Topological frustration reflects the breakdown of the conventional notion that non-contractible loop operators exist for all anyon types (Fig.~\ref{fig:anyon_hop}(a)). 
	In BB codes, this breakdown arises because anyon propagation deforms their shape, potentially violating the scenario in Fig.~\ref{fig:anyon_hop}(b), which is required for the existence of such logical loop operators.

	For an explicit demonstration on how anyons move here, we revisit the $\mathbb{BB}(\overline{1},\overline{1},3,3)$ code in the infinite-size setting, and focus on its $h_Z$-excitations (additional examples are provided in SM~\cite{SM1}).
	Allowed local moves correspond to elements of $\im{h_Z^\dagger}$.
	An illustrative process, generated by $g=1+y+x^{3}y^{-1} \in \im{h_Z^\dagger}$, is shown in Fig.~\ref{fig:anyon_hop}(c). Note that $1\rightarrow \eta=x^{-3}(y+y^{2})$ and its translations are permitted because $1+\eta\propto g$.
	This “one-to-many” hopping, which violates energy conservation, differs from toric code anyons and resembles fracton behavior in Haah’s code~\cite{Haah11,Vijay16,Song24}.
	
	However, unlike fractons, excitations in BB codes can still move along lines and permit ``one-to-one'' hopping over large enough distances, as depicted in Fig.~\ref{fig:anyon_hop}(d). Detailed analysis can be made again with the aid of Gr\"obner basis. The topological charges of $h_Z$ excitations are labeled by elements of $R/\im h_Z^\dagger = R/(f,g)$.
	As a useful trick, we identify $R/(f,g)$ with $\mathbb{F}_2[x,y,\overline{x},\overline{y}]/\mathcal{I}_\infty$, where $\mathcal{I}_{\infty} \coloneqq \big(f,g,x\overline{x}-1, y\overline{y}-1\big)$. Here, Laurent polynomials are treated as ordinary polynomials in variables $x,y,\overline{x},\overline{y}$, with $\overline{x}$ and $\overline{y}$ corresponding to $x^{-1}$ and $y^{-1}$, respectively. Using Buchberger's algorithm, we compute the Gr\"obner basis $\mathcal{G}_{\infty}$ of the ideal $\mathcal{I}_\infty$ for the $\mathbb{BB}(\overline{1},\overline{1},3,3)$ code:
	\begin{multline}\label{eq:GB_33}
		\mathcal{G}_\infty=\{x^{6}+x^{5}+x^{3}+x+1,\;x^{2}y+xy+y+x^{5}+x+1,\;y^{2}+\\
		y+x^{3},\;\overline{x}+x^{5}+x^{4}+x^{2}+1,\;\overline{y}+y+x^{5}+x^{4}+x^{3}+1\}
	\end{multline}
	with respect to the lexicographic order $\overline{y}>\overline{x}>y>x$. 
	The first polynomial in $\mathcal{G}_\infty$ generates an anyon hopping process along the 
	$x$-direction (Fig.~\ref{fig:anyon_hop}(d)): an anyon initialized at $t_0$ evolves into multiple excitations, but eventually recombines into its original form at time $t_3$; see SM~\cite{SM1} for details.

	These mobility features can be understood as a manifestation of nontrivial action of lattice translations on anyon types. 
	Specifically, monomials of order below the leading terms of the polynomials in $\mathcal{G}_\infty$ represent a linear basis for topological charges of $h_Z$ excitations in the  $\mathbb{BB}(\overline{1},\overline{1},3,3)$ code.
	There are eight such monomials $\{1,x,x^{2},x^{3},x^{4},x^{5},y,xy\}$,
	reconfirming $Q=8$ from the BKK calculation. 
	This basis reveals a nontrivial translation action: for instance, the charge represented by 
	``1'' and its $x$-translate ``$x$'' are distinct, forbidding the smooth move $1\rightarrow x$.
	Yet, due to the finiteness of $R/(f,g)$, anyons return to their initial types after a finite number of translations, enabling “one-to-one” hopping at large distances.

	We refer to this behavior as \textit{quasi-fractonic mobility}, in contrast to the smooth motion of toric code anyons and the persistent immobility of fractons. Notably, it implies that not all anyons return to their original types after traversing non-contractible loops on a torus, thereby reducing the number of realizable logical qubits.
	
	\begin{figure}[t]
		\includegraphics[width=1\columnwidth]{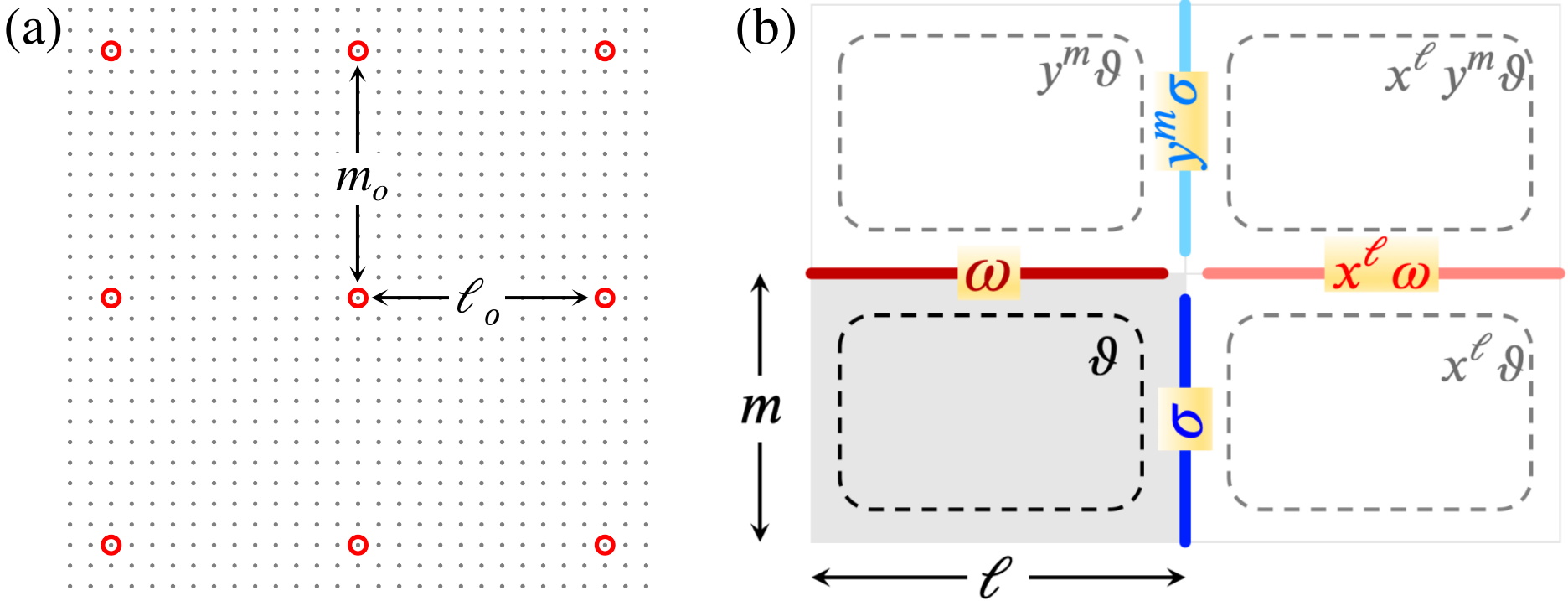}
		\caption{(a) Mobility sublattice $\Lambda_{\mathscr{C}}$ for the $\mathbb{BB}(\overline{1},\overline{1},3,3)$ code. Gray dots denote the original lattice $\Lambda$; red circles mark $\Lambda_{\mathscr{C}}$, with lattice spacings $\ell_o$ and $m_o$ defining the anyon periods. 
			(b) Generic correspondence between anyons and logical operators. The gray
			rectangle represents one copy of the system, with three replicas to illustrate
			PBCs. The $\omega$ and $\sigma$ strings correspond to $\mathcal{L}_{x}^{\omega}$
			and $\mathcal{L}_{y}^{\sigma}$ in Fig.~\ref{fig:anyon_hop}(a). These strings, along
			with their translations $x^{\ell}\omega$ and $y^{m}\sigma$, must
			be joined without leaving excitations to form a logical operator.
			Such a logical operator may still be trivial if the strings can be canceled
			by the four dashed loops, which represent the propagation of anyons $\vartheta$, $x^{\ell}\vartheta$, $y^{m}\vartheta$, and
			$x^{\ell}y^{m}\vartheta$, respectively. }
		\label{fig:periodicity}
	\end{figure}

	\emph{SET structure and generic anyon-logic duality}---The interplay between anyon types and lattice translations means that BB codes essentially feature symmetry-enriched topological (SET) orders~\cite{Wen02, Essin2013, Mesaros2013, Lu16, Barkeshli19, Barkeshli2022Classification, Delfino23}.
	Therefore, their full characterization requires not only identifying the emergent $\mathbb{Z}_2^Q$ gauge theory but also understanding how symmetries act on anyons.

	For systematic analysis, we introduce the notion of \emph{mobility sublattice} $\Lambda_{\mathscr{C}} = \{\lambda\}$, defined as the subgroup of translations that preserve all anyon types.
	For BB codes, it is determined by $\lambda- 1 \in \left(f,g\right)_R$, where  $ \left(f,g\right)_R$ denotes the ideal of $R$ generated by $f$ and $g$.
	The lattice spacings of $\Lambda_{\mathscr{C}}$ set emergent length scales, in units of which anyons move freely and within which anyons exhibit fracton-like behavior.
	
	We refer to the lattice spacings of $\Lambda_{\mathscr{C}}$ in the $x$ and $y$ directions as \textit{anyon periods} $\ell_o$ and $m_o$ (Fig.~\ref{fig:periodicity}(a)), which govern the code behavior under PBCs $x^\ell = y^m=1$. Topological frustration is absent, namely, $k(\ell, m)$ reaches its maximum $k_{\mathrm{max}} = 2Q$, if and only if $\ell_o$ divides $\ell$ and $m_o$ divides $m$.

	More remarkably, we prove a universal relation controlling the size dependence of $k(\ell,m)$~\cite{SM1},
	\begin{equation}\label{eq:k_relation}
		k\left(\ell,m\right)=k\left(\gcd\left(\ell,\ell_o\right),\gcd(m,m_o)
		\right).
	\end{equation}
	It reveals that $k$ depends only on the greatest common divisors
	$\gcd\left(\ell,\ell_o\right)$ and $\gcd(m,m_o)$, and periodicity follows naturally as $k\left(\ell,m\right)=k\left(\ell+\ell_o,m\right)=k(\ell,m+m_o)$.

	This relation provides an efficient way to search for finite-size codes.
	Clearly, for a code Hamiltonian $\mathbb{BB}(\overline{\alpha}, \overline{\beta}, a, b)$, its finite-size instantiations can accommodate logical qubits only when $(\ell, m)$ has non-trivial gcds with $(\ell_o, m_o)$.
	Thus, once the underlying anyon periods are known, we can predict the non-trivial sizes and organize them into appropriate sequences, avoiding brute-force examination of all $k(\ell, m)$.
	Moreover, it is sufficient to compute only the smallest size in each sequence since $k\left(\ell,m\right)$ is periodic.

	In general, anyon periods $(\ell_o, m_o)$ can be very large, but we develop an efficient algorithm capable of finding billion-sized periods within seconds.
	The key idea is to reduce the problem to computing polynomial periods, which can be done efficiently~\cite{SM1}.
	Take the $\mathbb{BB}(\overline{1},\overline{1},3,\overline{3})$ code, for instance.
	We find it has five non-trivial size sequences, with $k_{\max} = 26$ and $(\ell_o, m_o) = (762,762)$.
	The $[[288,12,18]]$ code with $k=12$ highlighted in Ref.~\cite{Bravyi24} is one of its finite-size instantiations.
	See SM~\cite{SM1} for details with various other examples.

	Another important connection yet to be established is the generic correspondence between the logical space and the anyon space, which rationalizes the anyon interpretation of logical operators. 
	Although such correspondence is intuitive for conventional topological codes, justifying it for BB codes is non-trivial due to the presence of topological frustration.
	
	Nevertheless, we demonstrate that such a correspondence exists and can be explicitly constructed. It can be viewed as a duality property between two Koszul complexes.
	The first complex is $R_{\ell, m} \longrightarrow R_{\ell, m}^2 \longrightarrow R_{\ell, m}$, corresponding to the BB codes  (\ref{eq:chain})
	with PBCs imposed (i.e., $R$ replaced by $R_{\ell, m}\coloneqq R/(x^{\ell}-1,y^{m}-1)$). 
	Its homology defines logical operators.
	The second complex is 
	\begin{equation}\label{eq:koszul}
		\frac{R}{\left(f,g\right)}\xrightarrow{\partial_{2}={y^{m}-1 \choose x^{\ell}-1}}\frac{R}{\left(f,g\right)}\oplus\frac{R}{\left(f,g\right)}\xrightarrow{\partial_{1}=\left(x^{\ell}-1\;\;y^{m}-1\right)}\frac{R}{\left(f,g\right)},
	\end{equation}
	where $\ker \partial_1$ represents the condition that anyon-labeled string operators can be connected to produce logical operators, and $\im{\partial_2}$ describes those equivalent to contractible loops, as visualized in  Fig.~\ref{fig:periodicity}(b). 
	We prove that the homologies of these two Koszul complexes are isomorphic, providing a systematic way to label logical operators by anyons, even in cases with topological frustration; see SM~\cite{SM1} for details.

	\emph{Summary and Outlook}---In this Letter, we established a topological framework that unifies generic BB codes under toric layouts to explore their universal properties.
	It offered systematic procedures to analyze the anyon structures and logical qubits, and revealed that BB codes are generically frustrated.
	We demonstrated the SET origin of the topological frustration and extended the scope of logical-anyon correspondence. These novel perspectives and algorithms are elaborated through five theorems in SM~\cite{SM1}.
	
	Our framework and findings rigorously underpin the topological theory of BB codes and provide a solid foundation for further exploring their fault tolerance.
	For instance, we can then apply anyon condensation theory to design fault-tolerant logical operations of BB codes, akin to the surface code and color code~\cite{Bombin09,Krishna21,Kesselring24}.
	Moreover, it enables the application of statistical-mechanical approaches and duality techniques to investigate their fault-tolerance thresholds~\cite{Dennis02,Wang03,Chubb21,Song22,Nishimori07}.
	Aside from BB codes, the topological framework and the algebraic tools developed here may extend to other implementation oriented QLDPC codes like hypergraph product (HGP) codes~\cite{Tillich14}, generalized bicycle (GB) codes~\cite{Kovalev13}, and more general product codes~\cite{Breuckmann21a,Rakovszky24}.

	\begin{acknowledgments}
		\emph{Acknowledgments}---This work is supported by the National Natural Science Foundation of China (Grants No.~12474145, No.~12522502, No.~12447101, and No.~12474491), the New Cornerstone Science Foundation through the XPLORER PRIZE, Anhui Initiative in Quantum Information Technologies, Shanghai Municipal Science and Technology Major Project (Grant No.~2019SHZDZX01), the Fundamental Research Funds for the Central Universities, Peking University.
	\end{acknowledgments}

	\bibliography{qec}

\begin{thebibliography}{100}%
\makeatletter
\providecommand \@ifxundefined [1]{%
 \@ifx{#1\undefined}
}%
\providecommand \@ifnum [1]{%
 \ifnum #1\expandafter \@firstoftwo
 \else \expandafter \@secondoftwo
 \fi
}%
\providecommand \@ifx [1]{%
 \ifx #1\expandafter \@firstoftwo
 \else \expandafter \@secondoftwo
 \fi
}%
\providecommand \natexlab [1]{#1}%
\providecommand \enquote  [1]{``#1''}%
\providecommand \bibnamefont  [1]{#1}%
\providecommand \bibfnamefont [1]{#1}%
\providecommand \citenamefont [1]{#1}%
\providecommand \href@noop [0]{\@secondoftwo}%
\providecommand \href [0]{\begingroup \@sanitize@url \@href}%
\providecommand \@href[1]{\@@startlink{#1}\@@href}%
\providecommand \@@href[1]{\endgroup#1\@@endlink}%
\providecommand \@sanitize@url [0]{\catcode `\\12\catcode `\$12\catcode
  `\&12\catcode `\#12\catcode `\^12\catcode `\_12\catcode `\%12\relax}%
\providecommand \@@startlink[1]{}%
\providecommand \@@endlink[0]{}%
\providecommand \url  [0]{\begingroup\@sanitize@url \@url }%
\providecommand \@url [1]{\endgroup\@href {#1}{\urlprefix }}%
\providecommand \urlprefix  [0]{URL }%
\providecommand \Eprint [0]{\href }%
\providecommand \doibase [0]{https://doi.org/}%
\providecommand \selectlanguage [0]{\@gobble}%
\providecommand \bibinfo  [0]{\@secondoftwo}%
\providecommand \bibfield  [0]{\@secondoftwo}%
\providecommand \translation [1]{[#1]}%
\providecommand \BibitemOpen [0]{}%
\providecommand \bibitemStop [0]{}%
\providecommand \bibitemNoStop [0]{.\EOS\space}%
\providecommand \EOS [0]{\spacefactor3000\relax}%
\providecommand \BibitemShut  [1]{\csname bibitem#1\endcsname}%
\let\auto@bib@innerbib\@empty
%</preamble>
\bibitem [{\citenamefont {Shor}(1995)}]{Shor95}%
  \BibitemOpen
  \bibfield  {author} {\bibinfo {author} {\bibfnamefont {P.~W.}\ \bibnamefont
  {Shor}},\ }\bibfield  {title} {\bibinfo {title} {Scheme for reducing
  decoherence in quantum computer memory},\ }\href
  {https://doi.org/10.1103/PhysRevA.52.R2493} {\bibfield  {journal} {\bibinfo
  {journal} {Phys. Rev. A}\ }\textbf {\bibinfo {volume} {52}},\ \bibinfo
  {pages} {R2493} (\bibinfo {year} {1995})}\BibitemShut {NoStop}%
\bibitem [{\citenamefont {Steane}(1996)}]{Steane96}%
  \BibitemOpen
  \bibfield  {author} {\bibinfo {author} {\bibfnamefont {A.~M.}\ \bibnamefont
  {Steane}},\ }\bibfield  {title} {\bibinfo {title} {Error correcting codes in
  quantum theory},\ }\href {https://doi.org/10.1103/PhysRevLett.77.793}
  {\bibfield  {journal} {\bibinfo  {journal} {Phys. Rev. Lett.}\ }\textbf
  {\bibinfo {volume} {77}},\ \bibinfo {pages} {793} (\bibinfo {year}
  {1996})}\BibitemShut {NoStop}%
\bibitem [{\citenamefont {Knill}\ and\ \citenamefont
  {Laflamme}(1997)}]{Knill97}%
  \BibitemOpen
  \bibfield  {author} {\bibinfo {author} {\bibfnamefont {E.}~\bibnamefont
  {Knill}}\ and\ \bibinfo {author} {\bibfnamefont {R.}~\bibnamefont
  {Laflamme}},\ }\bibfield  {title} {\bibinfo {title} {Theory of quantum
  error-correcting codes},\ }\href {https://doi.org/10.1103/PhysRevA.55.900}
  {\bibfield  {journal} {\bibinfo  {journal} {Phys. Rev. A}\ }\textbf {\bibinfo
  {volume} {55}},\ \bibinfo {pages} {900} (\bibinfo {year} {1997})}\BibitemShut
  {NoStop}%
\bibitem [{\citenamefont {Kitaev}(2003)}]{Kitaev03}%
  \BibitemOpen
  \bibfield  {author} {\bibinfo {author} {\bibfnamefont {A.}~\bibnamefont
  {Kitaev}},\ }\bibfield  {title} {\bibinfo {title} {Fault-tolerant quantum
  computation by anyons},\ }\href
  {https://doi.org/https://doi.org/10.1016/S0003-4916(02)00018-0} {\bibfield
  {journal} {\bibinfo  {journal} {Annals of Physics}\ }\textbf {\bibinfo
  {volume} {303}},\ \bibinfo {pages} {2 } (\bibinfo {year} {2003})}\BibitemShut
  {NoStop}%
\bibitem [{\citenamefont {Gottesman}(1997)}]{Gottesman97}%
  \BibitemOpen
  \bibfield  {author} {\bibinfo {author} {\bibfnamefont {D.}~\bibnamefont
  {Gottesman}},\ }\href@noop {} {\emph {\bibinfo {title} {Stabilizer codes and
  quantum error correction}}}\ (\bibinfo  {publisher} {California Institute of
  Technology},\ \bibinfo {year} {1997})\BibitemShut {NoStop}%
\bibitem [{\citenamefont {Dennis}\ \emph {et~al.}(2002)\citenamefont {Dennis},
  \citenamefont {Kitaev}, \citenamefont {Landahl},\ and\ \citenamefont
  {Preskill}}]{Dennis02}%
  \BibitemOpen
  \bibfield  {author} {\bibinfo {author} {\bibfnamefont {E.}~\bibnamefont
  {Dennis}}, \bibinfo {author} {\bibfnamefont {A.}~\bibnamefont {Kitaev}},
  \bibinfo {author} {\bibfnamefont {A.}~\bibnamefont {Landahl}},\ and\ \bibinfo
  {author} {\bibfnamefont {J.}~\bibnamefont {Preskill}},\ }\bibfield  {title}
  {\bibinfo {title} {Topological quantum memory},\ }\href
  {https://doi.org/10.1063/1.1499754} {\bibfield  {journal} {\bibinfo
  {journal} {Journal of Mathematical Physics}\ }\textbf {\bibinfo {volume}
  {43}},\ \bibinfo {pages} {4452} (\bibinfo {year} {2002})}\BibitemShut
  {NoStop}%
\bibitem [{\citenamefont {Terhal}(2015)}]{Terhal15}%
  \BibitemOpen
  \bibfield  {author} {\bibinfo {author} {\bibfnamefont {B.~M.}\ \bibnamefont
  {Terhal}},\ }\bibfield  {title} {\bibinfo {title} {Quantum error correction
  for quantum memories},\ }\href {https://doi.org/10.1103/RevModPhys.87.307}
  {\bibfield  {journal} {\bibinfo  {journal} {Rev. Mod. Phys.}\ }\textbf
  {\bibinfo {volume} {87}},\ \bibinfo {pages} {307} (\bibinfo {year}
  {2015})}\BibitemShut {NoStop}%
\bibitem [{\citenamefont {Campbell}\ \emph {et~al.}(2017)\citenamefont
  {Campbell}, \citenamefont {Terhal},\ and\ \citenamefont
  {Vuillot}}]{Campbell17}%
  \BibitemOpen
  \bibfield  {author} {\bibinfo {author} {\bibfnamefont {E.~T.}\ \bibnamefont
  {Campbell}}, \bibinfo {author} {\bibfnamefont {B.~M.}\ \bibnamefont
  {Terhal}},\ and\ \bibinfo {author} {\bibfnamefont {C.}~\bibnamefont
  {Vuillot}},\ }\bibfield  {title} {\bibinfo {title} {Roads towards
  fault-tolerant universal quantum computation},\ }\href
  {https://doi.org/10.1038/nature23460} {\bibfield  {journal} {\bibinfo
  {journal} {Nature}\ }\textbf {\bibinfo {volume} {549}},\ \bibinfo {pages}
  {172} (\bibinfo {year} {2017})}\BibitemShut {NoStop}%
\bibitem [{\citenamefont {Fowler}\ \emph {et~al.}(2012)\citenamefont {Fowler},
  \citenamefont {Mariantoni}, \citenamefont {Martinis},\ and\ \citenamefont
  {Cleland}}]{Fowler12}%
  \BibitemOpen
  \bibfield  {author} {\bibinfo {author} {\bibfnamefont {A.~G.}\ \bibnamefont
  {Fowler}}, \bibinfo {author} {\bibfnamefont {M.}~\bibnamefont {Mariantoni}},
  \bibinfo {author} {\bibfnamefont {J.~M.}\ \bibnamefont {Martinis}},\ and\
  \bibinfo {author} {\bibfnamefont {A.~N.}\ \bibnamefont {Cleland}},\
  }\bibfield  {title} {\bibinfo {title} {Surface codes: Towards practical
  large-scale quantum computation},\ }\href
  {https://doi.org/10.1103/PhysRevA.86.032324} {\bibfield  {journal} {\bibinfo
  {journal} {Phys. Rev. A}\ }\textbf {\bibinfo {volume} {86}},\ \bibinfo
  {pages} {032324} (\bibinfo {year} {2012})}\BibitemShut {NoStop}%
\bibitem [{\citenamefont {Bravyi}\ and\ \citenamefont
  {Kitaev}(1998)}]{Bravyi98}%
  \BibitemOpen
  \bibfield  {author} {\bibinfo {author} {\bibfnamefont {S.~B.}\ \bibnamefont
  {Bravyi}}\ and\ \bibinfo {author} {\bibfnamefont {A.~Y.}\ \bibnamefont
  {Kitaev}},\ }\bibfield  {title} {\bibinfo {title} {Quantum codes on a lattice
  with boundary},\ }\href@noop {} {\bibfield  {journal} {\bibinfo  {journal}
  {arXiv preprint quant-ph/9811052}\ } (\bibinfo {year} {1998})}\BibitemShut
  {NoStop}%
\bibitem [{\citenamefont {Bombin}\ and\ \citenamefont
  {Martin-Delgado}(2006)}]{Bombin06}%
  \BibitemOpen
  \bibfield  {author} {\bibinfo {author} {\bibfnamefont {H.}~\bibnamefont
  {Bombin}}\ and\ \bibinfo {author} {\bibfnamefont {M.~A.}\ \bibnamefont
  {Martin-Delgado}},\ }\bibfield  {title} {\bibinfo {title} {Topological
  quantum distillation},\ }\href
  {https://doi.org/10.1103/PhysRevLett.97.180501} {\bibfield  {journal}
  {\bibinfo  {journal} {Phys. Rev. Lett.}\ }\textbf {\bibinfo {volume} {97}},\
  \bibinfo {pages} {180501} (\bibinfo {year} {2006})}\BibitemShut {NoStop}%
\bibitem [{\citenamefont {Bombin}\ and\ \citenamefont
  {Martin-Delgado}(2007)}]{Bombin07}%
  \BibitemOpen
  \bibfield  {author} {\bibinfo {author} {\bibfnamefont {H.}~\bibnamefont
  {Bombin}}\ and\ \bibinfo {author} {\bibfnamefont {M.~A.}\ \bibnamefont
  {Martin-Delgado}},\ }\bibfield  {title} {\bibinfo {title} {Topological
  computation without braiding},\ }\href
  {https://doi.org/10.1103/PhysRevLett.98.160502} {\bibfield  {journal}
  {\bibinfo  {journal} {Phys. Rev. Lett.}\ }\textbf {\bibinfo {volume} {98}},\
  \bibinfo {pages} {160502} (\bibinfo {year} {2007})}\BibitemShut {NoStop}%
\bibitem [{\citenamefont {Satzinger}\ \emph {et~al.}(2021)\citenamefont
  {Satzinger}, \citenamefont {Liu}, \citenamefont {Smith}, \citenamefont
  {Knapp}, \citenamefont {Newman}, \citenamefont {Jones}, \citenamefont {Chen},
  \citenamefont {Quintana}, \citenamefont {Mi}, \citenamefont {Dunsworth},
  \citenamefont {Gidney}, \citenamefont {Aleiner}, \citenamefont {Arute},
  \citenamefont {Arya}, \citenamefont {Atalaya}, \citenamefont {Babbush},
  \citenamefont {Bardin}, \citenamefont {Barends}, \citenamefont {Basso},
  \citenamefont {Bengtsson}, \citenamefont {Bilmes}, \citenamefont {Broughton},
  \citenamefont {Buckley}, \citenamefont {Buell}, \citenamefont {Burkett},
  \citenamefont {Bushnell}, \citenamefont {Chiaro}, \citenamefont {Collins},
  \citenamefont {Courtney}, \citenamefont {Demura}, \citenamefont {Derk},
  \citenamefont {Eppens}, \citenamefont {Erickson}, \citenamefont {Faoro},
  \citenamefont {Farhi}, \citenamefont {Fowler}, \citenamefont {Foxen},
  \citenamefont {Giustina}, \citenamefont {Greene}, \citenamefont {Gross},
  \citenamefont {Harrigan}, \citenamefont {Harrington}, \citenamefont {Hilton},
  \citenamefont {Hong}, \citenamefont {Huang}, \citenamefont {Huggins},
  \citenamefont {Ioffe}, \citenamefont {Isakov}, \citenamefont {Jeffrey},
  \citenamefont {Jiang}, \citenamefont {Kafri}, \citenamefont {Kechedzhi},
  \citenamefont {Khattar}, \citenamefont {Kim}, \citenamefont {Klimov},
  \citenamefont {Korotkov}, \citenamefont {Kostritsa}, \citenamefont
  {Landhuis}, \citenamefont {Laptev}, \citenamefont {Locharla}, \citenamefont
  {Lucero}, \citenamefont {Martin}, \citenamefont {McClean}, \citenamefont
  {McEwen}, \citenamefont {Miao}, \citenamefont {Mohseni}, \citenamefont
  {Montazeri}, \citenamefont {Mruczkiewicz}, \citenamefont {Mutus},
  \citenamefont {Naaman}, \citenamefont {Neeley}, \citenamefont {Neill},
  \citenamefont {Niu}, \citenamefont {O'Brien}, \citenamefont {Opremcak},
  \citenamefont {Pat{\'o}}, \citenamefont {Petukhov}, \citenamefont {Rubin},
  \citenamefont {Sank}, \citenamefont {Shvarts}, \citenamefont {Strain},
  \citenamefont {Szalay}, \citenamefont {Villalonga}, \citenamefont {White},
  \citenamefont {Yao}, \citenamefont {Yeh}, \citenamefont {Yoo}, \citenamefont
  {Zalcman}, \citenamefont {Neven}, \citenamefont {Boixo}, \citenamefont
  {Megrant}, \citenamefont {Chen}, \citenamefont {Kelly}, \citenamefont
  {Smelyanskiy}, \citenamefont {Kitaev}, \citenamefont {Knap}, \citenamefont
  {Pollmann},\ and\ \citenamefont {Roushan}}]{Satzinger21}%
  \BibitemOpen
  \bibfield  {author} {\bibinfo {author} {\bibfnamefont {K.~J.}\ \bibnamefont
  {Satzinger}}, \bibinfo {author} {\bibfnamefont {Y.-J.}\ \bibnamefont {Liu}},
  \bibinfo {author} {\bibfnamefont {A.}~\bibnamefont {Smith}}, \bibinfo
  {author} {\bibfnamefont {C.}~\bibnamefont {Knapp}}, \bibinfo {author}
  {\bibfnamefont {M.}~\bibnamefont {Newman}}, \bibinfo {author} {\bibfnamefont
  {C.}~\bibnamefont {Jones}}, \bibinfo {author} {\bibfnamefont
  {Z.}~\bibnamefont {Chen}}, \bibinfo {author} {\bibfnamefont {C.}~\bibnamefont
  {Quintana}}, \bibinfo {author} {\bibfnamefont {X.}~\bibnamefont {Mi}},
  \bibinfo {author} {\bibfnamefont {A.}~\bibnamefont {Dunsworth}}, \bibinfo
  {author} {\bibfnamefont {C.}~\bibnamefont {Gidney}}, \bibinfo {author}
  {\bibfnamefont {I.}~\bibnamefont {Aleiner}}, \bibinfo {author} {\bibfnamefont
  {F.}~\bibnamefont {Arute}}, \bibinfo {author} {\bibfnamefont
  {K.}~\bibnamefont {Arya}}, \bibinfo {author} {\bibfnamefont {J.}~\bibnamefont
  {Atalaya}}, \bibinfo {author} {\bibfnamefont {R.}~\bibnamefont {Babbush}},
  \bibinfo {author} {\bibfnamefont {J.~C.}\ \bibnamefont {Bardin}}, \bibinfo
  {author} {\bibfnamefont {R.}~\bibnamefont {Barends}}, \bibinfo {author}
  {\bibfnamefont {J.}~\bibnamefont {Basso}}, \bibinfo {author} {\bibfnamefont
  {A.}~\bibnamefont {Bengtsson}}, \bibinfo {author} {\bibfnamefont
  {A.}~\bibnamefont {Bilmes}}, \bibinfo {author} {\bibfnamefont
  {M.}~\bibnamefont {Broughton}}, \bibinfo {author} {\bibfnamefont {B.~B.}\
  \bibnamefont {Buckley}}, \bibinfo {author} {\bibfnamefont {D.~A.}\
  \bibnamefont {Buell}}, \bibinfo {author} {\bibfnamefont {B.}~\bibnamefont
  {Burkett}}, \bibinfo {author} {\bibfnamefont {N.}~\bibnamefont {Bushnell}},
  \bibinfo {author} {\bibfnamefont {B.}~\bibnamefont {Chiaro}}, \bibinfo
  {author} {\bibfnamefont {R.}~\bibnamefont {Collins}}, \bibinfo {author}
  {\bibfnamefont {W.}~\bibnamefont {Courtney}}, \bibinfo {author}
  {\bibfnamefont {S.}~\bibnamefont {Demura}}, \bibinfo {author} {\bibfnamefont
  {A.~R.}\ \bibnamefont {Derk}}, \bibinfo {author} {\bibfnamefont
  {D.}~\bibnamefont {Eppens}}, \bibinfo {author} {\bibfnamefont
  {C.}~\bibnamefont {Erickson}}, \bibinfo {author} {\bibfnamefont
  {L.}~\bibnamefont {Faoro}}, \bibinfo {author} {\bibfnamefont
  {E.}~\bibnamefont {Farhi}}, \bibinfo {author} {\bibfnamefont {A.~G.}\
  \bibnamefont {Fowler}}, \bibinfo {author} {\bibfnamefont {B.}~\bibnamefont
  {Foxen}}, \bibinfo {author} {\bibfnamefont {M.}~\bibnamefont {Giustina}},
  \bibinfo {author} {\bibfnamefont {A.}~\bibnamefont {Greene}}, \bibinfo
  {author} {\bibfnamefont {J.~A.}\ \bibnamefont {Gross}}, \bibinfo {author}
  {\bibfnamefont {M.~P.}\ \bibnamefont {Harrigan}}, \bibinfo {author}
  {\bibfnamefont {S.~D.}\ \bibnamefont {Harrington}}, \bibinfo {author}
  {\bibfnamefont {J.}~\bibnamefont {Hilton}}, \bibinfo {author} {\bibfnamefont
  {S.}~\bibnamefont {Hong}}, \bibinfo {author} {\bibfnamefont {T.}~\bibnamefont
  {Huang}}, \bibinfo {author} {\bibfnamefont {W.~J.}\ \bibnamefont {Huggins}},
  \bibinfo {author} {\bibfnamefont {L.~B.}\ \bibnamefont {Ioffe}}, \bibinfo
  {author} {\bibfnamefont {S.~V.}\ \bibnamefont {Isakov}}, \bibinfo {author}
  {\bibfnamefont {E.}~\bibnamefont {Jeffrey}}, \bibinfo {author} {\bibfnamefont
  {Z.}~\bibnamefont {Jiang}}, \bibinfo {author} {\bibfnamefont
  {D.}~\bibnamefont {Kafri}}, \bibinfo {author} {\bibfnamefont
  {K.}~\bibnamefont {Kechedzhi}}, \bibinfo {author} {\bibfnamefont
  {T.}~\bibnamefont {Khattar}}, \bibinfo {author} {\bibfnamefont
  {S.}~\bibnamefont {Kim}}, \bibinfo {author} {\bibfnamefont {P.~V.}\
  \bibnamefont {Klimov}}, \bibinfo {author} {\bibfnamefont {A.~N.}\
  \bibnamefont {Korotkov}}, \bibinfo {author} {\bibfnamefont {F.}~\bibnamefont
  {Kostritsa}}, \bibinfo {author} {\bibfnamefont {D.}~\bibnamefont {Landhuis}},
  \bibinfo {author} {\bibfnamefont {P.}~\bibnamefont {Laptev}}, \bibinfo
  {author} {\bibfnamefont {A.}~\bibnamefont {Locharla}}, \bibinfo {author}
  {\bibfnamefont {E.}~\bibnamefont {Lucero}}, \bibinfo {author} {\bibfnamefont
  {O.}~\bibnamefont {Martin}}, \bibinfo {author} {\bibfnamefont {J.~R.}\
  \bibnamefont {McClean}}, \bibinfo {author} {\bibfnamefont {M.}~\bibnamefont
  {McEwen}}, \bibinfo {author} {\bibfnamefont {K.~C.}\ \bibnamefont {Miao}},
  \bibinfo {author} {\bibfnamefont {M.}~\bibnamefont {Mohseni}}, \bibinfo
  {author} {\bibfnamefont {S.}~\bibnamefont {Montazeri}}, \bibinfo {author}
  {\bibfnamefont {W.}~\bibnamefont {Mruczkiewicz}}, \bibinfo {author}
  {\bibfnamefont {J.}~\bibnamefont {Mutus}}, \bibinfo {author} {\bibfnamefont
  {O.}~\bibnamefont {Naaman}}, \bibinfo {author} {\bibfnamefont
  {M.}~\bibnamefont {Neeley}}, \bibinfo {author} {\bibfnamefont
  {C.}~\bibnamefont {Neill}}, \bibinfo {author} {\bibfnamefont {M.~Y.}\
  \bibnamefont {Niu}}, \bibinfo {author} {\bibfnamefont {T.~E.}\ \bibnamefont
  {O'Brien}}, \bibinfo {author} {\bibfnamefont {A.}~\bibnamefont {Opremcak}},
  \bibinfo {author} {\bibfnamefont {B.}~\bibnamefont {Pat{\'o}}}, \bibinfo
  {author} {\bibfnamefont {A.}~\bibnamefont {Petukhov}}, \bibinfo {author}
  {\bibfnamefont {N.~C.}\ \bibnamefont {Rubin}}, \bibinfo {author}
  {\bibfnamefont {D.}~\bibnamefont {Sank}}, \bibinfo {author} {\bibfnamefont
  {V.}~\bibnamefont {Shvarts}}, \bibinfo {author} {\bibfnamefont
  {D.}~\bibnamefont {Strain}}, \bibinfo {author} {\bibfnamefont
  {M.}~\bibnamefont {Szalay}}, \bibinfo {author} {\bibfnamefont
  {B.}~\bibnamefont {Villalonga}}, \bibinfo {author} {\bibfnamefont {T.~C.}\
  \bibnamefont {White}}, \bibinfo {author} {\bibfnamefont {Z.}~\bibnamefont
  {Yao}}, \bibinfo {author} {\bibfnamefont {P.}~\bibnamefont {Yeh}}, \bibinfo
  {author} {\bibfnamefont {J.}~\bibnamefont {Yoo}}, \bibinfo {author}
  {\bibfnamefont {A.}~\bibnamefont {Zalcman}}, \bibinfo {author} {\bibfnamefont
  {H.}~\bibnamefont {Neven}}, \bibinfo {author} {\bibfnamefont
  {S.}~\bibnamefont {Boixo}}, \bibinfo {author} {\bibfnamefont
  {A.}~\bibnamefont {Megrant}}, \bibinfo {author} {\bibfnamefont
  {Y.}~\bibnamefont {Chen}}, \bibinfo {author} {\bibfnamefont {J.}~\bibnamefont
  {Kelly}}, \bibinfo {author} {\bibfnamefont {V.}~\bibnamefont {Smelyanskiy}},
  \bibinfo {author} {\bibfnamefont {A.}~\bibnamefont {Kitaev}}, \bibinfo
  {author} {\bibfnamefont {M.}~\bibnamefont {Knap}}, \bibinfo {author}
  {\bibfnamefont {F.}~\bibnamefont {Pollmann}},\ and\ \bibinfo {author}
  {\bibfnamefont {P.}~\bibnamefont {Roushan}},\ }\bibfield  {title} {\bibinfo
  {title} {Realizing topologically ordered states on a quantum processor},\
  }\href {https://doi.org/10.1126/science.abi8378} {\bibfield  {journal}
  {\bibinfo  {journal} {Science}\ }\textbf {\bibinfo {volume} {374}},\ \bibinfo
  {pages} {1237} (\bibinfo {year} {2021})}\BibitemShut {NoStop}%
\bibitem [{\citenamefont {Krinner}\ \emph {et~al.}(2022)\citenamefont
  {Krinner}, \citenamefont {Lacroix}, \citenamefont {Remm}, \citenamefont
  {Di~Paolo}, \citenamefont {Genois}, \citenamefont {Leroux}, \citenamefont
  {Hellings}, \citenamefont {Lazar}, \citenamefont {Swiadek}, \citenamefont
  {Herrmann}, \citenamefont {Norris}, \citenamefont {Andersen}, \citenamefont
  {M{\"u}ller}, \citenamefont {Blais}, \citenamefont {Eichler},\ and\
  \citenamefont {Wallraff}}]{Krinner22}%
  \BibitemOpen
  \bibfield  {author} {\bibinfo {author} {\bibfnamefont {S.}~\bibnamefont
  {Krinner}}, \bibinfo {author} {\bibfnamefont {N.}~\bibnamefont {Lacroix}},
  \bibinfo {author} {\bibfnamefont {A.}~\bibnamefont {Remm}}, \bibinfo {author}
  {\bibfnamefont {A.}~\bibnamefont {Di~Paolo}}, \bibinfo {author}
  {\bibfnamefont {E.}~\bibnamefont {Genois}}, \bibinfo {author} {\bibfnamefont
  {C.}~\bibnamefont {Leroux}}, \bibinfo {author} {\bibfnamefont
  {C.}~\bibnamefont {Hellings}}, \bibinfo {author} {\bibfnamefont
  {S.}~\bibnamefont {Lazar}}, \bibinfo {author} {\bibfnamefont
  {F.}~\bibnamefont {Swiadek}}, \bibinfo {author} {\bibfnamefont
  {J.}~\bibnamefont {Herrmann}}, \bibinfo {author} {\bibfnamefont {G.~J.}\
  \bibnamefont {Norris}}, \bibinfo {author} {\bibfnamefont {C.~K.}\
  \bibnamefont {Andersen}}, \bibinfo {author} {\bibfnamefont {M.}~\bibnamefont
  {M{\"u}ller}}, \bibinfo {author} {\bibfnamefont {A.}~\bibnamefont {Blais}},
  \bibinfo {author} {\bibfnamefont {C.}~\bibnamefont {Eichler}},\ and\ \bibinfo
  {author} {\bibfnamefont {A.}~\bibnamefont {Wallraff}},\ }\bibfield  {title}
  {\bibinfo {title} {Realizing repeated quantum error correction in a
  distance-three surface code},\ }\href
  {https://doi.org/10.1038/s41586-022-04566-8} {\bibfield  {journal} {\bibinfo
  {journal} {Nature}\ }\textbf {\bibinfo {volume} {605}},\ \bibinfo {pages}
  {669} (\bibinfo {year} {2022})}\BibitemShut {NoStop}%
\bibitem [{\citenamefont {Zhao}\ \emph {et~al.}(2022)\citenamefont {Zhao},
  \citenamefont {Ye}, \citenamefont {Huang}, \citenamefont {Zhang},
  \citenamefont {Wu}, \citenamefont {Guan}, \citenamefont {Zhu}, \citenamefont
  {Wei}, \citenamefont {He}, \citenamefont {Cao}, \citenamefont {Chen},
  \citenamefont {Chung}, \citenamefont {Deng}, \citenamefont {Fan},
  \citenamefont {Gong}, \citenamefont {Guo}, \citenamefont {Guo}, \citenamefont
  {Han}, \citenamefont {Li}, \citenamefont {Li}, \citenamefont {Li},
  \citenamefont {Liang}, \citenamefont {Lin}, \citenamefont {Qian},
  \citenamefont {Rong}, \citenamefont {Su}, \citenamefont {Sun}, \citenamefont
  {Wang}, \citenamefont {Wu}, \citenamefont {Xu}, \citenamefont {Ying},
  \citenamefont {Yu}, \citenamefont {Zha}, \citenamefont {Zhang}, \citenamefont
  {Huo}, \citenamefont {Lu}, \citenamefont {Peng}, \citenamefont {Zhu},\ and\
  \citenamefont {Pan}}]{USTC22}%
  \BibitemOpen
  \bibfield  {author} {\bibinfo {author} {\bibfnamefont {Y.}~\bibnamefont
  {Zhao}}, \bibinfo {author} {\bibfnamefont {Y.}~\bibnamefont {Ye}}, \bibinfo
  {author} {\bibfnamefont {H.-L.}\ \bibnamefont {Huang}}, \bibinfo {author}
  {\bibfnamefont {Y.}~\bibnamefont {Zhang}}, \bibinfo {author} {\bibfnamefont
  {D.}~\bibnamefont {Wu}}, \bibinfo {author} {\bibfnamefont {H.}~\bibnamefont
  {Guan}}, \bibinfo {author} {\bibfnamefont {Q.}~\bibnamefont {Zhu}}, \bibinfo
  {author} {\bibfnamefont {Z.}~\bibnamefont {Wei}}, \bibinfo {author}
  {\bibfnamefont {T.}~\bibnamefont {He}}, \bibinfo {author} {\bibfnamefont
  {S.}~\bibnamefont {Cao}}, \bibinfo {author} {\bibfnamefont {F.}~\bibnamefont
  {Chen}}, \bibinfo {author} {\bibfnamefont {T.-H.}\ \bibnamefont {Chung}},
  \bibinfo {author} {\bibfnamefont {H.}~\bibnamefont {Deng}}, \bibinfo {author}
  {\bibfnamefont {D.}~\bibnamefont {Fan}}, \bibinfo {author} {\bibfnamefont
  {M.}~\bibnamefont {Gong}}, \bibinfo {author} {\bibfnamefont {C.}~\bibnamefont
  {Guo}}, \bibinfo {author} {\bibfnamefont {S.}~\bibnamefont {Guo}}, \bibinfo
  {author} {\bibfnamefont {L.}~\bibnamefont {Han}}, \bibinfo {author}
  {\bibfnamefont {N.}~\bibnamefont {Li}}, \bibinfo {author} {\bibfnamefont
  {S.}~\bibnamefont {Li}}, \bibinfo {author} {\bibfnamefont {Y.}~\bibnamefont
  {Li}}, \bibinfo {author} {\bibfnamefont {F.}~\bibnamefont {Liang}}, \bibinfo
  {author} {\bibfnamefont {J.}~\bibnamefont {Lin}}, \bibinfo {author}
  {\bibfnamefont {H.}~\bibnamefont {Qian}}, \bibinfo {author} {\bibfnamefont
  {H.}~\bibnamefont {Rong}}, \bibinfo {author} {\bibfnamefont {H.}~\bibnamefont
  {Su}}, \bibinfo {author} {\bibfnamefont {L.}~\bibnamefont {Sun}}, \bibinfo
  {author} {\bibfnamefont {S.}~\bibnamefont {Wang}}, \bibinfo {author}
  {\bibfnamefont {Y.}~\bibnamefont {Wu}}, \bibinfo {author} {\bibfnamefont
  {Y.}~\bibnamefont {Xu}}, \bibinfo {author} {\bibfnamefont {C.}~\bibnamefont
  {Ying}}, \bibinfo {author} {\bibfnamefont {J.}~\bibnamefont {Yu}}, \bibinfo
  {author} {\bibfnamefont {C.}~\bibnamefont {Zha}}, \bibinfo {author}
  {\bibfnamefont {K.}~\bibnamefont {Zhang}}, \bibinfo {author} {\bibfnamefont
  {Y.-H.}\ \bibnamefont {Huo}}, \bibinfo {author} {\bibfnamefont {C.-Y.}\
  \bibnamefont {Lu}}, \bibinfo {author} {\bibfnamefont {C.-Z.}\ \bibnamefont
  {Peng}}, \bibinfo {author} {\bibfnamefont {X.}~\bibnamefont {Zhu}},\ and\
  \bibinfo {author} {\bibfnamefont {J.-W.}\ \bibnamefont {Pan}},\ }\bibfield
  {title} {\bibinfo {title} {Realization of an error-correcting surface code
  with superconducting qubits},\ }\href
  {https://doi.org/10.1103/PhysRevLett.129.030501} {\bibfield  {journal}
  {\bibinfo  {journal} {Phys. Rev. Lett.}\ }\textbf {\bibinfo {volume} {129}},\
  \bibinfo {pages} {030501} (\bibinfo {year} {2022})}\BibitemShut {NoStop}%
\bibitem [{\citenamefont {Ye}\ \emph {et~al.}(2023)\citenamefont {Ye},
  \citenamefont {He}, \citenamefont {Huang}, \citenamefont {Wei}, \citenamefont
  {Zhang}, \citenamefont {Zhao}, \citenamefont {Wu}, \citenamefont {Zhu},
  \citenamefont {Guan}, \citenamefont {Cao}, \citenamefont {Chen},
  \citenamefont {Chung}, \citenamefont {Deng}, \citenamefont {Fan},
  \citenamefont {Gong}, \citenamefont {Guo}, \citenamefont {Guo}, \citenamefont
  {Han}, \citenamefont {Li}, \citenamefont {Li}, \citenamefont {Li},
  \citenamefont {Liang}, \citenamefont {Lin}, \citenamefont {Qian},
  \citenamefont {Rong}, \citenamefont {Su}, \citenamefont {Wang}, \citenamefont
  {Wu}, \citenamefont {Xu}, \citenamefont {Ying}, \citenamefont {Yu},
  \citenamefont {Zha}, \citenamefont {Zhang}, \citenamefont {Huo},
  \citenamefont {Lu}, \citenamefont {Peng}, \citenamefont {Zhu},\ and\
  \citenamefont {Pan}}]{USTC23}%
  \BibitemOpen
  \bibfield  {author} {\bibinfo {author} {\bibfnamefont {Y.}~\bibnamefont
  {Ye}}, \bibinfo {author} {\bibfnamefont {T.}~\bibnamefont {He}}, \bibinfo
  {author} {\bibfnamefont {H.-L.}\ \bibnamefont {Huang}}, \bibinfo {author}
  {\bibfnamefont {Z.}~\bibnamefont {Wei}}, \bibinfo {author} {\bibfnamefont
  {Y.}~\bibnamefont {Zhang}}, \bibinfo {author} {\bibfnamefont
  {Y.}~\bibnamefont {Zhao}}, \bibinfo {author} {\bibfnamefont {D.}~\bibnamefont
  {Wu}}, \bibinfo {author} {\bibfnamefont {Q.}~\bibnamefont {Zhu}}, \bibinfo
  {author} {\bibfnamefont {H.}~\bibnamefont {Guan}}, \bibinfo {author}
  {\bibfnamefont {S.}~\bibnamefont {Cao}}, \bibinfo {author} {\bibfnamefont
  {F.}~\bibnamefont {Chen}}, \bibinfo {author} {\bibfnamefont {T.-H.}\
  \bibnamefont {Chung}}, \bibinfo {author} {\bibfnamefont {H.}~\bibnamefont
  {Deng}}, \bibinfo {author} {\bibfnamefont {D.}~\bibnamefont {Fan}}, \bibinfo
  {author} {\bibfnamefont {M.}~\bibnamefont {Gong}}, \bibinfo {author}
  {\bibfnamefont {C.}~\bibnamefont {Guo}}, \bibinfo {author} {\bibfnamefont
  {S.}~\bibnamefont {Guo}}, \bibinfo {author} {\bibfnamefont {L.}~\bibnamefont
  {Han}}, \bibinfo {author} {\bibfnamefont {N.}~\bibnamefont {Li}}, \bibinfo
  {author} {\bibfnamefont {S.}~\bibnamefont {Li}}, \bibinfo {author}
  {\bibfnamefont {Y.}~\bibnamefont {Li}}, \bibinfo {author} {\bibfnamefont
  {F.}~\bibnamefont {Liang}}, \bibinfo {author} {\bibfnamefont
  {J.}~\bibnamefont {Lin}}, \bibinfo {author} {\bibfnamefont {H.}~\bibnamefont
  {Qian}}, \bibinfo {author} {\bibfnamefont {H.}~\bibnamefont {Rong}}, \bibinfo
  {author} {\bibfnamefont {H.}~\bibnamefont {Su}}, \bibinfo {author}
  {\bibfnamefont {S.}~\bibnamefont {Wang}}, \bibinfo {author} {\bibfnamefont
  {Y.}~\bibnamefont {Wu}}, \bibinfo {author} {\bibfnamefont {Y.}~\bibnamefont
  {Xu}}, \bibinfo {author} {\bibfnamefont {C.}~\bibnamefont {Ying}}, \bibinfo
  {author} {\bibfnamefont {J.}~\bibnamefont {Yu}}, \bibinfo {author}
  {\bibfnamefont {C.}~\bibnamefont {Zha}}, \bibinfo {author} {\bibfnamefont
  {K.}~\bibnamefont {Zhang}}, \bibinfo {author} {\bibfnamefont {Y.-H.}\
  \bibnamefont {Huo}}, \bibinfo {author} {\bibfnamefont {C.-Y.}\ \bibnamefont
  {Lu}}, \bibinfo {author} {\bibfnamefont {C.-Z.}\ \bibnamefont {Peng}},
  \bibinfo {author} {\bibfnamefont {X.}~\bibnamefont {Zhu}},\ and\ \bibinfo
  {author} {\bibfnamefont {J.-W.}\ \bibnamefont {Pan}},\ }\bibfield  {title}
  {\bibinfo {title} {Logical magic state preparation with fidelity beyond the
  distillation threshold on a superconducting quantum processor},\ }\href
  {https://doi.org/10.1103/PhysRevLett.131.210603} {\bibfield  {journal}
  {\bibinfo  {journal} {Phys. Rev. Lett.}\ }\textbf {\bibinfo {volume} {131}},\
  \bibinfo {pages} {210603} (\bibinfo {year} {2023})}\BibitemShut {NoStop}%
\bibitem [{\citenamefont {Acharya}\ \emph {et~al.}(2025)\citenamefont
  {Acharya}, \citenamefont {Abanin}, \citenamefont {Aghababaie-Beni},
  \citenamefont {Aleiner}, \citenamefont {Andersen}, \citenamefont {Ansmann},
  \citenamefont {Arute}, \citenamefont {Arya}, \citenamefont {Asfaw},
  \citenamefont {Astrakhantsev}, \citenamefont {Atalaya}, \citenamefont
  {Babbush}, \citenamefont {Bacon}, \citenamefont {Ballard}, \citenamefont
  {Bardin}, \citenamefont {Bausch}, \citenamefont {Bengtsson}, \citenamefont
  {Bilmes}, \citenamefont {Blackwell}, \citenamefont {Boixo}, \citenamefont
  {Bortoli}, \citenamefont {Bourassa}, \citenamefont {Bovaird}, \citenamefont
  {Brill}, \citenamefont {Broughton}, \citenamefont {Browne}, \citenamefont
  {Buchea}, \citenamefont {Buckley}, \citenamefont {Buell}, \citenamefont
  {Burger}, \citenamefont {Burkett}, \citenamefont {Bushnell}, \citenamefont
  {Cabrera}, \citenamefont {Campero}, \citenamefont {Chang}, \citenamefont
  {Chen}, \citenamefont {Chen}, \citenamefont {Chiaro}, \citenamefont {Chik},
  \citenamefont {Chou}, \citenamefont {Claes}, \citenamefont {Cleland},
  \citenamefont {Cogan}, \citenamefont {Collins}, \citenamefont {Conner},
  \citenamefont {Courtney}, \citenamefont {Crook}, \citenamefont {Curtin},
  \citenamefont {Das}, \citenamefont {Davies}, \citenamefont {De~Lorenzo},
  \citenamefont {Debroy}, \citenamefont {Demura}, \citenamefont {Devoret},
  \citenamefont {Di~Paolo}, \citenamefont {Donohoe}, \citenamefont {Drozdov},
  \citenamefont {Dunsworth}, \citenamefont {Earle}, \citenamefont {Edlich},
  \citenamefont {Eickbusch}, \citenamefont {Elbag}, \citenamefont {Elzouka},
  \citenamefont {Erickson}, \citenamefont {Faoro}, \citenamefont {Farhi},
  \citenamefont {Ferreira}, \citenamefont {Burgos}, \citenamefont {Forati},
  \citenamefont {Fowler}, \citenamefont {Foxen}, \citenamefont {Ganjam},
  \citenamefont {Garcia}, \citenamefont {Gasca}, \citenamefont {Genois},
  \citenamefont {Giang}, \citenamefont {Gidney}, \citenamefont {Gilboa},
  \citenamefont {Gosula}, \citenamefont {Dau}, \citenamefont {Graumann},
  \citenamefont {Greene}, \citenamefont {Gross}, \citenamefont {Habegger},
  \citenamefont {Hall}, \citenamefont {Hamilton}, \citenamefont {Hansen},
  \citenamefont {Harrigan}, \citenamefont {Harrington}, \citenamefont {Heras},
  \citenamefont {Heslin}, \citenamefont {Heu}, \citenamefont {Higgott},
  \citenamefont {Hill}, \citenamefont {Hilton}, \citenamefont {Holland},
  \citenamefont {Hong}, \citenamefont {Huang}, \citenamefont {Huff},
  \citenamefont {Huggins}, \citenamefont {Ioffe}, \citenamefont {Isakov},
  \citenamefont {Iveland}, \citenamefont {Jeffrey}, \citenamefont {Jiang},
  \citenamefont {Jones}, \citenamefont {Jordan}, \citenamefont {Joshi},
  \citenamefont {Juhas}, \citenamefont {Kafri}, \citenamefont {Kang},
  \citenamefont {Karamlou}, \citenamefont {Kechedzhi}, \citenamefont {Kelly},
  \citenamefont {Khaire}, \citenamefont {Khattar}, \citenamefont {Khezri},
  \citenamefont {Kim}, \citenamefont {Klimov}, \citenamefont {Klots},
  \citenamefont {Kobrin}, \citenamefont {Kohli}, \citenamefont {Korotkov},
  \citenamefont {Kostritsa}, \citenamefont {Kothari}, \citenamefont
  {Kozlovskii}, \citenamefont {Kreikebaum}, \citenamefont {Kurilovich},
  \citenamefont {Lacroix}, \citenamefont {Landhuis}, \citenamefont {Lange-Dei},
  \citenamefont {Langley}, \citenamefont {Laptev}, \citenamefont {Lau},
  \citenamefont {Le~Guevel}, \citenamefont {Ledford}, \citenamefont {Lee},
  \citenamefont {Lee}, \citenamefont {Lensky}, \citenamefont {Leon},
  \citenamefont {Lester}, \citenamefont {Li}, \citenamefont {Li}, \citenamefont
  {Lill}, \citenamefont {Liu}, \citenamefont {Livingston}, \citenamefont
  {Locharla}, \citenamefont {Lucero}, \citenamefont {Lundahl}, \citenamefont
  {Lunt}, \citenamefont {Madhuk}, \citenamefont {Malone}, \citenamefont
  {Maloney}, \citenamefont {Mandr{\`a}}, \citenamefont {Manyika}, \citenamefont
  {Martin}, \citenamefont {Martin}, \citenamefont {Martin}, \citenamefont
  {Maxfield}, \citenamefont {McClean}, \citenamefont {McEwen}, \citenamefont
  {Meeks}, \citenamefont {Megrant}, \citenamefont {Mi}, \citenamefont {Miao},
  \citenamefont {Mieszala}, \citenamefont {Molavi}, \citenamefont {Molina},
  \citenamefont {Montazeri}, \citenamefont {Morvan}, \citenamefont {Movassagh},
  \citenamefont {Mruczkiewicz}, \citenamefont {Naaman}, \citenamefont {Neeley},
  \citenamefont {Neill}, \citenamefont {Nersisyan}, \citenamefont {Neven},
  \citenamefont {Newman}, \citenamefont {Ng}, \citenamefont {Nguyen},
  \citenamefont {Nguyen}, \citenamefont {Ni}, \citenamefont {Niu},
  \citenamefont {O'Brien}, \citenamefont {Oliver}, \citenamefont {Opremcak},
  \citenamefont {Ottosson}, \citenamefont {Petukhov}, \citenamefont {Pizzuto},
  \citenamefont {Platt}, \citenamefont {Potter}, \citenamefont {Pritchard},
  \citenamefont {Pryadko}, \citenamefont {Quintana}, \citenamefont
  {Ramachandran}, \citenamefont {Reagor}, \citenamefont {Redding},
  \citenamefont {Rhodes}, \citenamefont {Roberts}, \citenamefont {Rosenberg},
  \citenamefont {Rosenfeld}, \citenamefont {Roushan}, \citenamefont {Rubin},
  \citenamefont {Saei}, \citenamefont {Sank}, \citenamefont {Sankaragomathi},
  \citenamefont {Satzinger}, \citenamefont {Schurkus}, \citenamefont
  {Schuster}, \citenamefont {Senior}, \citenamefont {Shearn}, \citenamefont
  {Shorter}, \citenamefont {Shutty}, \citenamefont {Shvarts}, \citenamefont
  {Singh}, \citenamefont {Sivak}, \citenamefont {Skruzny}, \citenamefont
  {Small}, \citenamefont {Smelyanskiy}, \citenamefont {Smith}, \citenamefont
  {Somma}, \citenamefont {Springer}, \citenamefont {Sterling}, \citenamefont
  {Strain}, \citenamefont {Suchard}, \citenamefont {Szasz}, \citenamefont
  {Sztein}, \citenamefont {Thor}, \citenamefont {Torres}, \citenamefont
  {Torunbalci}, \citenamefont {Vaishnav}, \citenamefont {Vargas}, \citenamefont
  {Vdovichev}, \citenamefont {Vidal}, \citenamefont {Villalonga}, \citenamefont
  {Heidweiller}, \citenamefont {Waltman}, \citenamefont {Wang}, \citenamefont
  {Ware}, \citenamefont {Weber}, \citenamefont {Weidel}, \citenamefont {White},
  \citenamefont {Wong}, \citenamefont {Woo}, \citenamefont {Xing},
  \citenamefont {Yao}, \citenamefont {Yeh}, \citenamefont {Ying}, \citenamefont
  {Yoo}, \citenamefont {Yosri}, \citenamefont {Young}, \citenamefont {Zalcman},
  \citenamefont {Zhang}, \citenamefont {Zhu}, \citenamefont {Zobrist},
  \citenamefont {AI},\ and\ \citenamefont {Collaborators}}]{Google24a}%
  \BibitemOpen
  \bibfield  {author} {\bibinfo {author} {\bibfnamefont {R.}~\bibnamefont
  {Acharya}}, \bibinfo {author} {\bibfnamefont {D.~A.}\ \bibnamefont {Abanin}},
  \bibinfo {author} {\bibfnamefont {L.}~\bibnamefont {Aghababaie-Beni}},
  \bibinfo {author} {\bibfnamefont {I.}~\bibnamefont {Aleiner}}, \bibinfo
  {author} {\bibfnamefont {T.~I.}\ \bibnamefont {Andersen}}, \bibinfo {author}
  {\bibfnamefont {M.}~\bibnamefont {Ansmann}}, \bibinfo {author} {\bibfnamefont
  {F.}~\bibnamefont {Arute}}, \bibinfo {author} {\bibfnamefont
  {K.}~\bibnamefont {Arya}}, \bibinfo {author} {\bibfnamefont {A.}~\bibnamefont
  {Asfaw}}, \bibinfo {author} {\bibfnamefont {N.}~\bibnamefont
  {Astrakhantsev}}, \bibinfo {author} {\bibfnamefont {J.}~\bibnamefont
  {Atalaya}}, \bibinfo {author} {\bibfnamefont {R.}~\bibnamefont {Babbush}},
  \bibinfo {author} {\bibfnamefont {D.}~\bibnamefont {Bacon}}, \bibinfo
  {author} {\bibfnamefont {B.}~\bibnamefont {Ballard}}, \bibinfo {author}
  {\bibfnamefont {J.~C.}\ \bibnamefont {Bardin}}, \bibinfo {author}
  {\bibfnamefont {J.}~\bibnamefont {Bausch}}, \bibinfo {author} {\bibfnamefont
  {A.}~\bibnamefont {Bengtsson}}, \bibinfo {author} {\bibfnamefont
  {A.}~\bibnamefont {Bilmes}}, \bibinfo {author} {\bibfnamefont
  {S.}~\bibnamefont {Blackwell}}, \bibinfo {author} {\bibfnamefont
  {S.}~\bibnamefont {Boixo}}, \bibinfo {author} {\bibfnamefont
  {G.}~\bibnamefont {Bortoli}}, \bibinfo {author} {\bibfnamefont
  {A.}~\bibnamefont {Bourassa}}, \bibinfo {author} {\bibfnamefont
  {J.}~\bibnamefont {Bovaird}}, \bibinfo {author} {\bibfnamefont
  {L.}~\bibnamefont {Brill}}, \bibinfo {author} {\bibfnamefont
  {M.}~\bibnamefont {Broughton}}, \bibinfo {author} {\bibfnamefont {D.~A.}\
  \bibnamefont {Browne}}, \bibinfo {author} {\bibfnamefont {B.}~\bibnamefont
  {Buchea}}, \bibinfo {author} {\bibfnamefont {B.~B.}\ \bibnamefont {Buckley}},
  \bibinfo {author} {\bibfnamefont {D.~A.}\ \bibnamefont {Buell}}, \bibinfo
  {author} {\bibfnamefont {T.}~\bibnamefont {Burger}}, \bibinfo {author}
  {\bibfnamefont {B.}~\bibnamefont {Burkett}}, \bibinfo {author} {\bibfnamefont
  {N.}~\bibnamefont {Bushnell}}, \bibinfo {author} {\bibfnamefont
  {A.}~\bibnamefont {Cabrera}}, \bibinfo {author} {\bibfnamefont
  {J.}~\bibnamefont {Campero}}, \bibinfo {author} {\bibfnamefont {H.-S.}\
  \bibnamefont {Chang}}, \bibinfo {author} {\bibfnamefont {Y.}~\bibnamefont
  {Chen}}, \bibinfo {author} {\bibfnamefont {Z.}~\bibnamefont {Chen}}, \bibinfo
  {author} {\bibfnamefont {B.}~\bibnamefont {Chiaro}}, \bibinfo {author}
  {\bibfnamefont {D.}~\bibnamefont {Chik}}, \bibinfo {author} {\bibfnamefont
  {C.}~\bibnamefont {Chou}}, \bibinfo {author} {\bibfnamefont {J.}~\bibnamefont
  {Claes}}, \bibinfo {author} {\bibfnamefont {A.~Y.}\ \bibnamefont {Cleland}},
  \bibinfo {author} {\bibfnamefont {J.}~\bibnamefont {Cogan}}, \bibinfo
  {author} {\bibfnamefont {R.}~\bibnamefont {Collins}}, \bibinfo {author}
  {\bibfnamefont {P.}~\bibnamefont {Conner}}, \bibinfo {author} {\bibfnamefont
  {W.}~\bibnamefont {Courtney}}, \bibinfo {author} {\bibfnamefont {A.~L.}\
  \bibnamefont {Crook}}, \bibinfo {author} {\bibfnamefont {B.}~\bibnamefont
  {Curtin}}, \bibinfo {author} {\bibfnamefont {S.}~\bibnamefont {Das}},
  \bibinfo {author} {\bibfnamefont {A.}~\bibnamefont {Davies}}, \bibinfo
  {author} {\bibfnamefont {L.}~\bibnamefont {De~Lorenzo}}, \bibinfo {author}
  {\bibfnamefont {D.~M.}\ \bibnamefont {Debroy}}, \bibinfo {author}
  {\bibfnamefont {S.}~\bibnamefont {Demura}}, \bibinfo {author} {\bibfnamefont
  {M.}~\bibnamefont {Devoret}}, \bibinfo {author} {\bibfnamefont
  {A.}~\bibnamefont {Di~Paolo}}, \bibinfo {author} {\bibfnamefont
  {P.}~\bibnamefont {Donohoe}}, \bibinfo {author} {\bibfnamefont
  {I.}~\bibnamefont {Drozdov}}, \bibinfo {author} {\bibfnamefont
  {A.}~\bibnamefont {Dunsworth}}, \bibinfo {author} {\bibfnamefont
  {C.}~\bibnamefont {Earle}}, \bibinfo {author} {\bibfnamefont
  {T.}~\bibnamefont {Edlich}}, \bibinfo {author} {\bibfnamefont
  {A.}~\bibnamefont {Eickbusch}}, \bibinfo {author} {\bibfnamefont {A.~M.}\
  \bibnamefont {Elbag}}, \bibinfo {author} {\bibfnamefont {M.}~\bibnamefont
  {Elzouka}}, \bibinfo {author} {\bibfnamefont {C.}~\bibnamefont {Erickson}},
  \bibinfo {author} {\bibfnamefont {L.}~\bibnamefont {Faoro}}, \bibinfo
  {author} {\bibfnamefont {E.}~\bibnamefont {Farhi}}, \bibinfo {author}
  {\bibfnamefont {V.~S.}\ \bibnamefont {Ferreira}}, \bibinfo {author}
  {\bibfnamefont {L.~F.}\ \bibnamefont {Burgos}}, \bibinfo {author}
  {\bibfnamefont {E.}~\bibnamefont {Forati}}, \bibinfo {author} {\bibfnamefont
  {A.~G.}\ \bibnamefont {Fowler}}, \bibinfo {author} {\bibfnamefont
  {B.}~\bibnamefont {Foxen}}, \bibinfo {author} {\bibfnamefont
  {S.}~\bibnamefont {Ganjam}}, \bibinfo {author} {\bibfnamefont
  {G.}~\bibnamefont {Garcia}}, \bibinfo {author} {\bibfnamefont
  {R.}~\bibnamefont {Gasca}}, \bibinfo {author} {\bibfnamefont
  {{\'E}.}~\bibnamefont {Genois}}, \bibinfo {author} {\bibfnamefont
  {W.}~\bibnamefont {Giang}}, \bibinfo {author} {\bibfnamefont
  {C.}~\bibnamefont {Gidney}}, \bibinfo {author} {\bibfnamefont
  {D.}~\bibnamefont {Gilboa}}, \bibinfo {author} {\bibfnamefont
  {R.}~\bibnamefont {Gosula}}, \bibinfo {author} {\bibfnamefont {A.~G.}\
  \bibnamefont {Dau}}, \bibinfo {author} {\bibfnamefont {D.}~\bibnamefont
  {Graumann}}, \bibinfo {author} {\bibfnamefont {A.}~\bibnamefont {Greene}},
  \bibinfo {author} {\bibfnamefont {J.~A.}\ \bibnamefont {Gross}}, \bibinfo
  {author} {\bibfnamefont {S.}~\bibnamefont {Habegger}}, \bibinfo {author}
  {\bibfnamefont {J.}~\bibnamefont {Hall}}, \bibinfo {author} {\bibfnamefont
  {M.~C.}\ \bibnamefont {Hamilton}}, \bibinfo {author} {\bibfnamefont
  {M.}~\bibnamefont {Hansen}}, \bibinfo {author} {\bibfnamefont {M.~P.}\
  \bibnamefont {Harrigan}}, \bibinfo {author} {\bibfnamefont {S.~D.}\
  \bibnamefont {Harrington}}, \bibinfo {author} {\bibfnamefont {F.~J.~H.}\
  \bibnamefont {Heras}}, \bibinfo {author} {\bibfnamefont {S.}~\bibnamefont
  {Heslin}}, \bibinfo {author} {\bibfnamefont {P.}~\bibnamefont {Heu}},
  \bibinfo {author} {\bibfnamefont {O.}~\bibnamefont {Higgott}}, \bibinfo
  {author} {\bibfnamefont {G.}~\bibnamefont {Hill}}, \bibinfo {author}
  {\bibfnamefont {J.}~\bibnamefont {Hilton}}, \bibinfo {author} {\bibfnamefont
  {G.}~\bibnamefont {Holland}}, \bibinfo {author} {\bibfnamefont
  {S.}~\bibnamefont {Hong}}, \bibinfo {author} {\bibfnamefont {H.-Y.}\
  \bibnamefont {Huang}}, \bibinfo {author} {\bibfnamefont {A.}~\bibnamefont
  {Huff}}, \bibinfo {author} {\bibfnamefont {W.~J.}\ \bibnamefont {Huggins}},
  \bibinfo {author} {\bibfnamefont {L.~B.}\ \bibnamefont {Ioffe}}, \bibinfo
  {author} {\bibfnamefont {S.~V.}\ \bibnamefont {Isakov}}, \bibinfo {author}
  {\bibfnamefont {J.}~\bibnamefont {Iveland}}, \bibinfo {author} {\bibfnamefont
  {E.}~\bibnamefont {Jeffrey}}, \bibinfo {author} {\bibfnamefont
  {Z.}~\bibnamefont {Jiang}}, \bibinfo {author} {\bibfnamefont
  {C.}~\bibnamefont {Jones}}, \bibinfo {author} {\bibfnamefont
  {S.}~\bibnamefont {Jordan}}, \bibinfo {author} {\bibfnamefont
  {C.}~\bibnamefont {Joshi}}, \bibinfo {author} {\bibfnamefont
  {P.}~\bibnamefont {Juhas}}, \bibinfo {author} {\bibfnamefont
  {D.}~\bibnamefont {Kafri}}, \bibinfo {author} {\bibfnamefont
  {H.}~\bibnamefont {Kang}}, \bibinfo {author} {\bibfnamefont {A.~H.}\
  \bibnamefont {Karamlou}}, \bibinfo {author} {\bibfnamefont {K.}~\bibnamefont
  {Kechedzhi}}, \bibinfo {author} {\bibfnamefont {J.}~\bibnamefont {Kelly}},
  \bibinfo {author} {\bibfnamefont {T.}~\bibnamefont {Khaire}}, \bibinfo
  {author} {\bibfnamefont {T.}~\bibnamefont {Khattar}}, \bibinfo {author}
  {\bibfnamefont {M.}~\bibnamefont {Khezri}}, \bibinfo {author} {\bibfnamefont
  {S.}~\bibnamefont {Kim}}, \bibinfo {author} {\bibfnamefont {P.~V.}\
  \bibnamefont {Klimov}}, \bibinfo {author} {\bibfnamefont {A.~R.}\
  \bibnamefont {Klots}}, \bibinfo {author} {\bibfnamefont {B.}~\bibnamefont
  {Kobrin}}, \bibinfo {author} {\bibfnamefont {P.}~\bibnamefont {Kohli}},
  \bibinfo {author} {\bibfnamefont {A.~N.}\ \bibnamefont {Korotkov}}, \bibinfo
  {author} {\bibfnamefont {F.}~\bibnamefont {Kostritsa}}, \bibinfo {author}
  {\bibfnamefont {R.}~\bibnamefont {Kothari}}, \bibinfo {author} {\bibfnamefont
  {B.}~\bibnamefont {Kozlovskii}}, \bibinfo {author} {\bibfnamefont {J.~M.}\
  \bibnamefont {Kreikebaum}}, \bibinfo {author} {\bibfnamefont {V.~D.}\
  \bibnamefont {Kurilovich}}, \bibinfo {author} {\bibfnamefont
  {N.}~\bibnamefont {Lacroix}}, \bibinfo {author} {\bibfnamefont
  {D.}~\bibnamefont {Landhuis}}, \bibinfo {author} {\bibfnamefont
  {T.}~\bibnamefont {Lange-Dei}}, \bibinfo {author} {\bibfnamefont {B.~W.}\
  \bibnamefont {Langley}}, \bibinfo {author} {\bibfnamefont {P.}~\bibnamefont
  {Laptev}}, \bibinfo {author} {\bibfnamefont {K.-M.}\ \bibnamefont {Lau}},
  \bibinfo {author} {\bibfnamefont {L.}~\bibnamefont {Le~Guevel}}, \bibinfo
  {author} {\bibfnamefont {J.}~\bibnamefont {Ledford}}, \bibinfo {author}
  {\bibfnamefont {J.}~\bibnamefont {Lee}}, \bibinfo {author} {\bibfnamefont
  {K.}~\bibnamefont {Lee}}, \bibinfo {author} {\bibfnamefont {Y.~D.}\
  \bibnamefont {Lensky}}, \bibinfo {author} {\bibfnamefont {S.}~\bibnamefont
  {Leon}}, \bibinfo {author} {\bibfnamefont {B.~J.}\ \bibnamefont {Lester}},
  \bibinfo {author} {\bibfnamefont {W.~Y.}\ \bibnamefont {Li}}, \bibinfo
  {author} {\bibfnamefont {Y.}~\bibnamefont {Li}}, \bibinfo {author}
  {\bibfnamefont {A.~T.}\ \bibnamefont {Lill}}, \bibinfo {author}
  {\bibfnamefont {W.}~\bibnamefont {Liu}}, \bibinfo {author} {\bibfnamefont
  {W.~P.}\ \bibnamefont {Livingston}}, \bibinfo {author} {\bibfnamefont
  {A.}~\bibnamefont {Locharla}}, \bibinfo {author} {\bibfnamefont
  {E.}~\bibnamefont {Lucero}}, \bibinfo {author} {\bibfnamefont
  {D.}~\bibnamefont {Lundahl}}, \bibinfo {author} {\bibfnamefont
  {A.}~\bibnamefont {Lunt}}, \bibinfo {author} {\bibfnamefont {S.}~\bibnamefont
  {Madhuk}}, \bibinfo {author} {\bibfnamefont {F.~D.}\ \bibnamefont {Malone}},
  \bibinfo {author} {\bibfnamefont {A.}~\bibnamefont {Maloney}}, \bibinfo
  {author} {\bibfnamefont {S.}~\bibnamefont {Mandr{\`a}}}, \bibinfo {author}
  {\bibfnamefont {J.}~\bibnamefont {Manyika}}, \bibinfo {author} {\bibfnamefont
  {L.~S.}\ \bibnamefont {Martin}}, \bibinfo {author} {\bibfnamefont
  {O.}~\bibnamefont {Martin}}, \bibinfo {author} {\bibfnamefont
  {S.}~\bibnamefont {Martin}}, \bibinfo {author} {\bibfnamefont
  {C.}~\bibnamefont {Maxfield}}, \bibinfo {author} {\bibfnamefont {J.~R.}\
  \bibnamefont {McClean}}, \bibinfo {author} {\bibfnamefont {M.}~\bibnamefont
  {McEwen}}, \bibinfo {author} {\bibfnamefont {S.}~\bibnamefont {Meeks}},
  \bibinfo {author} {\bibfnamefont {A.}~\bibnamefont {Megrant}}, \bibinfo
  {author} {\bibfnamefont {X.}~\bibnamefont {Mi}}, \bibinfo {author}
  {\bibfnamefont {K.~C.}\ \bibnamefont {Miao}}, \bibinfo {author}
  {\bibfnamefont {A.}~\bibnamefont {Mieszala}}, \bibinfo {author}
  {\bibfnamefont {R.}~\bibnamefont {Molavi}}, \bibinfo {author} {\bibfnamefont
  {S.}~\bibnamefont {Molina}}, \bibinfo {author} {\bibfnamefont
  {S.}~\bibnamefont {Montazeri}}, \bibinfo {author} {\bibfnamefont
  {A.}~\bibnamefont {Morvan}}, \bibinfo {author} {\bibfnamefont
  {R.}~\bibnamefont {Movassagh}}, \bibinfo {author} {\bibfnamefont
  {W.}~\bibnamefont {Mruczkiewicz}}, \bibinfo {author} {\bibfnamefont
  {O.}~\bibnamefont {Naaman}}, \bibinfo {author} {\bibfnamefont
  {M.}~\bibnamefont {Neeley}}, \bibinfo {author} {\bibfnamefont
  {C.}~\bibnamefont {Neill}}, \bibinfo {author} {\bibfnamefont
  {A.}~\bibnamefont {Nersisyan}}, \bibinfo {author} {\bibfnamefont
  {H.}~\bibnamefont {Neven}}, \bibinfo {author} {\bibfnamefont
  {M.}~\bibnamefont {Newman}}, \bibinfo {author} {\bibfnamefont {J.~H.}\
  \bibnamefont {Ng}}, \bibinfo {author} {\bibfnamefont {A.}~\bibnamefont
  {Nguyen}}, \bibinfo {author} {\bibfnamefont {M.}~\bibnamefont {Nguyen}},
  \bibinfo {author} {\bibfnamefont {C.-H.}\ \bibnamefont {Ni}}, \bibinfo
  {author} {\bibfnamefont {M.~Y.}\ \bibnamefont {Niu}}, \bibinfo {author}
  {\bibfnamefont {T.~E.}\ \bibnamefont {O'Brien}}, \bibinfo {author}
  {\bibfnamefont {W.~D.}\ \bibnamefont {Oliver}}, \bibinfo {author}
  {\bibfnamefont {A.}~\bibnamefont {Opremcak}}, \bibinfo {author}
  {\bibfnamefont {K.}~\bibnamefont {Ottosson}}, \bibinfo {author}
  {\bibfnamefont {A.}~\bibnamefont {Petukhov}}, \bibinfo {author}
  {\bibfnamefont {A.}~\bibnamefont {Pizzuto}}, \bibinfo {author} {\bibfnamefont
  {J.}~\bibnamefont {Platt}}, \bibinfo {author} {\bibfnamefont
  {R.}~\bibnamefont {Potter}}, \bibinfo {author} {\bibfnamefont
  {O.}~\bibnamefont {Pritchard}}, \bibinfo {author} {\bibfnamefont {L.~P.}\
  \bibnamefont {Pryadko}}, \bibinfo {author} {\bibfnamefont {C.}~\bibnamefont
  {Quintana}}, \bibinfo {author} {\bibfnamefont {G.}~\bibnamefont
  {Ramachandran}}, \bibinfo {author} {\bibfnamefont {M.~J.}\ \bibnamefont
  {Reagor}}, \bibinfo {author} {\bibfnamefont {J.}~\bibnamefont {Redding}},
  \bibinfo {author} {\bibfnamefont {D.~M.}\ \bibnamefont {Rhodes}}, \bibinfo
  {author} {\bibfnamefont {G.}~\bibnamefont {Roberts}}, \bibinfo {author}
  {\bibfnamefont {E.}~\bibnamefont {Rosenberg}}, \bibinfo {author}
  {\bibfnamefont {E.}~\bibnamefont {Rosenfeld}}, \bibinfo {author}
  {\bibfnamefont {P.}~\bibnamefont {Roushan}}, \bibinfo {author} {\bibfnamefont
  {N.~C.}\ \bibnamefont {Rubin}}, \bibinfo {author} {\bibfnamefont
  {N.}~\bibnamefont {Saei}}, \bibinfo {author} {\bibfnamefont {D.}~\bibnamefont
  {Sank}}, \bibinfo {author} {\bibfnamefont {K.}~\bibnamefont
  {Sankaragomathi}}, \bibinfo {author} {\bibfnamefont {K.~J.}\ \bibnamefont
  {Satzinger}}, \bibinfo {author} {\bibfnamefont {H.~F.}\ \bibnamefont
  {Schurkus}}, \bibinfo {author} {\bibfnamefont {C.}~\bibnamefont {Schuster}},
  \bibinfo {author} {\bibfnamefont {A.~W.}\ \bibnamefont {Senior}}, \bibinfo
  {author} {\bibfnamefont {M.~J.}\ \bibnamefont {Shearn}}, \bibinfo {author}
  {\bibfnamefont {A.}~\bibnamefont {Shorter}}, \bibinfo {author} {\bibfnamefont
  {N.}~\bibnamefont {Shutty}}, \bibinfo {author} {\bibfnamefont
  {V.}~\bibnamefont {Shvarts}}, \bibinfo {author} {\bibfnamefont
  {S.}~\bibnamefont {Singh}}, \bibinfo {author} {\bibfnamefont
  {V.}~\bibnamefont {Sivak}}, \bibinfo {author} {\bibfnamefont
  {J.}~\bibnamefont {Skruzny}}, \bibinfo {author} {\bibfnamefont
  {S.}~\bibnamefont {Small}}, \bibinfo {author} {\bibfnamefont
  {V.}~\bibnamefont {Smelyanskiy}}, \bibinfo {author} {\bibfnamefont {W.~C.}\
  \bibnamefont {Smith}}, \bibinfo {author} {\bibfnamefont {R.~D.}\ \bibnamefont
  {Somma}}, \bibinfo {author} {\bibfnamefont {S.}~\bibnamefont {Springer}},
  \bibinfo {author} {\bibfnamefont {G.}~\bibnamefont {Sterling}}, \bibinfo
  {author} {\bibfnamefont {D.}~\bibnamefont {Strain}}, \bibinfo {author}
  {\bibfnamefont {J.}~\bibnamefont {Suchard}}, \bibinfo {author} {\bibfnamefont
  {A.}~\bibnamefont {Szasz}}, \bibinfo {author} {\bibfnamefont
  {A.}~\bibnamefont {Sztein}}, \bibinfo {author} {\bibfnamefont
  {D.}~\bibnamefont {Thor}}, \bibinfo {author} {\bibfnamefont {A.}~\bibnamefont
  {Torres}}, \bibinfo {author} {\bibfnamefont {M.~M.}\ \bibnamefont
  {Torunbalci}}, \bibinfo {author} {\bibfnamefont {A.}~\bibnamefont
  {Vaishnav}}, \bibinfo {author} {\bibfnamefont {J.}~\bibnamefont {Vargas}},
  \bibinfo {author} {\bibfnamefont {S.}~\bibnamefont {Vdovichev}}, \bibinfo
  {author} {\bibfnamefont {G.}~\bibnamefont {Vidal}}, \bibinfo {author}
  {\bibfnamefont {B.}~\bibnamefont {Villalonga}}, \bibinfo {author}
  {\bibfnamefont {C.~V.}\ \bibnamefont {Heidweiller}}, \bibinfo {author}
  {\bibfnamefont {S.}~\bibnamefont {Waltman}}, \bibinfo {author} {\bibfnamefont
  {S.~X.}\ \bibnamefont {Wang}}, \bibinfo {author} {\bibfnamefont
  {B.}~\bibnamefont {Ware}}, \bibinfo {author} {\bibfnamefont {K.}~\bibnamefont
  {Weber}}, \bibinfo {author} {\bibfnamefont {T.}~\bibnamefont {Weidel}},
  \bibinfo {author} {\bibfnamefont {T.}~\bibnamefont {White}}, \bibinfo
  {author} {\bibfnamefont {K.}~\bibnamefont {Wong}}, \bibinfo {author}
  {\bibfnamefont {B.~W.~K.}\ \bibnamefont {Woo}}, \bibinfo {author}
  {\bibfnamefont {C.}~\bibnamefont {Xing}}, \bibinfo {author} {\bibfnamefont
  {Z.~J.}\ \bibnamefont {Yao}}, \bibinfo {author} {\bibfnamefont
  {P.}~\bibnamefont {Yeh}}, \bibinfo {author} {\bibfnamefont {B.}~\bibnamefont
  {Ying}}, \bibinfo {author} {\bibfnamefont {J.}~\bibnamefont {Yoo}}, \bibinfo
  {author} {\bibfnamefont {N.}~\bibnamefont {Yosri}}, \bibinfo {author}
  {\bibfnamefont {G.}~\bibnamefont {Young}}, \bibinfo {author} {\bibfnamefont
  {A.}~\bibnamefont {Zalcman}}, \bibinfo {author} {\bibfnamefont
  {Y.}~\bibnamefont {Zhang}}, \bibinfo {author} {\bibfnamefont
  {N.}~\bibnamefont {Zhu}}, \bibinfo {author} {\bibfnamefont {N.}~\bibnamefont
  {Zobrist}}, \bibinfo {author} {\bibfnamefont {G.~Q.}\ \bibnamefont {AI}},\
  and\ \bibinfo {author} {\bibnamefont {Collaborators}},\ }\bibfield  {title}
  {\bibinfo {title} {Quantum error correction below the surface code
  threshold},\ }\href {https://doi.org/10.1038/s41586-024-08449-y} {\bibfield
  {journal} {\bibinfo  {journal} {Nature}\ }\textbf {\bibinfo {volume} {638}},\
  \bibinfo {pages} {920} (\bibinfo {year} {2025})}\BibitemShut {NoStop}%
\bibitem [{\citenamefont {Ryan-Anderson}\ \emph {et~al.}(2024)\citenamefont
  {Ryan-Anderson}, \citenamefont {Brown}, \citenamefont {Baldwin},
  \citenamefont {Dreiling}, \citenamefont {Foltz}, \citenamefont {Gaebler},
  \citenamefont {Gatterman}, \citenamefont {Hewitt}, \citenamefont {Holliman},
  \citenamefont {Horst}, \citenamefont {Johansen}, \citenamefont {Lucchetti},
  \citenamefont {Mengle}, \citenamefont {Matheny}, \citenamefont {Matsuoka},
  \citenamefont {Mayer}, \citenamefont {Mills}, \citenamefont {Moses},
  \citenamefont {Neyenhuis}, \citenamefont {Pino}, \citenamefont {Siegfried},
  \citenamefont {Stutz}, \citenamefont {Walker},\ and\ \citenamefont
  {Hayes}}]{Quantinuum24}%
  \BibitemOpen
  \bibfield  {author} {\bibinfo {author} {\bibfnamefont {C.}~\bibnamefont
  {Ryan-Anderson}}, \bibinfo {author} {\bibfnamefont {N.~C.}\ \bibnamefont
  {Brown}}, \bibinfo {author} {\bibfnamefont {C.~H.}\ \bibnamefont {Baldwin}},
  \bibinfo {author} {\bibfnamefont {J.~M.}\ \bibnamefont {Dreiling}}, \bibinfo
  {author} {\bibfnamefont {C.}~\bibnamefont {Foltz}}, \bibinfo {author}
  {\bibfnamefont {J.~P.}\ \bibnamefont {Gaebler}}, \bibinfo {author}
  {\bibfnamefont {T.~M.}\ \bibnamefont {Gatterman}}, \bibinfo {author}
  {\bibfnamefont {N.}~\bibnamefont {Hewitt}}, \bibinfo {author} {\bibfnamefont
  {C.}~\bibnamefont {Holliman}}, \bibinfo {author} {\bibfnamefont {C.~V.}\
  \bibnamefont {Horst}}, \bibinfo {author} {\bibfnamefont {J.}~\bibnamefont
  {Johansen}}, \bibinfo {author} {\bibfnamefont {D.}~\bibnamefont {Lucchetti}},
  \bibinfo {author} {\bibfnamefont {T.}~\bibnamefont {Mengle}}, \bibinfo
  {author} {\bibfnamefont {M.}~\bibnamefont {Matheny}}, \bibinfo {author}
  {\bibfnamefont {Y.}~\bibnamefont {Matsuoka}}, \bibinfo {author}
  {\bibfnamefont {K.}~\bibnamefont {Mayer}}, \bibinfo {author} {\bibfnamefont
  {M.}~\bibnamefont {Mills}}, \bibinfo {author} {\bibfnamefont {S.~A.}\
  \bibnamefont {Moses}}, \bibinfo {author} {\bibfnamefont {B.}~\bibnamefont
  {Neyenhuis}}, \bibinfo {author} {\bibfnamefont {J.}~\bibnamefont {Pino}},
  \bibinfo {author} {\bibfnamefont {P.}~\bibnamefont {Siegfried}}, \bibinfo
  {author} {\bibfnamefont {R.~P.}\ \bibnamefont {Stutz}}, \bibinfo {author}
  {\bibfnamefont {J.}~\bibnamefont {Walker}},\ and\ \bibinfo {author}
  {\bibfnamefont {D.}~\bibnamefont {Hayes}},\ }\bibfield  {title} {\bibinfo
  {title} {High-fidelity teleportation of a logical qubit using transversal
  gates and lattice surgery},\ }\href {https://doi.org/10.1126/science.adp6016}
  {\bibfield  {journal} {\bibinfo  {journal} {Science}\ }\textbf {\bibinfo
  {volume} {385}},\ \bibinfo {pages} {1327} (\bibinfo {year} {2024})},\ \Eprint
  {https://arxiv.org/abs/https://www.science.org/doi/pdf/10.1126/science.adp6016}
  {https://www.science.org/doi/pdf/10.1126/science.adp6016} \BibitemShut
  {NoStop}%
\bibitem [{\citenamefont {Bluvstein}\ \emph {et~al.}(2024)\citenamefont
  {Bluvstein}, \citenamefont {Evered}, \citenamefont {Geim}, \citenamefont
  {Li}, \citenamefont {Zhou}, \citenamefont {Manovitz}, \citenamefont {Ebadi},
  \citenamefont {Cain}, \citenamefont {Kalinowski}, \citenamefont {Hangleiter},
  \citenamefont {Bonilla~Ataides}, \citenamefont {Maskara}, \citenamefont
  {Cong}, \citenamefont {Gao}, \citenamefont {Sales~Rodriguez}, \citenamefont
  {Karolyshyn}, \citenamefont {Semeghini}, \citenamefont {Gullans},
  \citenamefont {Greiner}, \citenamefont {Vuleti{\'c}},\ and\ \citenamefont
  {Lukin}}]{QuEra24}%
  \BibitemOpen
  \bibfield  {author} {\bibinfo {author} {\bibfnamefont {D.}~\bibnamefont
  {Bluvstein}}, \bibinfo {author} {\bibfnamefont {S.~J.}\ \bibnamefont
  {Evered}}, \bibinfo {author} {\bibfnamefont {A.~A.}\ \bibnamefont {Geim}},
  \bibinfo {author} {\bibfnamefont {S.~H.}\ \bibnamefont {Li}}, \bibinfo
  {author} {\bibfnamefont {H.}~\bibnamefont {Zhou}}, \bibinfo {author}
  {\bibfnamefont {T.}~\bibnamefont {Manovitz}}, \bibinfo {author}
  {\bibfnamefont {S.}~\bibnamefont {Ebadi}}, \bibinfo {author} {\bibfnamefont
  {M.}~\bibnamefont {Cain}}, \bibinfo {author} {\bibfnamefont {M.}~\bibnamefont
  {Kalinowski}}, \bibinfo {author} {\bibfnamefont {D.}~\bibnamefont
  {Hangleiter}}, \bibinfo {author} {\bibfnamefont {J.~P.}\ \bibnamefont
  {Bonilla~Ataides}}, \bibinfo {author} {\bibfnamefont {N.}~\bibnamefont
  {Maskara}}, \bibinfo {author} {\bibfnamefont {I.}~\bibnamefont {Cong}},
  \bibinfo {author} {\bibfnamefont {X.}~\bibnamefont {Gao}}, \bibinfo {author}
  {\bibfnamefont {P.}~\bibnamefont {Sales~Rodriguez}}, \bibinfo {author}
  {\bibfnamefont {T.}~\bibnamefont {Karolyshyn}}, \bibinfo {author}
  {\bibfnamefont {G.}~\bibnamefont {Semeghini}}, \bibinfo {author}
  {\bibfnamefont {M.~J.}\ \bibnamefont {Gullans}}, \bibinfo {author}
  {\bibfnamefont {M.}~\bibnamefont {Greiner}}, \bibinfo {author} {\bibfnamefont
  {V.}~\bibnamefont {Vuleti{\'c}}},\ and\ \bibinfo {author} {\bibfnamefont
  {M.~D.}\ \bibnamefont {Lukin}},\ }\bibfield  {title} {\bibinfo {title}
  {Logical quantum processor based on reconfigurable atom arrays},\ }\href
  {https://doi.org/10.1038/s41586-023-06927-3} {\bibfield  {journal} {\bibinfo
  {journal} {Nature}\ }\textbf {\bibinfo {volume} {626}},\ \bibinfo {pages}
  {58} (\bibinfo {year} {2024})}\BibitemShut {NoStop}%
\bibitem [{\citenamefont {Erhard}\ \emph {et~al.}(2021)\citenamefont {Erhard},
  \citenamefont {Poulsen~Nautrup}, \citenamefont {Meth}, \citenamefont
  {Postler}, \citenamefont {Stricker}, \citenamefont {Stadler}, \citenamefont
  {Negnevitsky}, \citenamefont {Ringbauer}, \citenamefont {Schindler},
  \citenamefont {Briegel}, \citenamefont {Blatt}, \citenamefont {Friis},\ and\
  \citenamefont {Monz}}]{Monz21}%
  \BibitemOpen
  \bibfield  {author} {\bibinfo {author} {\bibfnamefont {A.}~\bibnamefont
  {Erhard}}, \bibinfo {author} {\bibfnamefont {H.}~\bibnamefont
  {Poulsen~Nautrup}}, \bibinfo {author} {\bibfnamefont {M.}~\bibnamefont
  {Meth}}, \bibinfo {author} {\bibfnamefont {L.}~\bibnamefont {Postler}},
  \bibinfo {author} {\bibfnamefont {R.}~\bibnamefont {Stricker}}, \bibinfo
  {author} {\bibfnamefont {M.}~\bibnamefont {Stadler}}, \bibinfo {author}
  {\bibfnamefont {V.}~\bibnamefont {Negnevitsky}}, \bibinfo {author}
  {\bibfnamefont {M.}~\bibnamefont {Ringbauer}}, \bibinfo {author}
  {\bibfnamefont {P.}~\bibnamefont {Schindler}}, \bibinfo {author}
  {\bibfnamefont {H.~J.}\ \bibnamefont {Briegel}}, \bibinfo {author}
  {\bibfnamefont {R.}~\bibnamefont {Blatt}}, \bibinfo {author} {\bibfnamefont
  {N.}~\bibnamefont {Friis}},\ and\ \bibinfo {author} {\bibfnamefont
  {T.}~\bibnamefont {Monz}},\ }\bibfield  {title} {\bibinfo {title} {Entangling
  logical qubits with lattice surgery},\ }\href
  {https://doi.org/10.1038/s41586-020-03079-6} {\bibfield  {journal} {\bibinfo
  {journal} {Nature}\ }\textbf {\bibinfo {volume} {589}},\ \bibinfo {pages}
  {220} (\bibinfo {year} {2021})}\BibitemShut {NoStop}%
\bibitem [{\citenamefont {Nigg}\ \emph {et~al.}(2014)\citenamefont {Nigg},
  \citenamefont {M{\"u}ller}, \citenamefont {Martinez}, \citenamefont
  {Schindler}, \citenamefont {Hennrich}, \citenamefont {Monz}, \citenamefont
  {Martin-Delgado},\ and\ \citenamefont {Blatt}}]{Blatt14}%
  \BibitemOpen
  \bibfield  {author} {\bibinfo {author} {\bibfnamefont {D.}~\bibnamefont
  {Nigg}}, \bibinfo {author} {\bibfnamefont {M.}~\bibnamefont {M{\"u}ller}},
  \bibinfo {author} {\bibfnamefont {E.~A.}\ \bibnamefont {Martinez}}, \bibinfo
  {author} {\bibfnamefont {P.}~\bibnamefont {Schindler}}, \bibinfo {author}
  {\bibfnamefont {M.}~\bibnamefont {Hennrich}}, \bibinfo {author}
  {\bibfnamefont {T.}~\bibnamefont {Monz}}, \bibinfo {author} {\bibfnamefont
  {M.~A.}\ \bibnamefont {Martin-Delgado}},\ and\ \bibinfo {author}
  {\bibfnamefont {R.}~\bibnamefont {Blatt}},\ }\bibfield  {title} {\bibinfo
  {title} {Quantum computations on a topologically encoded qubit},\ }\href
  {https://doi.org/10.1126/science.1253742} {\bibfield  {journal} {\bibinfo
  {journal} {Science}\ }\textbf {\bibinfo {volume} {345}},\ \bibinfo {pages}
  {302} (\bibinfo {year} {2014})}\BibitemShut {NoStop}%
\bibitem [{\citenamefont {Litinski}(2019)}]{Litinski19}%
  \BibitemOpen
  \bibfield  {author} {\bibinfo {author} {\bibfnamefont {D.}~\bibnamefont
  {Litinski}},\ }\bibfield  {title} {\bibinfo {title} {A {G}ame of {S}urface
  {C}odes: {L}arge-{S}cale {Q}uantum {C}omputing with {L}attice {S}urgery},\
  }\href {https://doi.org/10.22331/q-2019-03-05-128} {\bibfield  {journal}
  {\bibinfo  {journal} {{Quantum}}\ }\textbf {\bibinfo {volume} {3}},\ \bibinfo
  {pages} {128} (\bibinfo {year} {2019})}\BibitemShut {NoStop}%
\bibitem [{\citenamefont {Gidney}\ and\ \citenamefont
  {Eker{\aa{}}}(2021)}]{Gidney21}%
  \BibitemOpen
  \bibfield  {author} {\bibinfo {author} {\bibfnamefont {C.}~\bibnamefont
  {Gidney}}\ and\ \bibinfo {author} {\bibfnamefont {M.}~\bibnamefont
  {Eker{\aa{}}}},\ }\bibfield  {title} {\bibinfo {title} {How to factor 2048
  bit {RSA} integers in 8 hours using 20 million noisy qubits},\ }\href
  {https://doi.org/10.22331/q-2021-04-15-433} {\bibfield  {journal} {\bibinfo
  {journal} {{Quantum}}\ }\textbf {\bibinfo {volume} {5}},\ \bibinfo {pages}
  {433} (\bibinfo {year} {2021})}\BibitemShut {NoStop}%
\bibitem [{\citenamefont {Alexeev}\ \emph {et~al.}(2021)\citenamefont
  {Alexeev}, \citenamefont {Bacon}, \citenamefont {Brown}, \citenamefont
  {Calderbank}, \citenamefont {Carr}, \citenamefont {Chong}, \citenamefont
  {DeMarco}, \citenamefont {Englund}, \citenamefont {Farhi}, \citenamefont
  {Fefferman}, \citenamefont {Gorshkov}, \citenamefont {Houck}, \citenamefont
  {Kim}, \citenamefont {Kimmel}, \citenamefont {Lange}, \citenamefont {Lloyd},
  \citenamefont {Lukin}, \citenamefont {Maslov}, \citenamefont {Maunz},
  \citenamefont {Monroe}, \citenamefont {Preskill}, \citenamefont {Roetteler},
  \citenamefont {Savage},\ and\ \citenamefont {Thompson}}]{Alexeev21}%
  \BibitemOpen
  \bibfield  {author} {\bibinfo {author} {\bibfnamefont {Y.}~\bibnamefont
  {Alexeev}}, \bibinfo {author} {\bibfnamefont {D.}~\bibnamefont {Bacon}},
  \bibinfo {author} {\bibfnamefont {K.~R.}\ \bibnamefont {Brown}}, \bibinfo
  {author} {\bibfnamefont {R.}~\bibnamefont {Calderbank}}, \bibinfo {author}
  {\bibfnamefont {L.~D.}\ \bibnamefont {Carr}}, \bibinfo {author}
  {\bibfnamefont {F.~T.}\ \bibnamefont {Chong}}, \bibinfo {author}
  {\bibfnamefont {B.}~\bibnamefont {DeMarco}}, \bibinfo {author} {\bibfnamefont
  {D.}~\bibnamefont {Englund}}, \bibinfo {author} {\bibfnamefont
  {E.}~\bibnamefont {Farhi}}, \bibinfo {author} {\bibfnamefont
  {B.}~\bibnamefont {Fefferman}}, \bibinfo {author} {\bibfnamefont {A.~V.}\
  \bibnamefont {Gorshkov}}, \bibinfo {author} {\bibfnamefont {A.}~\bibnamefont
  {Houck}}, \bibinfo {author} {\bibfnamefont {J.}~\bibnamefont {Kim}}, \bibinfo
  {author} {\bibfnamefont {S.}~\bibnamefont {Kimmel}}, \bibinfo {author}
  {\bibfnamefont {M.}~\bibnamefont {Lange}}, \bibinfo {author} {\bibfnamefont
  {S.}~\bibnamefont {Lloyd}}, \bibinfo {author} {\bibfnamefont {M.~D.}\
  \bibnamefont {Lukin}}, \bibinfo {author} {\bibfnamefont {D.}~\bibnamefont
  {Maslov}}, \bibinfo {author} {\bibfnamefont {P.}~\bibnamefont {Maunz}},
  \bibinfo {author} {\bibfnamefont {C.}~\bibnamefont {Monroe}}, \bibinfo
  {author} {\bibfnamefont {J.}~\bibnamefont {Preskill}}, \bibinfo {author}
  {\bibfnamefont {M.}~\bibnamefont {Roetteler}}, \bibinfo {author}
  {\bibfnamefont {M.~J.}\ \bibnamefont {Savage}},\ and\ \bibinfo {author}
  {\bibfnamefont {J.}~\bibnamefont {Thompson}},\ }\bibfield  {title} {\bibinfo
  {title} {Quantum computer systems for scientific discovery},\ }\href
  {https://doi.org/10.1103/PRXQuantum.2.017001} {\bibfield  {journal} {\bibinfo
   {journal} {PRX Quantum}\ }\textbf {\bibinfo {volume} {2}},\ \bibinfo {pages}
  {017001} (\bibinfo {year} {2021})}\BibitemShut {NoStop}%
\bibitem [{\citenamefont {Beverland}\ \emph {et~al.}(2021)\citenamefont
  {Beverland}, \citenamefont {Kubica},\ and\ \citenamefont
  {Svore}}]{Beverland21}%
  \BibitemOpen
  \bibfield  {author} {\bibinfo {author} {\bibfnamefont {M.~E.}\ \bibnamefont
  {Beverland}}, \bibinfo {author} {\bibfnamefont {A.}~\bibnamefont {Kubica}},\
  and\ \bibinfo {author} {\bibfnamefont {K.~M.}\ \bibnamefont {Svore}},\
  }\bibfield  {title} {\bibinfo {title} {Cost of universality: A comparative
  study of the overhead of state distillation and code switching with color
  codes},\ }\href {https://doi.org/10.1103/PRXQuantum.2.020341} {\bibfield
  {journal} {\bibinfo  {journal} {PRX Quantum}\ }\textbf {\bibinfo {volume}
  {2}},\ \bibinfo {pages} {020341} (\bibinfo {year} {2021})}\BibitemShut
  {NoStop}%
\bibitem [{\citenamefont {Beverland}\ \emph {et~al.}(2022)\citenamefont
  {Beverland}, \citenamefont {Murali}, \citenamefont {Troyer}, \citenamefont
  {Svore}, \citenamefont {Hoefler}, \citenamefont {Kliuchnikov}, \citenamefont
  {Low}, \citenamefont {Soeken}, \citenamefont {Sundaram},\ and\ \citenamefont
  {Vaschillo}}]{Beverland22}%
  \BibitemOpen
  \bibfield  {author} {\bibinfo {author} {\bibfnamefont {M.~E.}\ \bibnamefont
  {Beverland}}, \bibinfo {author} {\bibfnamefont {P.}~\bibnamefont {Murali}},
  \bibinfo {author} {\bibfnamefont {M.}~\bibnamefont {Troyer}}, \bibinfo
  {author} {\bibfnamefont {K.~M.}\ \bibnamefont {Svore}}, \bibinfo {author}
  {\bibfnamefont {T.}~\bibnamefont {Hoefler}}, \bibinfo {author} {\bibfnamefont
  {V.}~\bibnamefont {Kliuchnikov}}, \bibinfo {author} {\bibfnamefont {G.~H.}\
  \bibnamefont {Low}}, \bibinfo {author} {\bibfnamefont {M.}~\bibnamefont
  {Soeken}}, \bibinfo {author} {\bibfnamefont {A.}~\bibnamefont {Sundaram}},\
  and\ \bibinfo {author} {\bibfnamefont {A.}~\bibnamefont {Vaschillo}},\
  }\bibfield  {title} {\bibinfo {title} {Assessing requirements to scale to
  practical quantum advantage},\ }\href@noop {} {\bibfield  {journal} {\bibinfo
   {journal} {arXiv preprint arXiv:2211.07629}\ } (\bibinfo {year}
  {2022})}\BibitemShut {NoStop}%
\bibitem [{\citenamefont {Gottesman}(2014)}]{Gottesman14}%
  \BibitemOpen
  \bibfield  {author} {\bibinfo {author} {\bibfnamefont {D.}~\bibnamefont
  {Gottesman}},\ }\bibfield  {title} {\bibinfo {title} {Fault-tolerant quantum
  computation with constant overhead},\ }\href@noop {} {\bibfield  {journal}
  {\bibinfo  {journal} {Quantum Information and Computation}\ }\textbf
  {\bibinfo {volume} {14}},\ \bibinfo {pages} {1339} (\bibinfo {year}
  {2014})}\BibitemShut {NoStop}%
\bibitem [{\citenamefont {MacKay}\ \emph {et~al.}(2004)\citenamefont {MacKay},
  \citenamefont {Mitchison},\ and\ \citenamefont {McFadden}}]{MacKay04}%
  \BibitemOpen
  \bibfield  {author} {\bibinfo {author} {\bibfnamefont {D.}~\bibnamefont
  {MacKay}}, \bibinfo {author} {\bibfnamefont {G.}~\bibnamefont {Mitchison}},\
  and\ \bibinfo {author} {\bibfnamefont {P.}~\bibnamefont {McFadden}},\
  }\bibfield  {title} {\bibinfo {title} {Sparse-graph codes for quantum error
  correction},\ }\href {https://doi.org/10.1109/TIT.2004.834737} {\bibfield
  {journal} {\bibinfo  {journal} {IEEE Transactions on Information Theory}\
  }\textbf {\bibinfo {volume} {50}},\ \bibinfo {pages} {2315} (\bibinfo {year}
  {2004})}\BibitemShut {NoStop}%
\bibitem [{\citenamefont {Kovalev}\ and\ \citenamefont
  {Pryadko}(2013)}]{Kovalev13}%
  \BibitemOpen
  \bibfield  {author} {\bibinfo {author} {\bibfnamefont {A.~A.}\ \bibnamefont
  {Kovalev}}\ and\ \bibinfo {author} {\bibfnamefont {L.~P.}\ \bibnamefont
  {Pryadko}},\ }\bibfield  {title} {\bibinfo {title} {Quantum kronecker
  sum-product low-density parity-check codes with finite rate},\ }\href
  {https://doi.org/10.1103/PhysRevA.88.012311} {\bibfield  {journal} {\bibinfo
  {journal} {Phys. Rev. A}\ }\textbf {\bibinfo {volume} {88}},\ \bibinfo
  {pages} {012311} (\bibinfo {year} {2013})}\BibitemShut {NoStop}%
\bibitem [{\citenamefont {Tillich}\ and\ \citenamefont
  {Z{\'e}mor}(2014)}]{Tillich14}%
  \BibitemOpen
  \bibfield  {author} {\bibinfo {author} {\bibfnamefont {J.-P.}\ \bibnamefont
  {Tillich}}\ and\ \bibinfo {author} {\bibfnamefont {G.}~\bibnamefont
  {Z{\'e}mor}},\ }\bibfield  {title} {\bibinfo {title} {Quantum ldpc codes with
  positive rate and minimum distance proportional to the square root of the
  blocklength},\ }\href {https://doi.org/10.1109/TIT.2013.2292061} {\bibfield
  {journal} {\bibinfo  {journal} {IEEE Transactions on Information Theory}\
  }\textbf {\bibinfo {volume} {60}},\ \bibinfo {pages} {1193} (\bibinfo {year}
  {2014})}\BibitemShut {NoStop}%
\bibitem [{\citenamefont {Bravyi}\ and\ \citenamefont
  {Hastings}(2014)}]{Bravyi14}%
  \BibitemOpen
  \bibfield  {author} {\bibinfo {author} {\bibfnamefont {S.}~\bibnamefont
  {Bravyi}}\ and\ \bibinfo {author} {\bibfnamefont {M.~B.}\ \bibnamefont
  {Hastings}},\ }\bibfield  {title} {\bibinfo {title} {Homological product
  codes},\ }in\ \href {https://doi.org/10.1145/2591796.2591870} {\emph
  {\bibinfo {booktitle} {Proceedings of the Forty-Sixth Annual ACM Symposium on
  Theory of Computing}}},\ \bibinfo {series and number} {STOC '14}\ (\bibinfo
  {publisher} {Association for Computing Machinery},\ \bibinfo {address} {New
  York, NY, USA},\ \bibinfo {year} {2014})\ p.\ \bibinfo {pages}
  {273–282}\BibitemShut {NoStop}%
\bibitem [{\citenamefont {Hastings}\ \emph {et~al.}(2021)\citenamefont
  {Hastings}, \citenamefont {Haah},\ and\ \citenamefont
  {O'Donnell}}]{Hastings20}%
  \BibitemOpen
  \bibfield  {author} {\bibinfo {author} {\bibfnamefont {M.~B.}\ \bibnamefont
  {Hastings}}, \bibinfo {author} {\bibfnamefont {J.}~\bibnamefont {Haah}},\
  and\ \bibinfo {author} {\bibfnamefont {R.}~\bibnamefont {O'Donnell}},\
  }\bibfield  {title} {\bibinfo {title} {Fiber bundle codes: breaking the n1/2
  polylog(n) barrier for quantum ldpc codes},\ }in\ \href
  {https://doi.org/10.1145/3406325.3451005} {\emph {\bibinfo {booktitle}
  {Proceedings of the 53rd Annual ACM SIGACT Symposium on Theory of
  Computing}}},\ \bibinfo {series and number} {STOC 2021}\ (\bibinfo
  {publisher} {Association for Computing Machinery},\ \bibinfo {address} {New
  York, NY, USA},\ \bibinfo {year} {2021})\ p.\ \bibinfo {pages}
  {1276–1288}\BibitemShut {NoStop}%
\bibitem [{\citenamefont {Panteleev}\ and\ \citenamefont
  {Kalachev}(2021)}]{Panteleev21}%
  \BibitemOpen
  \bibfield  {author} {\bibinfo {author} {\bibfnamefont {P.}~\bibnamefont
  {Panteleev}}\ and\ \bibinfo {author} {\bibfnamefont {G.}~\bibnamefont
  {Kalachev}},\ }\bibfield  {title} {\bibinfo {title} {Degenerate {Q}uantum
  {LDPC} {C}odes {W}ith {G}ood {F}inite {L}ength {P}erformance},\ }\href
  {https://doi.org/10.22331/q-2021-11-22-585} {\bibfield  {journal} {\bibinfo
  {journal} {{Quantum}}\ }\textbf {\bibinfo {volume} {5}},\ \bibinfo {pages}
  {585} (\bibinfo {year} {2021})}\BibitemShut {NoStop}%
\bibitem [{\citenamefont {Panteleev}\ and\ \citenamefont
  {Kalachev}(2022{\natexlab{a}})}]{Panteleev22}%
  \BibitemOpen
  \bibfield  {author} {\bibinfo {author} {\bibfnamefont {P.}~\bibnamefont
  {Panteleev}}\ and\ \bibinfo {author} {\bibfnamefont {G.}~\bibnamefont
  {Kalachev}},\ }\bibfield  {title} {\bibinfo {title} {Quantum ldpc codes with
  almost linear minimum distance},\ }\href
  {https://doi.org/10.1109/TIT.2021.3119384} {\bibfield  {journal} {\bibinfo
  {journal} {IEEE Transactions on Information Theory}\ }\textbf {\bibinfo
  {volume} {68}},\ \bibinfo {pages} {213} (\bibinfo {year}
  {2022}{\natexlab{a}})}\BibitemShut {NoStop}%
\bibitem [{\citenamefont {Panteleev}\ and\ \citenamefont
  {Kalachev}(2022{\natexlab{b}})}]{Panteleev22b}%
  \BibitemOpen
  \bibfield  {author} {\bibinfo {author} {\bibfnamefont {P.}~\bibnamefont
  {Panteleev}}\ and\ \bibinfo {author} {\bibfnamefont {G.}~\bibnamefont
  {Kalachev}},\ }\bibfield  {title} {\bibinfo {title} {Asymptotically good
  quantum and locally testable classical ldpc codes},\ }in\ \href
  {https://doi.org/10.1145/3519935.3520017} {\emph {\bibinfo {booktitle}
  {Proceedings of the 54th Annual ACM SIGACT Symposium on Theory of
  Computing}}},\ \bibinfo {series and number} {STOC 2022}\ (\bibinfo
  {publisher} {Association for Computing Machinery},\ \bibinfo {address} {New
  York, NY, USA},\ \bibinfo {year} {2022})\ p.\ \bibinfo {pages}
  {375–388}\BibitemShut {NoStop}%
\bibitem [{\citenamefont {Breuckmann}\ and\ \citenamefont
  {Eberhardt}(2021{\natexlab{a}})}]{Breuckmann21a}%
  \BibitemOpen
  \bibfield  {author} {\bibinfo {author} {\bibfnamefont {N.~P.}\ \bibnamefont
  {Breuckmann}}\ and\ \bibinfo {author} {\bibfnamefont {J.~N.}\ \bibnamefont
  {Eberhardt}},\ }\bibfield  {title} {\bibinfo {title} {Balanced product
  quantum codes},\ }\href {https://doi.org/10.1109/TIT.2021.3097347} {\bibfield
   {journal} {\bibinfo  {journal} {IEEE Trans. Inf. Theor.}\ }\textbf {\bibinfo
  {volume} {67}},\ \bibinfo {pages} {6653–6674} (\bibinfo {year}
  {2021}{\natexlab{a}})}\BibitemShut {NoStop}%
\bibitem [{\citenamefont {Breuckmann}\ and\ \citenamefont
  {Eberhardt}(2021{\natexlab{b}})}]{Breuckmann21b}%
  \BibitemOpen
  \bibfield  {author} {\bibinfo {author} {\bibfnamefont {N.~P.}\ \bibnamefont
  {Breuckmann}}\ and\ \bibinfo {author} {\bibfnamefont {J.~N.}\ \bibnamefont
  {Eberhardt}},\ }\bibfield  {title} {\bibinfo {title} {Quantum low-density
  parity-check codes},\ }\href {https://doi.org/10.1103/PRXQuantum.2.040101}
  {\bibfield  {journal} {\bibinfo  {journal} {PRX Quantum}\ }\textbf {\bibinfo
  {volume} {2}},\ \bibinfo {pages} {040101} (\bibinfo {year}
  {2021}{\natexlab{b}})}\BibitemShut {NoStop}%
\bibitem [{\citenamefont {Leverrier}\ and\ \citenamefont
  {Z{\'e}mor}(2022)}]{Leverrier22}%
  \BibitemOpen
  \bibfield  {author} {\bibinfo {author} {\bibfnamefont {A.}~\bibnamefont
  {Leverrier}}\ and\ \bibinfo {author} {\bibfnamefont {G.}~\bibnamefont
  {Z{\'e}mor}},\ }\bibfield  {title} {\bibinfo {title} {Quantum tanner codes},\
  }in\ \href {https://doi.org/10.1109/FOCS54457.2022.00117} {\emph {\bibinfo
  {booktitle} {2022 IEEE 63rd Annual Symposium on Foundations of Computer
  Science (FOCS)}}}\ (\bibinfo {year} {2022})\ pp.\ \bibinfo {pages}
  {872--883}\BibitemShut {NoStop}%
\bibitem [{\citenamefont {Rakovszky}\ and\ \citenamefont
  {Khemani}(2023)}]{Rakovszky23}%
  \BibitemOpen
  \bibfield  {author} {\bibinfo {author} {\bibfnamefont {T.}~\bibnamefont
  {Rakovszky}}\ and\ \bibinfo {author} {\bibfnamefont {V.}~\bibnamefont
  {Khemani}},\ }\bibfield  {title} {\bibinfo {title} {The physics of (good)
  ldpc codes i. gauging and dualities},\ }\href@noop {} {\bibfield  {journal}
  {\bibinfo  {journal} {arXiv preprint arXiv:2310.16032}\ } (\bibinfo {year}
  {2023})}\BibitemShut {NoStop}%
\bibitem [{\citenamefont {Rakovszky}\ and\ \citenamefont
  {Khemani}(2024)}]{Rakovszky24}%
  \BibitemOpen
  \bibfield  {author} {\bibinfo {author} {\bibfnamefont {T.}~\bibnamefont
  {Rakovszky}}\ and\ \bibinfo {author} {\bibfnamefont {V.}~\bibnamefont
  {Khemani}},\ }\bibfield  {title} {\bibinfo {title} {The physics of (good)
  ldpc codes ii. product constructions},\ }\href@noop {} {\bibfield  {journal}
  {\bibinfo  {journal} {arXiv preprint arXiv:2402.16831}\ } (\bibinfo {year}
  {2024})}\BibitemShut {NoStop}%
\bibitem [{\citenamefont {De~Roeck}\ \emph {et~al.}(2024)\citenamefont
  {De~Roeck}, \citenamefont {Khemani}, \citenamefont {Li}, \citenamefont
  {O'Dea},\ and\ \citenamefont {Rakovszky}}]{DeRoeck24}%
  \BibitemOpen
  \bibfield  {author} {\bibinfo {author} {\bibfnamefont {W.}~\bibnamefont
  {De~Roeck}}, \bibinfo {author} {\bibfnamefont {V.}~\bibnamefont {Khemani}},
  \bibinfo {author} {\bibfnamefont {Y.}~\bibnamefont {Li}}, \bibinfo {author}
  {\bibfnamefont {N.}~\bibnamefont {O'Dea}},\ and\ \bibinfo {author}
  {\bibfnamefont {T.}~\bibnamefont {Rakovszky}},\ }\bibfield  {title} {\bibinfo
  {title} {Ldpc stabilizer codes as gapped quantum phases: stability under
  graph-local perturbations},\ }\href@noop {} {\bibfield  {journal} {\bibinfo
  {journal} {arXiv preprint arXiv:2411.02384}\ } (\bibinfo {year}
  {2024})}\BibitemShut {NoStop}%
\bibitem [{\citenamefont {Yin}\ and\ \citenamefont {Lucas}(2024)}]{Yin24}%
  \BibitemOpen
  \bibfield  {author} {\bibinfo {author} {\bibfnamefont {C.}~\bibnamefont
  {Yin}}\ and\ \bibinfo {author} {\bibfnamefont {A.}~\bibnamefont {Lucas}},\
  }\bibfield  {title} {\bibinfo {title} {Low-density parity-check codes as
  stable phases of quantum matter},\ }\href@noop {} {\bibfield  {journal}
  {\bibinfo  {journal} {arXiv preprint arXiv:2411.01002}\ } (\bibinfo {year}
  {2024})}\BibitemShut {NoStop}%
\bibitem [{\citenamefont {Liang}\ \emph
  {et~al.}(2024{\natexlab{a}})\citenamefont {Liang}, \citenamefont {Xu},
  \citenamefont {Iosue},\ and\ \citenamefont {Chen}}]{liang2024extracting}%
  \BibitemOpen
  \bibfield  {author} {\bibinfo {author} {\bibfnamefont {Z.}~\bibnamefont
  {Liang}}, \bibinfo {author} {\bibfnamefont {Y.}~\bibnamefont {Xu}}, \bibinfo
  {author} {\bibfnamefont {J.~T.}\ \bibnamefont {Iosue}},\ and\ \bibinfo
  {author} {\bibfnamefont {Y.-A.}\ \bibnamefont {Chen}},\ }\bibfield  {title}
  {\bibinfo {title} {Extracting topological orders of generalized pauli
  stabilizer codes in two dimensions},\ }\href
  {https://doi.org/10.1103/PRXQuantum.5.030328} {\bibfield  {journal} {\bibinfo
   {journal} {PRX Quantum}\ }\textbf {\bibinfo {volume} {5}},\ \bibinfo {pages}
  {030328} (\bibinfo {year} {2024}{\natexlab{a}})}\BibitemShut {NoStop}%
\bibitem [{\citenamefont {Liang}\ \emph
  {et~al.}(2024{\natexlab{b}})\citenamefont {Liang}, \citenamefont {Yang},
  \citenamefont {Iosue},\ and\ \citenamefont {Chen}}]{liang2024operator}%
  \BibitemOpen
  \bibfield  {author} {\bibinfo {author} {\bibfnamefont {Z.}~\bibnamefont
  {Liang}}, \bibinfo {author} {\bibfnamefont {B.}~\bibnamefont {Yang}},
  \bibinfo {author} {\bibfnamefont {J.~T.}\ \bibnamefont {Iosue}},\ and\
  \bibinfo {author} {\bibfnamefont {Y.-A.}\ \bibnamefont {Chen}},\ }\bibfield
  {title} {\bibinfo {title} {Operator algebra and algorithmic construction of
  boundaries and defects in (2+ 1) d topological pauli stabilizer codes},\
  }\href@noop {} {\bibfield  {journal} {\bibinfo  {journal} {arXiv preprint
  arXiv:2410.11942}\ } (\bibinfo {year} {2024}{\natexlab{b}})}\BibitemShut
  {NoStop}%
\bibitem [{\citenamefont {Zhu}(2025)}]{Zhu25}%
  \BibitemOpen
  \bibfield  {author} {\bibinfo {author} {\bibfnamefont {G.}~\bibnamefont
  {Zhu}},\ }\bibfield  {title} {\bibinfo {title} {A topological theory for
  qldpc: non-clifford gates and magic state fountain on homological product
  codes with constant rate and beyond the $n^{1/3}$ distance barrier},\
  }\href@noop {} {\bibfield  {journal} {\bibinfo  {journal} {arXiv preprint
  arXiv:2501.19375}\ } (\bibinfo {year} {2025})}\BibitemShut {NoStop}%
\bibitem [{\citenamefont {Tremblay}\ \emph {et~al.}(2022)\citenamefont
  {Tremblay}, \citenamefont {Delfosse},\ and\ \citenamefont
  {Beverland}}]{Tremblay22}%
  \BibitemOpen
  \bibfield  {author} {\bibinfo {author} {\bibfnamefont {M.~A.}\ \bibnamefont
  {Tremblay}}, \bibinfo {author} {\bibfnamefont {N.}~\bibnamefont {Delfosse}},\
  and\ \bibinfo {author} {\bibfnamefont {M.~E.}\ \bibnamefont {Beverland}},\
  }\bibfield  {title} {\bibinfo {title} {Constant-overhead quantum error
  correction with thin planar connectivity},\ }\href
  {https://doi.org/10.1103/PhysRevLett.129.050504} {\bibfield  {journal}
  {\bibinfo  {journal} {Phys. Rev. Lett.}\ }\textbf {\bibinfo {volume} {129}},\
  \bibinfo {pages} {050504} (\bibinfo {year} {2022})}\BibitemShut {NoStop}%
\bibitem [{\citenamefont {Bravyi}\ \emph {et~al.}(2024)\citenamefont {Bravyi},
  \citenamefont {Cross}, \citenamefont {Gambetta}, \citenamefont {Maslov},
  \citenamefont {Rall},\ and\ \citenamefont {Yoder}}]{Bravyi24}%
  \BibitemOpen
  \bibfield  {author} {\bibinfo {author} {\bibfnamefont {S.}~\bibnamefont
  {Bravyi}}, \bibinfo {author} {\bibfnamefont {A.~W.}\ \bibnamefont {Cross}},
  \bibinfo {author} {\bibfnamefont {J.~M.}\ \bibnamefont {Gambetta}}, \bibinfo
  {author} {\bibfnamefont {D.}~\bibnamefont {Maslov}}, \bibinfo {author}
  {\bibfnamefont {P.}~\bibnamefont {Rall}},\ and\ \bibinfo {author}
  {\bibfnamefont {T.~J.}\ \bibnamefont {Yoder}},\ }\bibfield  {title} {\bibinfo
  {title} {High-threshold and low-overhead fault-tolerant quantum memory},\
  }\href {https://doi.org/10.1038/s41586-024-07107-7} {\bibfield  {journal}
  {\bibinfo  {journal} {Nature}\ }\textbf {\bibinfo {volume} {627}},\ \bibinfo
  {pages} {778} (\bibinfo {year} {2024})},\ \Eprint
  {https://arxiv.org/abs/2308.07915} {arXiv:2308.07915} \BibitemShut {NoStop}%
\bibitem [{\citenamefont {Xu}\ \emph {et~al.}(2024)\citenamefont {Xu},
  \citenamefont {Bonilla~Ataides}, \citenamefont {Pattison}, \citenamefont
  {Raveendran}, \citenamefont {Bluvstein}, \citenamefont {Wurtz}, \citenamefont
  {Vasi{\'c}}, \citenamefont {Lukin}, \citenamefont {Jiang},\ and\
  \citenamefont {Zhou}}]{Xu24}%
  \BibitemOpen
  \bibfield  {author} {\bibinfo {author} {\bibfnamefont {Q.}~\bibnamefont
  {Xu}}, \bibinfo {author} {\bibfnamefont {J.~P.}\ \bibnamefont
  {Bonilla~Ataides}}, \bibinfo {author} {\bibfnamefont {C.~A.}\ \bibnamefont
  {Pattison}}, \bibinfo {author} {\bibfnamefont {N.}~\bibnamefont
  {Raveendran}}, \bibinfo {author} {\bibfnamefont {D.}~\bibnamefont
  {Bluvstein}}, \bibinfo {author} {\bibfnamefont {J.}~\bibnamefont {Wurtz}},
  \bibinfo {author} {\bibfnamefont {B.}~\bibnamefont {Vasi{\'c}}}, \bibinfo
  {author} {\bibfnamefont {M.~D.}\ \bibnamefont {Lukin}}, \bibinfo {author}
  {\bibfnamefont {L.}~\bibnamefont {Jiang}},\ and\ \bibinfo {author}
  {\bibfnamefont {H.}~\bibnamefont {Zhou}},\ }\bibfield  {title} {\bibinfo
  {title} {Constant-overhead fault-tolerant quantum computation with
  reconfigurable atom arrays},\ }\href
  {https://doi.org/10.1038/s41567-024-02479-z} {\bibfield  {journal} {\bibinfo
  {journal} {Nature Physics}\ }\textbf {\bibinfo {volume} {20}},\ \bibinfo
  {pages} {1084} (\bibinfo {year} {2024})}\BibitemShut {NoStop}%
\bibitem [{\citenamefont {Nguyen}\ and\ \citenamefont
  {Pattison}(2024)}]{Nguyen24}%
  \BibitemOpen
  \bibfield  {author} {\bibinfo {author} {\bibfnamefont {Q.~T.}\ \bibnamefont
  {Nguyen}}\ and\ \bibinfo {author} {\bibfnamefont {C.~A.}\ \bibnamefont
  {Pattison}},\ }\bibfield  {title} {\bibinfo {title} {Quantum fault tolerance
  with constant-space and logarithmic-time overheads},\ }\href@noop {}
  {\bibfield  {journal} {\bibinfo  {journal} {arXiv preprint arXiv:2411.03632}\
  } (\bibinfo {year} {2024})}\BibitemShut {NoStop}%
\bibitem [{\citenamefont {Tamiya}\ \emph {et~al.}(2024)\citenamefont {Tamiya},
  \citenamefont {Koashi},\ and\ \citenamefont {Yamasaki}}]{Tamiya24}%
  \BibitemOpen
  \bibfield  {author} {\bibinfo {author} {\bibfnamefont {S.}~\bibnamefont
  {Tamiya}}, \bibinfo {author} {\bibfnamefont {M.}~\bibnamefont {Koashi}},\
  and\ \bibinfo {author} {\bibfnamefont {H.}~\bibnamefont {Yamasaki}},\
  }\bibfield  {title} {\bibinfo {title} {Polylog-time-and
  constant-space-overhead fault-tolerant quantum computation with quantum
  low-density parity-check codes},\ }\href@noop {} {\bibfield  {journal}
  {\bibinfo  {journal} {arXiv preprint arXiv:2411.03683}\ } (\bibinfo {year}
  {2024})}\BibitemShut {NoStop}%
\bibitem [{\citenamefont {Wang}\ and\ \citenamefont {Pryadko}(2022)}]{Wang22}%
  \BibitemOpen
  \bibfield  {author} {\bibinfo {author} {\bibfnamefont {R.}~\bibnamefont
  {Wang}}\ and\ \bibinfo {author} {\bibfnamefont {L.~P.}\ \bibnamefont
  {Pryadko}},\ }\bibfield  {title} {\bibinfo {title} {Distance bounds for
  generalized bicycle codes},\ }\bibfield  {journal} {\bibinfo  {journal}
  {Symmetry}\ }\textbf {\bibinfo {volume} {14}},\ \href
  {https://doi.org/10.3390/sym14071348} {10.3390/sym14071348} (\bibinfo {year}
  {2022})\BibitemShut {NoStop}%
\bibitem [{\citenamefont {Hong}\ \emph {et~al.}(2024)\citenamefont {Hong},
  \citenamefont {Durso-Sabina}, \citenamefont {Hayes},\ and\ \citenamefont
  {Lucas}}]{Hong24}%
  \BibitemOpen
  \bibfield  {author} {\bibinfo {author} {\bibfnamefont {Y.}~\bibnamefont
  {Hong}}, \bibinfo {author} {\bibfnamefont {E.}~\bibnamefont {Durso-Sabina}},
  \bibinfo {author} {\bibfnamefont {D.}~\bibnamefont {Hayes}},\ and\ \bibinfo
  {author} {\bibfnamefont {A.}~\bibnamefont {Lucas}},\ }\bibfield  {title}
  {\bibinfo {title} {Entangling four logical qubits beyond break-even in a
  nonlocal code},\ }\href {https://doi.org/10.1103/PhysRevLett.133.180601}
  {\bibfield  {journal} {\bibinfo  {journal} {Phys. Rev. Lett.}\ }\textbf
  {\bibinfo {volume} {133}},\ \bibinfo {pages} {180601} (\bibinfo {year}
  {2024})}\BibitemShut {NoStop}%
\bibitem [{\citenamefont {Lin}\ and\ \citenamefont {Pryadko}(2024)}]{Lin24}%
  \BibitemOpen
  \bibfield  {author} {\bibinfo {author} {\bibfnamefont {H.-K.}\ \bibnamefont
  {Lin}}\ and\ \bibinfo {author} {\bibfnamefont {L.~P.}\ \bibnamefont
  {Pryadko}},\ }\bibfield  {title} {\bibinfo {title} {Quantum two-block group
  algebra codes},\ }\href {https://doi.org/10.1103/PhysRevA.109.022407}
  {\bibfield  {journal} {\bibinfo  {journal} {Phys. Rev. A}\ }\textbf {\bibinfo
  {volume} {109}},\ \bibinfo {pages} {022407} (\bibinfo {year}
  {2024})}\BibitemShut {NoStop}%
\bibitem [{\citenamefont {Pecorari}\ \emph {et~al.}(2025)\citenamefont
  {Pecorari}, \citenamefont {Jandura}, \citenamefont {Brennen},\ and\
  \citenamefont {Pupillo}}]{Pecorari25}%
  \BibitemOpen
  \bibfield  {author} {\bibinfo {author} {\bibfnamefont {L.}~\bibnamefont
  {Pecorari}}, \bibinfo {author} {\bibfnamefont {S.}~\bibnamefont {Jandura}},
  \bibinfo {author} {\bibfnamefont {G.~K.}\ \bibnamefont {Brennen}},\ and\
  \bibinfo {author} {\bibfnamefont {G.}~\bibnamefont {Pupillo}},\ }\bibfield
  {title} {\bibinfo {title} {High-rate quantum ldpc codes for
  long-range-connected neutral atom registers},\ }\href
  {https://doi.org/10.1038/s41467-025-56255-5} {\bibfield  {journal} {\bibinfo
  {journal} {Nature Communications}\ }\textbf {\bibinfo {volume} {16}},\
  \bibinfo {pages} {1111} (\bibinfo {year} {2025})}\BibitemShut {NoStop}%
\bibitem [{\citenamefont {Ruiz}\ \emph {et~al.}(2025)\citenamefont {Ruiz},
  \citenamefont {Guillaud}, \citenamefont {Leverrier}, \citenamefont
  {Mirrahimi},\ and\ \citenamefont {Vuillot}}]{Ruiz25}%
  \BibitemOpen
  \bibfield  {author} {\bibinfo {author} {\bibfnamefont {D.}~\bibnamefont
  {Ruiz}}, \bibinfo {author} {\bibfnamefont {J.}~\bibnamefont {Guillaud}},
  \bibinfo {author} {\bibfnamefont {A.}~\bibnamefont {Leverrier}}, \bibinfo
  {author} {\bibfnamefont {M.}~\bibnamefont {Mirrahimi}},\ and\ \bibinfo
  {author} {\bibfnamefont {C.}~\bibnamefont {Vuillot}},\ }\bibfield  {title}
  {\bibinfo {title} {Ldpc-cat codes for low-overhead quantum computing in 2d},\
  }\href {https://doi.org/10.1038/s41467-025-56298-8} {\bibfield  {journal}
  {\bibinfo  {journal} {Nature Communications}\ }\textbf {\bibinfo {volume}
  {16}},\ \bibinfo {pages} {1040} (\bibinfo {year} {2025})}\BibitemShut
  {NoStop}%
\bibitem [{\citenamefont {Scruby}\ \emph {et~al.}(2024)\citenamefont {Scruby},
  \citenamefont {Hillmann},\ and\ \citenamefont {Roffe}}]{Scruby24}%
  \BibitemOpen
  \bibfield  {author} {\bibinfo {author} {\bibfnamefont {T.~R.}\ \bibnamefont
  {Scruby}}, \bibinfo {author} {\bibfnamefont {T.}~\bibnamefont {Hillmann}},\
  and\ \bibinfo {author} {\bibfnamefont {J.}~\bibnamefont {Roffe}},\ }\bibfield
   {title} {\bibinfo {title} {High-threshold, low-overhead and single-shot
  decodable fault-tolerant quantum memory},\ }\href@noop {} {\bibfield
  {journal} {\bibinfo  {journal} {arXiv preprint arXiv:2406.14445}\ } (\bibinfo
  {year} {2024})}\BibitemShut {NoStop}%
\bibitem [{\citenamefont {Zhang}\ and\ \citenamefont {Li}(2025)}]{Zhang25}%
  \BibitemOpen
  \bibfield  {author} {\bibinfo {author} {\bibfnamefont {G.}~\bibnamefont
  {Zhang}}\ and\ \bibinfo {author} {\bibfnamefont {Y.}~\bibnamefont {Li}},\
  }\bibfield  {title} {\bibinfo {title} {Time-efficient logical operations on
  quantum low-density parity check codes},\ }\href
  {https://doi.org/10.1103/PhysRevLett.134.070602} {\bibfield  {journal}
  {\bibinfo  {journal} {Phys. Rev. Lett.}\ }\textbf {\bibinfo {volume} {134}},\
  \bibinfo {pages} {070602} (\bibinfo {year} {2025})}\BibitemShut {NoStop}%
\bibitem [{\citenamefont {Bravyi}\ \emph {et~al.}(2022)\citenamefont {Bravyi},
  \citenamefont {Dial}, \citenamefont {Gambetta}, \citenamefont {Gil},\ and\
  \citenamefont {Nazario}}]{Bravyi22}%
  \BibitemOpen
  \bibfield  {author} {\bibinfo {author} {\bibfnamefont {S.}~\bibnamefont
  {Bravyi}}, \bibinfo {author} {\bibfnamefont {O.}~\bibnamefont {Dial}},
  \bibinfo {author} {\bibfnamefont {J.~M.}\ \bibnamefont {Gambetta}}, \bibinfo
  {author} {\bibfnamefont {D.}~\bibnamefont {Gil}},\ and\ \bibinfo {author}
  {\bibfnamefont {Z.}~\bibnamefont {Nazario}},\ }\bibfield  {title} {\bibinfo
  {title} {The future of quantum computing with superconducting qubits},\
  }\href {https://doi.org/10.1063/5.0082975} {\bibfield  {journal} {\bibinfo
  {journal} {Journal of Applied Physics}\ }\textbf {\bibinfo {volume} {132}},\
  \bibinfo {pages} {160902} (\bibinfo {year} {2022})},\ \Eprint
  {https://arxiv.org/abs/https://pubs.aip.org/aip/jap/article-pdf/doi/10.1063/5.0082975/20034201/160902\_1\_5.0082975.pdf}
  {https://pubs.aip.org/aip/jap/article-pdf/doi/10.1063/5.0082975/20034201/160902\_1\_5.0082975.pdf}
  \BibitemShut {NoStop}%
\bibitem [{\citenamefont {Berthusen}\ \emph {et~al.}(2025)\citenamefont
  {Berthusen}, \citenamefont {Devulapalli}, \citenamefont {Schoute},
  \citenamefont {Childs}, \citenamefont {Gullans}, \citenamefont {Gorshkov},\
  and\ \citenamefont {Gottesman}}]{Berthusen25}%
  \BibitemOpen
  \bibfield  {author} {\bibinfo {author} {\bibfnamefont {N.}~\bibnamefont
  {Berthusen}}, \bibinfo {author} {\bibfnamefont {D.}~\bibnamefont
  {Devulapalli}}, \bibinfo {author} {\bibfnamefont {E.}~\bibnamefont
  {Schoute}}, \bibinfo {author} {\bibfnamefont {A.~M.}\ \bibnamefont {Childs}},
  \bibinfo {author} {\bibfnamefont {M.~J.}\ \bibnamefont {Gullans}}, \bibinfo
  {author} {\bibfnamefont {A.~V.}\ \bibnamefont {Gorshkov}},\ and\ \bibinfo
  {author} {\bibfnamefont {D.}~\bibnamefont {Gottesman}},\ }\bibfield  {title}
  {\bibinfo {title} {Toward a 2d local implementation of quantum low-density
  parity-check codes},\ }\href {https://doi.org/10.1103/PRXQuantum.6.010306}
  {\bibfield  {journal} {\bibinfo  {journal} {PRX Quantum}\ }\textbf {\bibinfo
  {volume} {6}},\ \bibinfo {pages} {010306} (\bibinfo {year}
  {2025})}\BibitemShut {NoStop}%
\bibitem [{\citenamefont {Poole}\ \emph {et~al.}(2024)\citenamefont {Poole},
  \citenamefont {Graham}, \citenamefont {Perlin}, \citenamefont {Otten},\ and\
  \citenamefont {Saffman}}]{Poole24}%
  \BibitemOpen
  \bibfield  {author} {\bibinfo {author} {\bibfnamefont {C.}~\bibnamefont
  {Poole}}, \bibinfo {author} {\bibfnamefont {T.}~\bibnamefont {Graham}},
  \bibinfo {author} {\bibfnamefont {M.}~\bibnamefont {Perlin}}, \bibinfo
  {author} {\bibfnamefont {M.}~\bibnamefont {Otten}},\ and\ \bibinfo {author}
  {\bibfnamefont {M.}~\bibnamefont {Saffman}},\ }\bibfield  {title} {\bibinfo
  {title} {Architecture for fast implementation of qldpc codes with optimized
  rydberg gates},\ }\href@noop {} {\bibfield  {journal} {\bibinfo  {journal}
  {arXiv preprint arXiv:2404.18809}\ } (\bibinfo {year} {2024})}\BibitemShut
  {NoStop}%
\bibitem [{\citenamefont {Gong}\ \emph {et~al.}(2024)\citenamefont {Gong},
  \citenamefont {Cammerer},\ and\ \citenamefont {Renes}}]{Gong24}%
  \BibitemOpen
  \bibfield  {author} {\bibinfo {author} {\bibfnamefont {A.}~\bibnamefont
  {Gong}}, \bibinfo {author} {\bibfnamefont {S.}~\bibnamefont {Cammerer}},\
  and\ \bibinfo {author} {\bibfnamefont {J.~M.}\ \bibnamefont {Renes}},\
  }\bibfield  {title} {\bibinfo {title} {Toward low-latency iterative decoding
  of qldpc codes under circuit-level noise},\ }\href@noop {} {\bibfield
  {journal} {\bibinfo  {journal} {arXiv preprint arXiv:2403.18901}\ } (\bibinfo
  {year} {2024})}\BibitemShut {NoStop}%
\bibitem [{\citenamefont {Shaw}\ and\ \citenamefont {Terhal}(2024)}]{Shaw24}%
  \BibitemOpen
  \bibfield  {author} {\bibinfo {author} {\bibfnamefont {M.~H.}\ \bibnamefont
  {Shaw}}\ and\ \bibinfo {author} {\bibfnamefont {B.~M.}\ \bibnamefont
  {Terhal}},\ }\bibfield  {title} {\bibinfo {title} {Lowering connectivity
  requirements for bivariate bicycle codes using morphing circuits},\
  }\href@noop {} {\bibfield  {journal} {\bibinfo  {journal} {arXiv preprint
  arXiv:2407.16336}\ } (\bibinfo {year} {2024})}\BibitemShut {NoStop}%
\bibitem [{\citenamefont {Wang}\ and\ \citenamefont {Mueller}(2024)}]{Wang24}%
  \BibitemOpen
  \bibfield  {author} {\bibinfo {author} {\bibfnamefont {M.}~\bibnamefont
  {Wang}}\ and\ \bibinfo {author} {\bibfnamefont {F.}~\bibnamefont {Mueller}},\
  }\bibfield  {title} {\bibinfo {title} {Rate adjustable bivariate bicycle
  codes for quantum error correction},\ }in\ \href
  {https://doi.org/10.1109/QCE60285.2024.10331} {\emph {\bibinfo {booktitle}
  {2024 IEEE International Conference on Quantum Computing and Engineering
  (QCE)}}},\ Vol.~\bibinfo {volume} {02}\ (\bibinfo {year} {2024})\ pp.\
  \bibinfo {pages} {412--413}\BibitemShut {NoStop}%
\bibitem [{\citenamefont {Sayginel}\ \emph {et~al.}(2024)\citenamefont
  {Sayginel}, \citenamefont {Koutsioumpas}, \citenamefont {Webster},
  \citenamefont {Rajput},\ and\ \citenamefont {Browne}}]{Sayginel24}%
  \BibitemOpen
  \bibfield  {author} {\bibinfo {author} {\bibfnamefont {H.}~\bibnamefont
  {Sayginel}}, \bibinfo {author} {\bibfnamefont {S.}~\bibnamefont
  {Koutsioumpas}}, \bibinfo {author} {\bibfnamefont {M.}~\bibnamefont
  {Webster}}, \bibinfo {author} {\bibfnamefont {A.}~\bibnamefont {Rajput}},\
  and\ \bibinfo {author} {\bibfnamefont {D.~E.}\ \bibnamefont {Browne}},\
  }\bibfield  {title} {\bibinfo {title} {Fault-tolerant logical clifford gates
  from code automorphisms},\ }\href@noop {} {\bibfield  {journal} {\bibinfo
  {journal} {arXiv preprint arXiv:2409.18175}\ } (\bibinfo {year}
  {2024})}\BibitemShut {NoStop}%
\bibitem [{\citenamefont {Eberhardt}\ \emph {et~al.}(2024)\citenamefont
  {Eberhardt}, \citenamefont {Pereira},\ and\ \citenamefont
  {Steffan}}]{eberhardt2024pruning}%
  \BibitemOpen
  \bibfield  {author} {\bibinfo {author} {\bibfnamefont {J.~N.}\ \bibnamefont
  {Eberhardt}}, \bibinfo {author} {\bibfnamefont {F.~R.~F.}\ \bibnamefont
  {Pereira}},\ and\ \bibinfo {author} {\bibfnamefont {V.}~\bibnamefont
  {Steffan}},\ }\bibfield  {title} {\bibinfo {title} {Pruning qldpc codes:
  Towards bivariate bicycle codes with open boundary conditions},\ }\href@noop
  {} {\bibfield  {journal} {\bibinfo  {journal} {arXiv preprint
  arXiv:2412.04181}\ } (\bibinfo {year} {2024})}\BibitemShut {NoStop}%
\bibitem [{\citenamefont {Cross}\ \emph {et~al.}(2024)\citenamefont {Cross},
  \citenamefont {He}, \citenamefont {Rall},\ and\ \citenamefont
  {Yoder}}]{Cross24}%
  \BibitemOpen
  \bibfield  {author} {\bibinfo {author} {\bibfnamefont {A.}~\bibnamefont
  {Cross}}, \bibinfo {author} {\bibfnamefont {Z.}~\bibnamefont {He}}, \bibinfo
  {author} {\bibfnamefont {P.}~\bibnamefont {Rall}},\ and\ \bibinfo {author}
  {\bibfnamefont {T.}~\bibnamefont {Yoder}},\ }\bibfield  {title} {\bibinfo
  {title} {Improved qldpc surgery: Logical measurements and bridging codes},\
  }\href@noop {} {\bibfield  {journal} {\bibinfo  {journal} {arXiv preprint
  arXiv:2407.18393}\ } (\bibinfo {year} {2024})}\BibitemShut {NoStop}%
\bibitem [{\citenamefont {Cowtan}(2024)}]{Cowtan24}%
  \BibitemOpen
  \bibfield  {author} {\bibinfo {author} {\bibfnamefont {A.}~\bibnamefont
  {Cowtan}},\ }\bibfield  {title} {\bibinfo {title} {Ssip: automated surgery
  with quantum ldpc codes},\ }\href@noop {} {\bibfield  {journal} {\bibinfo
  {journal} {arXiv preprint arXiv:2407.09423}\ } (\bibinfo {year}
  {2024})}\BibitemShut {NoStop}%
\bibitem [{\citenamefont {Eberhardt}\ and\ \citenamefont
  {Steffan}(2025)}]{Eberhardt24}%
  \BibitemOpen
  \bibfield  {author} {\bibinfo {author} {\bibfnamefont {J.~N.}\ \bibnamefont
  {Eberhardt}}\ and\ \bibinfo {author} {\bibfnamefont {V.}~\bibnamefont
  {Steffan}},\ }\bibfield  {title} {\bibinfo {title} {Logical operators and
  fold-transversal gates of bivariate bicycle codes},\ }\href
  {https://doi.org/10.1109/TIT.2024.3521638} {\bibfield  {journal} {\bibinfo
  {journal} {IEEE Transactions on Information Theory}\ }\textbf {\bibinfo
  {volume} {71}},\ \bibinfo {pages} {1140} (\bibinfo {year}
  {2025})}\BibitemShut {NoStop}%
\bibitem [{\citenamefont {Williamson}\ and\ \citenamefont
  {Yoder}(2024)}]{Williamson24}%
  \BibitemOpen
  \bibfield  {author} {\bibinfo {author} {\bibfnamefont {D.~J.}\ \bibnamefont
  {Williamson}}\ and\ \bibinfo {author} {\bibfnamefont {T.~J.}\ \bibnamefont
  {Yoder}},\ }\bibfield  {title} {\bibinfo {title} {Low-overhead fault-tolerant
  quantum computation by gauging logical operators},\ }\href@noop {} {\bibfield
   {journal} {\bibinfo  {journal} {arXiv preprint arXiv:2410.02213}\ }
  (\bibinfo {year} {2024})}\BibitemShut {NoStop}%
\bibitem [{\citenamefont {He}\ \emph {et~al.}(2025)\citenamefont {He},
  \citenamefont {Cowtan}, \citenamefont {Williamson},\ and\ \citenamefont
  {Yoder}}]{He25Extractors}%
  \BibitemOpen
  \bibfield  {author} {\bibinfo {author} {\bibfnamefont {Z.}~\bibnamefont
  {He}}, \bibinfo {author} {\bibfnamefont {A.}~\bibnamefont {Cowtan}}, \bibinfo
  {author} {\bibfnamefont {D.~J.}\ \bibnamefont {Williamson}},\ and\ \bibinfo
  {author} {\bibfnamefont {T.~J.}\ \bibnamefont {Yoder}},\ }\bibfield  {title}
  {\bibinfo {title} {Extractors: Qldpc architectures for efficient pauli-based
  computation},\ }\href@noop {} {\bibfield  {journal} {\bibinfo  {journal}
  {arXiv preprint arXiv:2503.10390}\ } (\bibinfo {year} {2025})}\BibitemShut
  {NoStop}%
\bibitem [{\citenamefont {Haah}(2013{\natexlab{a}})}]{Haah13}%
  \BibitemOpen
  \bibfield  {author} {\bibinfo {author} {\bibfnamefont {J.}~\bibnamefont
  {Haah}},\ }\bibfield  {title} {\bibinfo {title} {Commuting pauli hamiltonians
  as maps between free modules},\ }\href
  {https://doi.org/10.1007/s00220-013-1810-2} {\bibfield  {journal} {\bibinfo
  {journal} {Communications in Mathematical Physics}\ }\textbf {\bibinfo
  {volume} {324}},\ \bibinfo {pages} {351} (\bibinfo {year}
  {2013}{\natexlab{a}})}\BibitemShut {NoStop}%
\bibitem [{SM1()}]{SM1}%
  \BibitemOpen
  \href@noop {} {}\bibinfo {note} {See Supplemental Material for details of the
  polinomial ring formalism, Gr\"obner bases and their explicit applications,
  the formal proof of BB codes' topological condition, the use of BKK theorem,
  the computation of anyon periods and mobility sublattice, the count of
  logical qubits, the formal proof of the anyon-logic duality, the relation of
  BB codes with other product constructions, and the subsystem-symmetry
  structure of BB codes, which contains additional examples.}\BibitemShut
  {Stop}%
\bibitem [{\citenamefont {Haah}(2013{\natexlab{b}})}]{Haah13_thesis}%
  \BibitemOpen
  \bibfield  {author} {\bibinfo {author} {\bibfnamefont {J.}~\bibnamefont
  {Haah}},\ }\bibfield  {title} {\bibinfo {title} {Lattice quantum codes and
  exotic topological phases of matter},\ }\href@noop {} {\bibfield  {journal}
  {\bibinfo  {journal} {arXiv preprint arXiv:1305.6973}\ } (\bibinfo {year}
  {2013}{\natexlab{b}})}\BibitemShut {NoStop}%
\bibitem [{\citenamefont {Haah}(2011)}]{Haah11}%
  \BibitemOpen
  \bibfield  {author} {\bibinfo {author} {\bibfnamefont {J.}~\bibnamefont
  {Haah}},\ }\bibfield  {title} {\bibinfo {title} {Local stabilizer codes in
  three dimensions without string logical operators},\ }\href
  {https://doi.org/10.1103/PhysRevA.83.042330} {\bibfield  {journal} {\bibinfo
  {journal} {Phys. Rev. A}\ }\textbf {\bibinfo {volume} {83}},\ \bibinfo
  {pages} {042330} (\bibinfo {year} {2011})}\BibitemShut {NoStop}%
\bibitem [{\citenamefont {Vijay}\ \emph {et~al.}(2016)\citenamefont {Vijay},
  \citenamefont {Haah},\ and\ \citenamefont {Fu}}]{Vijay16}%
  \BibitemOpen
  \bibfield  {author} {\bibinfo {author} {\bibfnamefont {S.}~\bibnamefont
  {Vijay}}, \bibinfo {author} {\bibfnamefont {J.}~\bibnamefont {Haah}},\ and\
  \bibinfo {author} {\bibfnamefont {L.}~\bibnamefont {Fu}},\ }\bibfield
  {title} {\bibinfo {title} {Fracton topological order, generalized lattice
  gauge theory, and duality},\ }\href
  {https://doi.org/10.1103/PhysRevB.94.235157} {\bibfield  {journal} {\bibinfo
  {journal} {Phys. Rev. B}\ }\textbf {\bibinfo {volume} {94}},\ \bibinfo
  {pages} {235157} (\bibinfo {year} {2016})}\BibitemShut {NoStop}%
\bibitem [{\citenamefont {Song}\ \emph {et~al.}(2024)\citenamefont {Song},
  \citenamefont {Tantivasadakarn}, \citenamefont {Shirley},\ and\ \citenamefont
  {Hermele}}]{Song24}%
  \BibitemOpen
  \bibfield  {author} {\bibinfo {author} {\bibfnamefont {H.}~\bibnamefont
  {Song}}, \bibinfo {author} {\bibfnamefont {N.}~\bibnamefont
  {Tantivasadakarn}}, \bibinfo {author} {\bibfnamefont {W.}~\bibnamefont
  {Shirley}},\ and\ \bibinfo {author} {\bibfnamefont {M.}~\bibnamefont
  {Hermele}},\ }\bibfield  {title} {\bibinfo {title} {Fracton
  self-statistics},\ }\href {https://doi.org/10.1103/PhysRevLett.132.016604}
  {\bibfield  {journal} {\bibinfo  {journal} {Phys. Rev. Lett.}\ }\textbf
  {\bibinfo {volume} {132}},\ \bibinfo {pages} {016604} (\bibinfo {year}
  {2024})}\BibitemShut {NoStop}%
\bibitem [{\citenamefont {Wen}(2002)}]{Wen02}%
  \BibitemOpen
  \bibfield  {author} {\bibinfo {author} {\bibfnamefont {X.-G.}\ \bibnamefont
  {Wen}},\ }\bibfield  {title} {\bibinfo {title} {Quantum orders and symmetric
  spin liquids},\ }\href {https://doi.org/10.1103/PhysRevB.65.165113}
  {\bibfield  {journal} {\bibinfo  {journal} {Phys. Rev. B}\ }\textbf {\bibinfo
  {volume} {65}},\ \bibinfo {pages} {165113} (\bibinfo {year}
  {2002})}\BibitemShut {NoStop}%
\bibitem [{\citenamefont {Essin}\ and\ \citenamefont
  {Hermele}(2013)}]{Essin2013}%
  \BibitemOpen
  \bibfield  {author} {\bibinfo {author} {\bibfnamefont {A.~M.}\ \bibnamefont
  {Essin}}\ and\ \bibinfo {author} {\bibfnamefont {M.}~\bibnamefont
  {Hermele}},\ }\bibfield  {title} {\bibinfo {title} {Classifying
  fractionalization: Symmetry classification of gapped $\mathbb{Z}_{2}$ spin
  liquids in two dimensions},\ }\href
  {https://doi.org/10.1103/PhysRevB.87.104406} {\bibfield  {journal} {\bibinfo
  {journal} {Phys. Rev. B}\ }\textbf {\bibinfo {volume} {87}},\ \bibinfo
  {pages} {104406} (\bibinfo {year} {2013})}\BibitemShut {NoStop}%
\bibitem [{\citenamefont {Mesaros}\ and\ \citenamefont
  {Ran}(2013)}]{Mesaros2013}%
  \BibitemOpen
  \bibfield  {author} {\bibinfo {author} {\bibfnamefont {A.}~\bibnamefont
  {Mesaros}}\ and\ \bibinfo {author} {\bibfnamefont {Y.}~\bibnamefont {Ran}},\
  }\bibfield  {title} {\bibinfo {title} {Classification of symmetry enriched
  topological phases with exactly solvable models},\ }\href
  {https://doi.org/10.1103/PhysRevB.87.155115} {\bibfield  {journal} {\bibinfo
  {journal} {Phys. Rev. B}\ }\textbf {\bibinfo {volume} {87}},\ \bibinfo
  {pages} {155115} (\bibinfo {year} {2013})}\BibitemShut {NoStop}%
\bibitem [{\citenamefont {Lu}\ and\ \citenamefont {Vishwanath}(2016)}]{Lu16}%
  \BibitemOpen
  \bibfield  {author} {\bibinfo {author} {\bibfnamefont {Y.-M.}\ \bibnamefont
  {Lu}}\ and\ \bibinfo {author} {\bibfnamefont {A.}~\bibnamefont
  {Vishwanath}},\ }\bibfield  {title} {\bibinfo {title} {Classification and
  properties of symmetry-enriched topological phases: Chern-simons approach
  with applications to ${Z}_{2}$ spin liquids},\ }\href
  {https://doi.org/10.1103/PhysRevB.93.155121} {\bibfield  {journal} {\bibinfo
  {journal} {Phys. Rev. B}\ }\textbf {\bibinfo {volume} {93}},\ \bibinfo
  {pages} {155121} (\bibinfo {year} {2016})}\BibitemShut {NoStop}%
\bibitem [{\citenamefont {Barkeshli}\ \emph {et~al.}(2019)\citenamefont
  {Barkeshli}, \citenamefont {Bonderson}, \citenamefont {Cheng},\ and\
  \citenamefont {Wang}}]{Barkeshli19}%
  \BibitemOpen
  \bibfield  {author} {\bibinfo {author} {\bibfnamefont {M.}~\bibnamefont
  {Barkeshli}}, \bibinfo {author} {\bibfnamefont {P.}~\bibnamefont
  {Bonderson}}, \bibinfo {author} {\bibfnamefont {M.}~\bibnamefont {Cheng}},\
  and\ \bibinfo {author} {\bibfnamefont {Z.}~\bibnamefont {Wang}},\ }\bibfield
  {title} {\bibinfo {title} {Symmetry fractionalization, defects, and gauging
  of topological phases},\ }\href {https://doi.org/10.1103/PhysRevB.100.115147}
  {\bibfield  {journal} {\bibinfo  {journal} {Phys. Rev. B}\ }\textbf {\bibinfo
  {volume} {100}},\ \bibinfo {pages} {115147} (\bibinfo {year}
  {2019})}\BibitemShut {NoStop}%
\bibitem [{\citenamefont {Barkeshli}\ \emph {et~al.}(2022)\citenamefont
  {Barkeshli}, \citenamefont {Chen}, \citenamefont {Hsin},\ and\ \citenamefont
  {Manjunath}}]{Barkeshli2022Classification}%
  \BibitemOpen
  \bibfield  {author} {\bibinfo {author} {\bibfnamefont {M.}~\bibnamefont
  {Barkeshli}}, \bibinfo {author} {\bibfnamefont {Y.-A.}\ \bibnamefont {Chen}},
  \bibinfo {author} {\bibfnamefont {P.-S.}\ \bibnamefont {Hsin}},\ and\
  \bibinfo {author} {\bibfnamefont {N.}~\bibnamefont {Manjunath}},\ }\bibfield
  {title} {\bibinfo {title} {Classification of $(2+1)$d invertible fermionic
  topological phases with symmetry},\ }\href
  {https://doi.org/10.1103/PhysRevB.105.235143} {\bibfield  {journal} {\bibinfo
   {journal} {Phys. Rev. B}\ }\textbf {\bibinfo {volume} {105}},\ \bibinfo
  {pages} {235143} (\bibinfo {year} {2022})}\BibitemShut {NoStop}%
\bibitem [{\citenamefont {Delfino}\ \emph {et~al.}(2023)\citenamefont
  {Delfino}, \citenamefont {Fontana}, \citenamefont {Gomes},\ and\
  \citenamefont {Chamon}}]{Delfino23}%
  \BibitemOpen
  \bibfield  {author} {\bibinfo {author} {\bibfnamefont {G.}~\bibnamefont
  {Delfino}}, \bibinfo {author} {\bibfnamefont {W.~B.}\ \bibnamefont
  {Fontana}}, \bibinfo {author} {\bibfnamefont {P.~R.~S.}\ \bibnamefont
  {Gomes}},\ and\ \bibinfo {author} {\bibfnamefont {C.}~\bibnamefont
  {Chamon}},\ }\bibfield  {title} {\bibinfo {title} {{Effective fractonic
  behavior in a two-dimensional exactly solvable spin liquid}},\ }\href
  {https://doi.org/10.21468/SciPostPhys.14.1.002} {\bibfield  {journal}
  {\bibinfo  {journal} {SciPost Phys.}\ }\textbf {\bibinfo {volume} {14}},\
  \bibinfo {pages} {002} (\bibinfo {year} {2023})}\BibitemShut {NoStop}%
\bibitem [{\citenamefont {Bombin}\ and\ \citenamefont
  {Martin-Delgado}(2009)}]{Bombin09}%
  \BibitemOpen
  \bibfield  {author} {\bibinfo {author} {\bibfnamefont {H.}~\bibnamefont
  {Bombin}}\ and\ \bibinfo {author} {\bibfnamefont {M.~A.}\ \bibnamefont
  {Martin-Delgado}},\ }\bibfield  {title} {\bibinfo {title} {Quantum
  measurements and gates by code deformation},\ }\href
  {https://doi.org/10.1088/1751-8113/42/9/095302} {\bibfield  {journal}
  {\bibinfo  {journal} {Journal of Physics A: Mathematical and Theoretical}\
  }\textbf {\bibinfo {volume} {42}},\ \bibinfo {pages} {095302} (\bibinfo
  {year} {2009})}\BibitemShut {NoStop}%
\bibitem [{\citenamefont {Krishna}\ and\ \citenamefont
  {Poulin}(2021)}]{Krishna21}%
  \BibitemOpen
  \bibfield  {author} {\bibinfo {author} {\bibfnamefont {A.}~\bibnamefont
  {Krishna}}\ and\ \bibinfo {author} {\bibfnamefont {D.}~\bibnamefont
  {Poulin}},\ }\bibfield  {title} {\bibinfo {title} {Fault-tolerant gates on
  hypergraph product codes},\ }\href
  {https://doi.org/10.1103/PhysRevX.11.011023} {\bibfield  {journal} {\bibinfo
  {journal} {Phys. Rev. X}\ }\textbf {\bibinfo {volume} {11}},\ \bibinfo
  {pages} {011023} (\bibinfo {year} {2021})}\BibitemShut {NoStop}%
\bibitem [{\citenamefont {Kesselring}\ \emph {et~al.}(2024)\citenamefont
  {Kesselring}, \citenamefont {Magdalena de~la Fuente}, \citenamefont
  {Thomsen}, \citenamefont {Eisert}, \citenamefont {Bartlett},\ and\
  \citenamefont {Brown}}]{Kesselring24}%
  \BibitemOpen
  \bibfield  {author} {\bibinfo {author} {\bibfnamefont {M.~S.}\ \bibnamefont
  {Kesselring}}, \bibinfo {author} {\bibfnamefont {J.~C.}\ \bibnamefont
  {Magdalena de~la Fuente}}, \bibinfo {author} {\bibfnamefont {F.}~\bibnamefont
  {Thomsen}}, \bibinfo {author} {\bibfnamefont {J.}~\bibnamefont {Eisert}},
  \bibinfo {author} {\bibfnamefont {S.~D.}\ \bibnamefont {Bartlett}},\ and\
  \bibinfo {author} {\bibfnamefont {B.~J.}\ \bibnamefont {Brown}},\ }\bibfield
  {title} {\bibinfo {title} {Anyon condensation and the color code},\ }\href
  {https://doi.org/10.1103/PRXQuantum.5.010342} {\bibfield  {journal} {\bibinfo
   {journal} {PRX Quantum}\ }\textbf {\bibinfo {volume} {5}},\ \bibinfo {pages}
  {010342} (\bibinfo {year} {2024})}\BibitemShut {NoStop}%
\bibitem [{\citenamefont {Wang}\ \emph {et~al.}(2003)\citenamefont {Wang},
  \citenamefont {Harrington},\ and\ \citenamefont {Preskill}}]{Wang03}%
  \BibitemOpen
  \bibfield  {author} {\bibinfo {author} {\bibfnamefont {C.}~\bibnamefont
  {Wang}}, \bibinfo {author} {\bibfnamefont {J.}~\bibnamefont {Harrington}},\
  and\ \bibinfo {author} {\bibfnamefont {J.}~\bibnamefont {Preskill}},\
  }\bibfield  {title} {\bibinfo {title} {Confinement-higgs transition in a
  disordered gauge theory and the accuracy threshold for quantum memory},\
  }\href {https://doi.org/https://doi.org/10.1016/S0003-4916(02)00019-2}
  {\bibfield  {journal} {\bibinfo  {journal} {Annals of Physics}\ }\textbf
  {\bibinfo {volume} {303}},\ \bibinfo {pages} {31} (\bibinfo {year}
  {2003})}\BibitemShut {NoStop}%
\bibitem [{\citenamefont {Chubb}\ and\ \citenamefont
  {Flammia}(2021)}]{Chubb21}%
  \BibitemOpen
  \bibfield  {author} {\bibinfo {author} {\bibfnamefont {C.~T.}\ \bibnamefont
  {Chubb}}\ and\ \bibinfo {author} {\bibfnamefont {S.~T.}\ \bibnamefont
  {Flammia}},\ }\bibfield  {title} {\bibinfo {title} {Statistical mechanical
  models for quantum codes with correlated noise},\ }\href@noop {} {\bibfield
  {journal} {\bibinfo  {journal} {Annales de l’Institut Henri Poincar{\'e}
  D}\ }\textbf {\bibinfo {volume} {8}},\ \bibinfo {pages} {269} (\bibinfo
  {year} {2021})}\BibitemShut {NoStop}%
\bibitem [{\citenamefont {Song}\ \emph {et~al.}(2022)\citenamefont {Song},
  \citenamefont {Sch\"onmeier-Kromer}, \citenamefont {Liu}, \citenamefont
  {Viyuela}, \citenamefont {Pollet},\ and\ \citenamefont
  {Martin-Delgado}}]{Song22}%
  \BibitemOpen
  \bibfield  {author} {\bibinfo {author} {\bibfnamefont {H.}~\bibnamefont
  {Song}}, \bibinfo {author} {\bibfnamefont {J.}~\bibnamefont
  {Sch\"onmeier-Kromer}}, \bibinfo {author} {\bibfnamefont {K.}~\bibnamefont
  {Liu}}, \bibinfo {author} {\bibfnamefont {O.}~\bibnamefont {Viyuela}},
  \bibinfo {author} {\bibfnamefont {L.}~\bibnamefont {Pollet}},\ and\ \bibinfo
  {author} {\bibfnamefont {M.~A.}\ \bibnamefont {Martin-Delgado}},\ }\bibfield
  {title} {\bibinfo {title} {Optimal thresholds for fracton codes and random
  spin models with subsystem symmetry},\ }\href
  {https://doi.org/10.1103/PhysRevLett.129.230502} {\bibfield  {journal}
  {\bibinfo  {journal} {Phys. Rev. Lett.}\ }\textbf {\bibinfo {volume} {129}},\
  \bibinfo {pages} {230502} (\bibinfo {year} {2022})}\BibitemShut {NoStop}%
\bibitem [{\citenamefont {Nishimori}(2007)}]{Nishimori07}%
  \BibitemOpen
  \bibfield  {author} {\bibinfo {author} {\bibfnamefont {H.}~\bibnamefont
  {Nishimori}},\ }\bibfield  {title} {\bibinfo {title} {Duality in
  finite-dimensional spin glasses},\ }\href
  {https://doi.org/10.1007/s10955-006-9156-1} {\bibfield  {journal} {\bibinfo
  {journal} {Journal of Statistical Physics}\ }\textbf {\bibinfo {volume}
  {126}},\ \bibinfo {pages} {977} (\bibinfo {year} {2007})}\BibitemShut
  {NoStop}%
\bibitem [{\citenamefont {Adams}\ and\ \citenamefont
  {Loustaunau}(1994)}]{adams94}%
  \BibitemOpen
  \bibfield  {author} {\bibinfo {author} {\bibfnamefont {W.~W.}\ \bibnamefont
  {Adams}}\ and\ \bibinfo {author} {\bibfnamefont {P.}~\bibnamefont
  {Loustaunau}},\ }\href@noop {} {\emph {\bibinfo {title} {An introduction to
  Grobner bases}}},\ \bibinfo {number} {3}\ (\bibinfo  {publisher} {American
  Mathematical Soc.},\ \bibinfo {year} {1994})\BibitemShut {NoStop}%
\bibitem [{\citenamefont {Grayson}\ and\ \citenamefont {Stillman}()}]{M2}%
  \BibitemOpen
  \bibfield  {author} {\bibinfo {author} {\bibfnamefont {D.~R.}\ \bibnamefont
  {Grayson}}\ and\ \bibinfo {author} {\bibfnamefont {M.~E.}\ \bibnamefont
  {Stillman}},\ }\href@noop {} {\bibinfo {title} {Macaulay2, a software system
  for research in algebraic geometry}},\ \bibinfo {howpublished} {Available at
  \url{http://www2.macaulay2.com}}\BibitemShut {NoStop}%
\bibitem [{\citenamefont {Decker}\ \emph {et~al.}(2024)\citenamefont {Decker},
  \citenamefont {Greuel}, \citenamefont {Pfister},\ and\ \citenamefont
  {Sch\"onemann}}]{DGPS}%
  \BibitemOpen
  \bibfield  {author} {\bibinfo {author} {\bibfnamefont {W.}~\bibnamefont
  {Decker}}, \bibinfo {author} {\bibfnamefont {G.-M.}\ \bibnamefont {Greuel}},
  \bibinfo {author} {\bibfnamefont {G.}~\bibnamefont {Pfister}},\ and\ \bibinfo
  {author} {\bibfnamefont {H.}~\bibnamefont {Sch\"onemann}},\ }\href@noop {}
  {\bibinfo {title} {{\sc Singular} {4-4-0} --- {A} computer algebra system for
  polynomial computations}},\ \bibinfo {howpublished}
  {\url{http://www.singular.uni-kl.de}} (\bibinfo {year} {2024})\BibitemShut
  {NoStop}%
\bibitem [{\citenamefont {Mondal}(2021)}]{mondal21}%
  \BibitemOpen
  \bibfield  {author} {\bibinfo {author} {\bibfnamefont {P.}~\bibnamefont
  {Mondal}},\ }\href@noop {} {\emph {\bibinfo {title} {How Many Zeroes?:
  Counting Solutions of Systems of Polynomials Via Toric Geometry at
  Infinity}}},\ Vol.~\bibinfo {volume} {2}\ (\bibinfo  {publisher} {Springer
  Nature},\ \bibinfo {year} {2021})\BibitemShut {NoStop}%
\bibitem [{\citenamefont {Berlekamp}(1984)}]{ber84}%
  \BibitemOpen
  \bibfield  {author} {\bibinfo {author} {\bibfnamefont {E.~R.}\ \bibnamefont
  {Berlekamp}},\ }\href {https://doi.org/10.1142/9407} {\emph {\bibinfo {title}
  {Algebraic coding theory}}},\ \bibinfo {edition} {revised edition}\ ed.\
  (\bibinfo  {publisher} {Aegean Park Press},\ \bibinfo {address} {Laguna
  Hills},\ \bibinfo {year} {1984})\ Chap.~\bibinfo {chapter} {6}\BibitemShut
  {NoStop}%
\bibitem [{\citenamefont {Kreuzer}\ and\ \citenamefont
  {Walsh}(2024)}]{kreuzer24}%
  \BibitemOpen
  \bibfield  {author} {\bibinfo {author} {\bibfnamefont {M.}~\bibnamefont
  {Kreuzer}}\ and\ \bibinfo {author} {\bibfnamefont {F.}~\bibnamefont
  {Walsh}},\ }\bibfield  {title} {\bibinfo {title} {Computing the binomial part
  of a polynomial ideal},\ }\href {https://doi.org/10.1016/j.jsc.2024.102298}
  {\bibfield  {journal} {\bibinfo  {journal} {Journal of Symbolic Computation}\
  }\textbf {\bibinfo {volume} {124}},\ \bibinfo {pages} {102298} (\bibinfo
  {year} {2024})}\BibitemShut {NoStop}%
\bibitem [{\citenamefont {Tan}\ \emph {et~al.}(2023)\citenamefont {Tan},
  \citenamefont {Roberts}, \citenamefont {Tantivasadakarn}, \citenamefont
  {Yoshida},\ and\ \citenamefont {Yao}}]{tan2023fracton}%
  \BibitemOpen
  \bibfield  {author} {\bibinfo {author} {\bibfnamefont {Y.}~\bibnamefont
  {Tan}}, \bibinfo {author} {\bibfnamefont {B.}~\bibnamefont {Roberts}},
  \bibinfo {author} {\bibfnamefont {N.}~\bibnamefont {Tantivasadakarn}},
  \bibinfo {author} {\bibfnamefont {B.}~\bibnamefont {Yoshida}},\ and\ \bibinfo
  {author} {\bibfnamefont {N.~Y.}\ \bibnamefont {Yao}},\ }\bibfield  {title}
  {\bibinfo {title} {Fracton models from product codes},\ }\href@noop {}
  {\bibfield  {journal} {\bibinfo  {journal} {arXiv preprint arXiv:2312.08462}\
  } (\bibinfo {year} {2023})}\BibitemShut {NoStop}%
\bibitem [{\citenamefont {Sala}\ \emph {et~al.}(2022)\citenamefont {Sala},
  \citenamefont {Lehmann}, \citenamefont {Rakovszky},\ and\ \citenamefont
  {Pollmann}}]{Sala22}%
  \BibitemOpen
  \bibfield  {author} {\bibinfo {author} {\bibfnamefont {P.}~\bibnamefont
  {Sala}}, \bibinfo {author} {\bibfnamefont {J.}~\bibnamefont {Lehmann}},
  \bibinfo {author} {\bibfnamefont {T.}~\bibnamefont {Rakovszky}},\ and\
  \bibinfo {author} {\bibfnamefont {F.}~\bibnamefont {Pollmann}},\ }\bibfield
  {title} {\bibinfo {title} {Dynamics in systems with modulated symmetries},\
  }\href {https://doi.org/10.1103/PhysRevLett.129.170601} {\bibfield  {journal}
  {\bibinfo  {journal} {Phys. Rev. Lett.}\ }\textbf {\bibinfo {volume} {129}},\
  \bibinfo {pages} {170601} (\bibinfo {year} {2022})}\BibitemShut {NoStop}%
\bibitem [{\citenamefont {Delfino}\ and\ \citenamefont
  {You}(2024)}]{Delfino2024}%
  \BibitemOpen
  \bibfield  {author} {\bibinfo {author} {\bibfnamefont {G.}~\bibnamefont
  {Delfino}}\ and\ \bibinfo {author} {\bibfnamefont {Y.}~\bibnamefont {You}},\
  }\bibfield  {title} {\bibinfo {title} {Anyon condensation web and
  multipartite entanglement in two-dimensional modulated gauge theories},\
  }\bibfield  {journal} {\bibinfo  {journal} {Physical Review B}\ }\textbf
  {\bibinfo {volume} {109}},\ \href
  {https://doi.org/10.1103/physrevb.109.205146} {10.1103/physrevb.109.205146}
  (\bibinfo {year} {2024})\BibitemShut {NoStop}%
\bibitem [{\citenamefont {Canossa}\ \emph {et~al.}(2024)\citenamefont
  {Canossa}, \citenamefont {Pollet}, \citenamefont {Martin-Delgado},
  \citenamefont {Song},\ and\ \citenamefont {Liu}}]{Canossa24}%
  \BibitemOpen
  \bibfield  {author} {\bibinfo {author} {\bibfnamefont {G.}~\bibnamefont
  {Canossa}}, \bibinfo {author} {\bibfnamefont {L.}~\bibnamefont {Pollet}},
  \bibinfo {author} {\bibfnamefont {M.~A.}\ \bibnamefont {Martin-Delgado}},
  \bibinfo {author} {\bibfnamefont {H.}~\bibnamefont {Song}},\ and\ \bibinfo
  {author} {\bibfnamefont {K.}~\bibnamefont {Liu}},\ }\bibfield  {title}
  {\bibinfo {title} {Exotic symmetry breaking properties of self-dual fracton
  spin models},\ }\href {https://doi.org/10.1103/PhysRevResearch.6.013304}
  {\bibfield  {journal} {\bibinfo  {journal} {Phys. Rev. Res.}\ }\textbf
  {\bibinfo {volume} {6}},\ \bibinfo {pages} {013304} (\bibinfo {year}
  {2024})}\BibitemShut {NoStop}%
\end{thebibliography}%
	
	\onecolumngrid
\clearpage
\makeatletter
\begin{center}
	\textbf{\large --- Supplemental Material ---\\[0.5em]Anyon Theory and Topological Frustration of High-Efficiency 
		\\Quantum Low-Density Parity-Check Codes}\\[1em]
	
	Keyang Chen$^{1,2,3}$, Yuanting Liu$^{1,3}$, Yiming Zhang$^{2,4}$, Zijian Liang$^{5}$, Yu-An Chen$^{5}$, Ke Liu$^{2,4}$ and Hao Song$^{1}$
	
	\vspace{5pt}
	
	{\centering \emph{$^1$Institute of Theoretical Physics, Chinese Academy of Sciences, Beijing 100190, China}}
	
	{\centering \emph{$^2$Hefei National Research Center for Physical Sciences at the Microscale and School of Physical Sciences, \\ University of Science and Technology of China, Hefei 230026, China}}
	
	{\centering \emph{$^3$School of Physical Sciences, University of Chinese Academy of Sciences, Beijing 100049, China.}}
	
	{\centering \emph{$^4$Shanghai Research Center for Quantum Science and CAS Center for Excellence in Quantum Information and Quantum Physics, University of Science and Technology of China, Shanghai 201315, China}}
	
	{\centering \emph{$^5$International Center for Quantum Materials, School of Physics, Peking University, Beijing 100871, China}}

	\thispagestyle{titlepage}
\end{center}
\setcounter{equation}{0}
\setcounter{figure}{0}
\setcounter{table}{0}
\setcounter{page}{1}
\setcounter{section}{0}
\renewcommand{\theequation}{S\arabic{equation}}
\renewcommand{\thefigure}{S\arabic{figure}}
\renewcommand{\thetable}{S\arabic{table}}
\renewcommand{\thesection}{S.\Roman{section}}

\section{Polynomial formalism for BB codes}\label{app:ring}

Additional details about the polynomial formalism, including (\ref{subsec:Para}) parameter reduction for BB codes via coordinate transformations,
(\ref{subsec:GSDtorus}) derivation of Eq.~\eqref{eq:GSDtorus}, and (\ref{subsec:ibm_code}) transformation of codes from Ref.~\cite{Bravyi24} to the convention of Eq.~\eqref{eq:generator}. 

\subsection{Parameter Reductions via Coordinate Transformations}\label{subsec:Para}

Here, we justify that the BB codes $\mathbb{BB}(\overline{\alpha},\overline{\beta},a,b)$
can be restricted to $\overline{\alpha}=-\alpha$ and $\overline{\beta}=-\beta$
with $\alpha\geq\beta\geq0$, without loss of generality, by showing that other parameter ranges can
be mapped from the current one by combining the following coordinate transformations. 

\paragraph{(1) Shifting coordinate origins. }

Shifting the choice of coordinate origins can lead to all transformations
of the form
\begin{equation}
	\mathbb{BB}\left(f,g\right)\rightarrow\mathbb{BB}\left(\xi f,\zeta g\right),\label{eq:shift_origin}
\end{equation}
where $\xi$ and $\zeta$ are mononials. Here and throughout, $\mathbb{BB}\left(f,g\right)$
denotes the code Hamiltonian specified by polynomials $f$ and $g$.

For simplicity, we demonstrate this by the toric code on a square lattice
(depicted in Fig.~\ref{fig:S1}). Each unit cell contains \emph{two}
qubits, corresponding to the horizontal and vertical links. Thus,
the system can be viewed as a composite of \emph{two} lattices of
qubits (illustrated as red and green circles, respectively). For each
qubit lattice, a monomial $x^{i}y^{j}$ denotes the vertex with position
$\left(i,j\right)$ relative to its origin. With respective to the
choice of origins (indicated by filled circles) in Fig.~\ref{fig:S1}~(a)
and (b),
\begin{equation}
	h_{X}={1+x \choose 1+y}\quad\text{and}\quad h_{Z}={1+y^{-1} \choose 1+x^{-1}}
\end{equation}
represent the vertex and plaquette operators. A different choice of
origins is used in Fig.~\ref{fig:S1}(c), with the origin of the
green lattice shifted by $\lambda$, such as $\lambda=x^{2}$ for
illustration. The representation of vertex and plaquette operators
turns into
\begin{equation}
	h_{X}^{\prime}={\lambda\left(1+x\right) \choose 1+y}\quad\text{and}\quad h_{Z}^{\prime}={1+y^{-1} \choose \lambda^{-1}\left(1+x^{-1}\right)}.
\end{equation}
The absolute positions of $h_{X}^{\prime}$ and $h_{Z}^{\prime}$
can also be shifted without affecting the stabilizer group; in particular,
\begin{equation}
	h_{X}^{\prime}\rightarrow\zeta h_{X}^{\prime}={\xi\left(1+x\right) \choose \zeta\left(1+y\right)}\quad\text{and}\quad h_{Z}^{\prime}\rightarrow\zeta^{-1}h_{Z}^{\prime}={\xi^{-1}(1+y^{-1}) \choose \zeta^{-1}\left(1+x^{-1}\right)},
\end{equation}
where $\xi=\lambda\zeta$ and $\zeta$ are monomials. The discussion
applies to all BB codes, leading to transformation in Eq.~\eqref{eq:shift_origin}. 

\begin{figure}[h]
	\includegraphics[width=1\columnwidth]{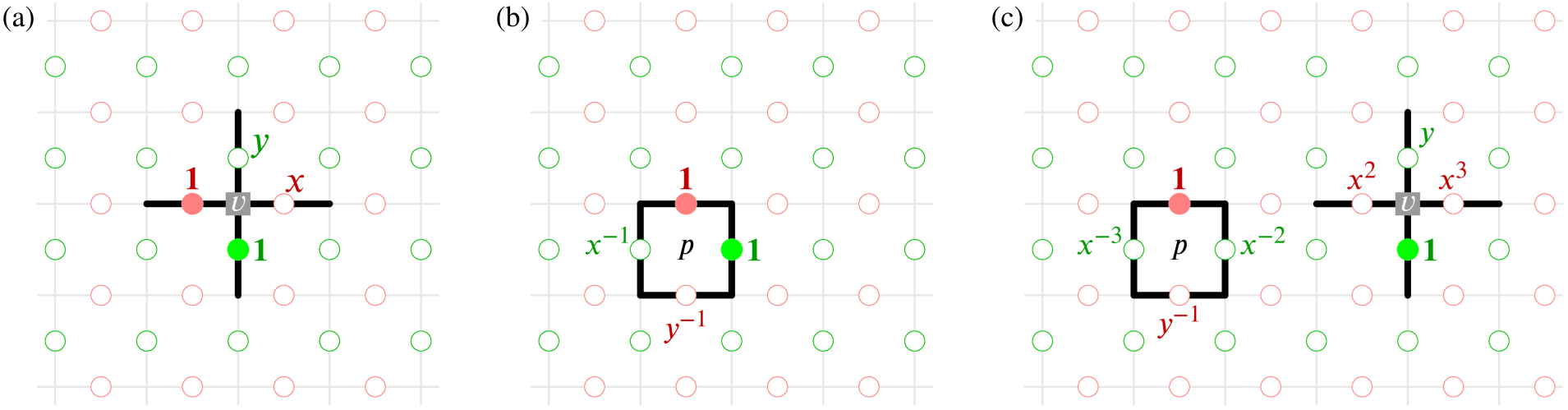}
	
	\caption{Transformation of the polynomial representation for the toric code
		under a coordinate origin shift. In (a) and (b), the toric code is
		represented by polynomials $f=1+x$ and $g=1+y$. In (c), shifting
		the origin of the green lattice by $\lambda=x^{2}$ transforms the
		ring representation from $\left(f,g\right)$ to $\left(\lambda f,g\right)$.}
	
	\label{fig:S1}
\end{figure}

\paragraph{(2) Redefining axis directions.}

Reversing the choices of positive $x$ and $y$ directions induces
the transformations
\begin{equation}
	\mathbb{BB}\left(f(x,y),g(x,y)\right)\rightarrow\mathbb{BB}(f(x^{-1},y),g(x^{-1},y))\quad\text{and}\quad\mathbb{BB}\left(f(x,y),g(x,y)\right)\rightarrow\mathbb{BB}(f(x,y^{-1}),g(x,y^{-1}))\label{eq:axis}
\end{equation}
in the polynomial representation, respectively.

These transformations, together with proper origin shifts, relates
parameters $-\alpha$ and $-\beta$ with $1+\alpha$ and $1+\beta$,
respectively: 
\[
\mathbb{BB}(-\alpha,-\beta,a,b)\sim\mathbb{BB}\left(1+\alpha,-\beta,-a,b\right)\sim\mathbb{BB}\left(-\alpha,1+\beta,a,-b\right)\sim\mathbb{BB}\left(1+\alpha,1+\beta,-a,-b\right).
\]
Eventually, combining the freedom to interchange which axis is labeled
$x$ and which is labeled $y$ allows us to require $\alpha\geq\beta\ge0$
without loss of generality. 

\subsection{Proof of Eq.~\eqref{eq:GSDtorus} } \label{subsec:GSDtorus}

In the main text, the expression for $k$ given in Eq.~\eqref{eq:GSDtorus} is introduced
by analogy, which does not constitute a formal proof. In this subsection,
we provide a rigorous derivation to establish its validity.

The ground state degeneracy (GSD) of the code Hamiltonian on a torus
can be calculated directly using the formula
\begin{equation}
	\text{GSD}=\text{tr}\left(\prod_{v}\frac{1+\hat{h}_{X,v}}{2}\prod_{p}\frac{1+\hat{h}_{Z,p}}{2}\right),
\end{equation}
which evaluates the trace of the projection operator that selects
the ground state Hilbert subspace. For computational purposes, we
use the basis of the full Hilbert space defined by $\left|\psi\right\rangle \coloneqq X\left(\psi\right)\left|0\right\rangle $,
where $\psi\in R_{\ell,m}^{2}$ with $R_{\ell,m}=R/(x^{\ell}-1,y^{m}-1)$
labels any collection of qubits in the system. Here, $X\left(\psi\right)$
represents the product of Pauli operators associated with these qubits,
and $\left|0\right\rangle $ denotes the tensor product of the $Z=+1$
states for all qubits. 

The factor $\prod_{p}\frac{1}{2}(1+\hat{h}_{Z,p})$ ensures that only
states $\left|\psi\right\rangle $ with $\psi\in\ker h_{Z}^{\dagger}$
contribute to the GSD, simplifying the expression to 
\begin{equation}
	\text{GSD}=\sum_{\psi\in\ker h_{Z}^{\dagger}}\left\langle \psi\right|\prod_{v}\frac{1+\hat{h}_{X,v}}{2}\left|\psi\right\rangle .
\end{equation}
For notational clarity, here we use $h_{Z}$ and $h_{X}$ to denote the linear maps from $R_{\ell,m}$
to $R_{\ell,m}^{2}$, and we use $\hat{h}_{Z,p}$ and $\hat{h}_{X,v}$ to denote the corresponding check operators.

Next, we observe that 
\begin{equation}
	\prod_{v}\frac{1}{2}(1+\hat{h}_{X,v})=\frac{1}{|R_{\ell,m}|}\sum_{\omega\in R_{\ell,m}}X\left(h_{X}\omega\right),
\end{equation}
where $|R_{\ell,m}|$ relates to the number of vertices $\left|V\right|$
by $|R_{\ell,m}|=2^{|V|}$. Only terms with $h_{X}\omega=0$ contribute
to the GSD, leading to
\begin{equation}
	\text{GSD}=\sum_{\psi\in\ker h_{Z}^{\dagger}}\frac{1}{\left|R\right|}\left\langle \psi\right|\sum_{\omega\in R_{\ell,m}}X\left(h_{X}\omega\right)\left|\psi\right\rangle =\frac{|\ker h_{Z}^{\dagger}|\cdot\left|\ker h_{X}\right|}{\left|R_{\ell,m}\right|}.
\end{equation}

Moreover, the linear maps $h_{X}:R_{\ell,m}\rightarrow R_{\ell,m}^{2}$
and $h_{Z}^{\dagger}:R_{\ell,m}^{2}\rightarrow R_{\ell,m}$ give rise
to ismorphisms 
\begin{equation}
	\text{im}\,h_{X}\cong\frac{R_{\ell,m}}{\ker h_{X}}\quad\text{and}\quad\text{im}\,h_{Z}^{\dagger}\cong\frac{R_{\ell,m}^{2}}{\ker h_{Z}^{\dagger}}.
\end{equation}
Consequently, we have $\left|\ker h_{X}\right|=|R_{\ell,m}|/\left|\text{im}\,h_{X}\right|$
and $|\ker h_{Z}^{\dagger}|=|R_{\ell,m}|^{2}/|\text{im}\,h_{Z}^{\dagger}|$.
Substituting these into the GSD expression yields
\begin{equation}
	\text{GSD}=\frac{|R_{\ell,m}|^{2}}{\left|\text{im}\,h_{X}\right||\text{im}\,h_{Z}^{\dagger}|}.
\end{equation}

Finally, we note that the matrices $h_{X}={f \choose g}$ and its
transpose $h_{X}^{T}=\left(f\ g\right)$ have the same rank, namely,
\begin{equation}
	\dim_{\mathbb{F}_{2}}\text{im}\,h_{X}=\dim_{\mathbb{F}_{2}}\text{im}\,h_{X}^{T}.
\end{equation}
Since the images of $h_{X}^{T}$ and $h_{Z}^{\dagger}=\left(g\ f\right)$
both result in the ideal of $R_{\ell,m}$ generated by $f$ and $g$,
we obtain
\begin{equation}
	\text{GSD}=\left|\frac{R_{\ell,m}}{\left(f,g\right)}\right|^{2}.
\end{equation}
Expressed in terms of the number of logical qubits, this gives 
\begin{equation}
	k(\ell,m)=\log_{2}\mathrm{GSD}=2\dim_{\mathbb{F}_{2}}\frac{R_{\ell,m}}{\left(f,g\right)}=2\dim_{\mathbb{F}_{2}}\frac{R}{\left(f,g,x^{\ell}-1,x^{m}-1\right)},
\end{equation}
which corresponds to Eq.~\eqref{eq:GSDtorus} in the main text.

\subsection{Relation to the codes in Ref.~\cite{Bravyi24}}\label{subsec:ibm_code}
The authors of Ref.~\cite{Bravyi24} examine BB codes with toric layouts, highlighting their suitability for implementation on superconducting hardware. They establish a sufficient condition for a code to admit a toric layout, formalized in the following lemma:
\begin{lem}[Lemma 4 in Ref.~\cite{Bravyi24}]
	Consider a code defined by three-term generating polynomials:
	\begin{equation}
		f = f_1+f_2+f_3, \quad g = g_1+g_2+g_3.
	\end{equation}
	where each $f_i$, $g_i$ is a power of $x$ or $y$. This code possesses a toric layout if there exist indices $i$, $j$, $r$, $s\in \{1,2,3\}$, such that
	\begin{enumerate}
		\item $\langle f_if_j^{-1}, g_rg_s^{-1} \rangle = \mathbb{Z}_\ell \otimes \mathbb{Z}_m$ and 
		\item ${\rm ord}(f_if_j^{-1}){\rm ord}(g_rg_s^{-1})=\ell m$.
	\end{enumerate}
\end{lem}

One can easily verify that the ansatz in Eq.~\eqref{eq:generator} in the main text satisfies these conditions by choosing $f_i = x$, $f_j = 1$, $g_r = y$, and $g_s = 1$.

As an illustrative example, consider the $[[288,12,18]]$ code from Ref.~\cite{Bravyi24}, generated by $A = x^3+y^{2}+y^{7}$, $B = y^3+x+x^{2}$ with $(\ell,m)=(12,12)$. Applying the group automorphism $y \mapsto y^{7}$, these two polynomials transform as
\begin{align}
	A &\mapsto x^{3} + y^2 + y = y\left(1+y + y^{-1}x^{3}\right), \\
	B &\mapsto y^{-3} + x + x^{2} = x\left(1+x + x^{-1}y^{-3}\right),
\end{align}
where the boundary conditions $x^{12}=y^{12}=1$ and the origin shift defined in Eq.~\eqref{eq:shift_origin} are applied. So, this code corresponds to the $\mathbb{BB}(\overline{1},\overline{1},3,\overline{3})$ code family in our convention.

Notably, all codes proposed in Ref.~\cite{Bravyi24} can be systematically mapped to the form of Eq.~\eqref{eq:generator} via an appropriate automorphism and origin shift. The specific transformations for these codes are detailed in Table~\ref{tab:ibm_code}.

\begin{table}[h]
	\begin{tabular}{|c|c|c|c|c|c|c|}
		\hline
		[[$n,k,d$]] & $\ell, m$ & $A$ & $B$ & automorphism & $f$ & $g$ \\ \hline
		$[[72,12,6]]$ & $6,6$ & $x^3+y+y^2$ & $y^3+x+x^2$ & * &$1+x+x^{-1}y^3$ & $1+y+y^{-1}x^3$ \\ \hline
		$[[90,8,10]]$ & $15,3$ & $x^9+y+y^2$ & $1+x^2+x^7$ & $x\mapsto x^8$ & $1+x+x^{-4}$ & $1+y+y^{-1}x^{-3}$\\ \hline
		$[[108,8,10]]$ & $9,6$ & \multirow{2}{*}{$x^3+y+y^2$} & \multirow{2}{*}{$y^3+x+x^2$} & \multirow{2}{*}{*} & \multirow{2}{*}{$1+x+x^{-1}y^3$} & \multirow{2}{*}{$1+y+y^{-1}x^3$} \\ \cline{1-2}
		$[[144,12,12]]$ & $12,6$ &  &  &  &  &  \\ \hline
		$[[288,12,18]]$& $12,12$ & $x^3+y^2+y^7$& $y^3+x+x^2$ & $y \mapsto y^7$ & $1+x + x^{-1}y^{-3}$ & $1+y + y^{-1}x^{3}$\\ \hline
		$[[360,12,\leq 24]]$& 30,6 & $x^9+y+y^2$ & $y^3+x^{25}+x^{26}$ & * & $1+x+x^{-25}y^3$ & $1+y+y^{-1}x^9$\\ \hline
		$[[756,16,\leq 34]]$& 21,18 & $x^3+y^{10}+y^{17}$ & $y^5+x^3+x^{19}$ & $x\mapsto x^4$, $y \mapsto y^{13}$ & $1+x + x^{-12}y^{11}$ & $1+y + y^{-4}x^{12}$\\ \hline
	\end{tabular}
	\caption{Transformations for BB codes with three-term generators proposed in Ref.~\cite{Bravyi24}. The table summarizes automorphisms and origin shifts required to map each code to the form in Eq.~\eqref{eq:generator}.}\label{tab:ibm_code}
\end{table}

\section{Gr\"{o}bner basis and their applications}
A Gr\"obner Basis is a fundamental concept in computational algebraic geometry and commutative algebra, serving as a powerful tool for solving systems of polynomial equations and performing algebraic computations in multivariate polynomial rings. This mathematical framework is instrumental in various fields, including optimization, cryptography, and automated theorem proving. 

To determine whether a given polynomial $f$ belongs to the ideal $\mathcal{I}=(f_1,f_2,\cdots, f_n)$ (we restrict our discussion to cases where the coefficient ring is a field), the task can be effectively addressed in the single-variable case. Specifically, one can divide $f$ by the greatest common divisor (gcd) of the polynomials $f_1,f_2,\cdots, f_n$. If the remainder is found to be zero, this indicates that $f$ is indeed an element of the ideal. The computation of the $\gcd$ is efficiently facilitated by the well-established Euclidean algorithm. 

In contrast, the situation becomes more complex when extending to multivariate polynomials. Unlike univariate polynomial rings, multivariate polynomial rings are not \textit{principal ideal domains}, which means that an ideal $\mathcal{I}$ cannot always be generated by a single element. Furthermore, even upon identifying a set of generators $\mathcal{G}$ for $\mathcal{I}$, the remainder of a polynomial $f$ with respect to $\mathcal{G}$ may not be unique and can vary depending on the division process employed. For instance, consider the ideal $\mathcal{I}=(yx+1, y^2+x)$ in the polynomial ring $\mathbb{F}_2[x,y]$, and let $f=y^2x+yx+y+x^3$, with lexicographic order $y>x$. If we first divide $f$ by $yx+1$ and then by $y^2+x$, the remainder obtained is $r_1=x^3+1$. Conversely, if we divide $f$ first by $y^2+x$ and then by $yx+1$, the remainder is $r_2 = y + x^3 + x^2 + 1$, with $r_1 \neq r_2$. 

The Gr\"obner basis in the context of multivariate polynomials serves a role analogous to that of the gcd in the univariate case. In brief, given a non-zero polynomial ideal $\mathcal{I}$, there always exists a Gr\"obner basis $\mathcal{G}=\{g_1,g_2,\cdots,g_m\}$ with respect to a certain monomial order, such that the remainder obtained from the division of $f$ by $\mathcal{G}$ is unique. Moreover, the remainder is zero if and only if $f$ is an element of $\mathcal{I}$. For the case above, it is noteworthy that $f$ is indeed an element of the ideal $\mathcal{I}$. To verify this, we calculate the Gr\"obner basis for $\mathcal{I}$, which is given by $\mathcal{G} = \{x^3+1, y+x^2\}$. Upon dividing $f$  by the Gr\"obner basis $\mathcal{G}$, we find that $f = y \cdot (x^3+1)+(yx+x) \cdot (y+x^2)$, which yields a remainder of zero. For a rigorous description, we refer the reader to \cite{adams94}.

The Gr\"obner basis is unique under certain conditions and can be computed using the celebrated Buchberger's algorithm, as well as several of its improvements. This algorithm is implemented in many computer algebra systems, including $\texttt{Macaulay2}$ \cite{M2}, $\texttt{Singular}$ \cite{DGPS}, and Mathematica, among others.

In exploring the application of Gr\"obner bases to characterize topological orders in BB codes, it is essential to consider both the construction of the codes and the boundary conditions within the context of polynomial rings and ideals. Specifically, for a system with infinitely open boundary conditions, the base ring is the Laurent polynomial ring in polynomial formalism. However, a direct computation of a Gr\"obner basis on this ring is not feasible. We use the trick $\mathbb{F}_2[x^\pm, y^\pm] \cong \mathbb{F}_2[x,\overline{x}, y, \overline{y}] / (x\overline{x}-1, y\overline{y}-1)$,  introducing two auxiliary variables $\overline{x}$ and $\overline{y}$ to replace $x^{-1}$ and $y^{-1}$. This substitution is subject to constraints $x\overline{x}=1$ and $y\overline{y}=1$, respectively, resulting in the ideals $\mathcal{I}_\infty = (f,g,x\overline{x}-1, y\overline{y}-1)$. Here, $f$ and $g$ should be interpreted as projections onto the ordinary polynomial ring.  For systems with periodic boundary conditions of sizes $\ell\times m$, periodic constraints are imposed such that $x^\ell = 1$ and $y^m=1$, leading to the ideals $\mathcal{I}_{\ell,m} = (f,g,x^\ell-1, y^m-1)$. We present the Gr\"obner basis for the toric code, the $\mathbb{BB}(\overline{1},\overline{1},3,3)$ code and the $\mathbb{BB}(\overline{1},\overline{1},\overline{3},3)$ code in Table.~\ref{table:groebner_basis}. For finite sizes,  non-trivial Gr\"obner basis arises only for specific values of $\ell$ and $m$. The non-trivial size sequences are discussed in detail in references \ref{app:anyon_periods} and \ref{app:logical_qubit}.

To further demonstrate the applications, we examine quasi-fractonic hopping dynamics in one-dimensional. In the case of $\mathbb{BB}(\overline{1}, \overline{1}, 3,\overline{3})$ code, the motion of $h_Z$-excitations
along the $x$ direction is governed by the first Gr\"obner basis $g = x^6+x^5+x^3+x+1$. Two $h_Z$-excitation configurations $\xi_1$ and $\xi_2$ are said to be related by a hopping process if
\begin{equation}
	\xi_1 \equiv \xi_2 \pmod g.
\end{equation}
To illustrate the hopping dynamics, we iteratively generate a sequence of patterns via the recurrence relation 
\begin{equation}
	\xi(t_{n+1}) = \xi(t_n) + \text{lowest\_term}(\xi(t_n))\cdot g.
\end{equation} 
This relation ensures that all $\xi(t_n)$ share the same charge. Starting from an initial configuration $\xi(t_0)=1$ (a single $h_Z$-excitation), the iterative computation yields the following configurations:
\begin{align}
	\xi(t_0) &= 1, \\
	\xi(t_1) &= x^6+x^5+x^3+x, \\
	\xi(t_2) &= x^7+x^5+x^4+x^3+x^2, \\
	\xi(t_3) &= x^8+x^4, \\
	\xi(t_4) &= x^{10} + x^9 + x^8 + x^7+x^5, \\
	\xi(t_5) &= x^{11} + x^9 + x^7 + x^6, \\
	\xi(t_6) &= x^{12}.
\end{align}
After six iterations, the anyon charge returns to its initial state, having translated 12 lattice units along the $x$-direction. Configurations $\xi(t_0)$, $\xi(t_1)$, $\xi(t_3)$, and $\xi(t_6)$ are illustrated in Fig.~3(d) to visualize the spatiotemporal evolution of the excitation patterns.

\begin{table}[]
	\begin{tabular}{clcl|c|l|}
		\hline
		\multicolumn{1}{|c|}{$\quad\gcd(\ell,m)\quad$} & \multicolumn{1}{c|}{Toric Code} & \multicolumn{1}{c|}{$\quad\gcd(\ell,m)\quad$} & \multicolumn{1}{c|}{$\mathbb{BB}(\overline{1}, \overline{1},3,3)$} & $\gcd(\ell,m)$ & \multicolumn{1}{c|}{$\mathbb{BB}(\overline{1},\overline{1}, 3, \overline{3})$} \\ \hline
		\multicolumn{1}{|c|}{\multirow{5}{*}{$\infty$}} & \multicolumn{1}{l|}{$x+1$} & \multicolumn{1}{c|}{\multirow{5}{*}{$\infty$}} & $x^6+x^5+x^3+x+1$ & \multirow{5}{*}{$\infty$} & $x^{11}+x^{10}+x^9+x^7+x^6+x^4+1$ \\
		\multicolumn{1}{|c|}{} & \multicolumn{1}{l|}{$y+1$} & \multicolumn{1}{c|}{} & $x^2y+xy+y+x^5+x+1$ &  & $x^2y+xy+y+x^8+x^7+x^5+x^4$ \\
		\multicolumn{1}{|c|}{} & \multicolumn{1}{l|}{$\overline{x}+1$} & \multicolumn{1}{c|}{} & $y^2+y+x^3$ &  & $y^2+y+x^3$ \\
		\multicolumn{1}{|c|}{} & \multicolumn{1}{l|}{$\overline{y}+1$} & \multicolumn{1}{c|}{} & $\overline{x}+x^5+x^4+x^2+1$ &  & $\overline{x}+x^{10}+x^9+x^8+x^6+x^5+x^3$ \\
		\multicolumn{1}{|c|}{} & \multicolumn{1}{l|}{} & \multicolumn{1}{c|}{} & $\overline{y}+y+x^5+x^4+x^3+1$ &  & $\overline{y}+y+x^8+x^7$ \\ \hline
		\multicolumn{1}{|c|}{\multirow{3}{*}{$\mathbb{Z}$}} & \multicolumn{1}{l|}{$x+1$} & \multicolumn{1}{c|}{\multirow{3}{*}{$12\mathbb{Z}$}} & $x^6+x^5+x^3+x+1$ & \multirow{3}{*}{$762\mathbb{Z}$} & $x^{11}+x^{10}+x^9+x^7+x^6+x^4+1$ \\
		\multicolumn{1}{|c|}{} & \multicolumn{1}{l|}{$y+1$} & \multicolumn{1}{c|}{} & $x^2y+xy+y+x^5+x+1$ &  & $x^2y+xy+y+x^8+x^7+x^5+x^4$ \\
		\multicolumn{1}{|c|}{} & \multicolumn{1}{l|}{} & \multicolumn{1}{c|}{} & $y^2+y+x^3$ &  & $y^2+y+x^3$ \\ \hline
		& \multicolumn{1}{l|}{} & \multicolumn{1}{c|}{\multirow{3}{*}{$6\mathbb{Z}-12\mathbb{Z}$}} & $x^4+x^2+1$ & \multirow{3}{*}{$381\mathbb{Z}-762\mathbb{Z}$} & $x^9+x^5+x^3+x+1$ \\
		& \multicolumn{1}{l|}{} & \multicolumn{1}{c|}{} & $x^2y+xy+y+x^3+1$ &  & $x^2y+xy+y+x^8+x^7+x^5+x^4$ \\
		& \multicolumn{1}{l|}{} & \multicolumn{1}{c|}{} & $y^2+y+x^3$ &  & $y^2+y+x^3$ \\ \cline{3-6} 
		& \multicolumn{1}{l|}{} & \multicolumn{1}{c|}{\multirow{2}{*}{$3\mathbb{Z}-6\mathbb{Z}$}} & $x^2+x+1$ & \multirow{2}{*}{$127\mathbb{Z}-381\mathbb{Z}$} & $x^7+x^6+x^4+x^2+1$ \\
		& \multicolumn{1}{l|}{} & \multicolumn{1}{c|}{} & $y^2+y+1$ &  & $y+x^6+x^4$ \\ \cline{3-6} 
		&  &  &  & \multirow{3}{*}{$6\mathbb{Z}-762\mathbb{Z}$} & $x^4+x^2+1$ \\
		&  &  &  &  & $x^2y+xy+y+x^3+1$ \\
		&  &  &  &  & $y^2+y+x^3$ \\ \cline{5-6} 
		&  &  &  & \multirow{2}{*}{$3\mathbb{Z}-6\mathbb{Z}-381\mathbb{Z}$} & $x^2+x+1$ \\
		&  &  &  &  & $y^2+y+1$ \\ \cline{5-6} 
	\end{tabular}
	\caption{Gr\"obner basis for toric code, $\mathbb{BB}(\overline{1},\overline{1},3,3)$ and $\mathbb{BB}(\overline{1},\overline{1},3,\overline{3})$ codes on infinite plane and non-trivial size sequences. The monomial orders are lexicographic order with $\overline{y}>\overline{x}>y>x$ for infinite case and $y>x$ for finite case.}\label{table:groebner_basis}
\end{table}

\section{Topological condition for BB codes}
Let $R$ be a commutative ring. A pair $(f,g)$ of elements in $R$ constitutes a \textit{regular sequence} if $\mathrm{Ann}_R(f)=0$ and $\mathrm{Ann}_{R/(f)}(g)=0$. Our initial aim is to establish the correspondence between the regular sequence and exactness. 
\begin{lem}
	Consider the Koszul complex related to the CSS code given by
	\begin{equation}
		R\xrightarrow{\;h_X={f \choose -g}\;}R^{2}\xrightarrow{h_Z^\dagger=\left(g,f\right)}R.
	\end{equation}
	There is a vector space isomorphism of the first homology group
	\begin{equation}
		H := \frac{\ker h_Z^\dagger}{\im h_X} \cong \frac{\mathrm{Ann}_R(f)}{(g)} \oplus \mathrm{Ann}_{\frac{R}{(f)}}(g).
	\end{equation}
\end{lem}

\begin{proof}
	Consider the map
	\begin{equation}
		\varphi: \quad  H \to \frac{R}{(f)}, \quad \;{u \choose v}\; + \;{f \choose -g}R\; \mapsto u + fR.
	\end{equation}
	It is well-defined. An element $\;{u \choose v}\;\in R^2$ belongs to $\ker h_Z^\dagger$ if and only if $ug+vf=0$, which implies $u \in \mathrm{Ann}_{R/(f)}(g)$. Consequently, the image of $\varphi$ is $\mathrm{Ann}_{R/(f)}(g)$. Moreover, elements of $\ker\varphi$ can be expressed as $\;{0 \choose v}\;+\;{f \choose -g}R$. Therefore, we have \begin{equation}
		\ker\varphi = \left\{\;{0 \choose v}+{f \choose -g}R\; \bigg|vg=0\right\} \cong \frac{\mathrm{Ann}(f)}{(g)}.
	\end{equation} 
	Consequently, we establish the vector space isomorphism:
	\begin{equation}
		H \cong \frac{\mathrm{Ann}_R(f)}{(g)} \oplus \mathrm{Ann}_{\frac{R}{(f)}}(g).
	\end{equation}
\end{proof}
We can conclude with the corollary as follows:
\begin{cor}\label{thm:topo_condition}
	If $(f,g)$ is a regular sequence, then this chain complex is exact, and hence the associated CSS code is topological.
\end{cor}

In the context of our BB code setting presented in the main text, it is evident that $\mathrm{Ann}_R(f)=0$. For any $[v] \in \mathrm{\Ann}_{R/(f)}(g)$, the relation $vg = uf$ holds for some $u\in R$. Assuming $[v]\neq [0]$, it follows that $f$ and $g$ share common factors. From this observation, it is clear that in order to show $\mathrm{\Ann}_{R/(f)}(g)=0$, it suffices to demonstrate that $f$ and $g$ share no common factors.
\begin{lem}\label{lem:Eisenstein}
	Let $R$ be an integral domain and consider $f(x)=a_n x^n+\cdots a_1 x + a_0 \in R[x]$. Suppose there exists a prime ideal $\mathfrak{p}$ such that 
	\begin{enumerate}
		\item $a_0, a_1, \cdots, a_{n-1} \in \mathfrak{p}$;
		\item $a_n \notin \mathfrak{p}$;
		\item $a_0 \notin \mathfrak{p}^2$.
	\end{enumerate}
	Then $f(x)$ is irreducible. Here, irreducible means that $f(x)$ itself is not a unit in $R[x]$ and $f(x)$ cannot be factored into the product of two non-unit polynomials in $R[x]$. 
\end{lem}

Before proving the lemma, we state some preliminary observations:
\begin{rem}
	This is a generalization of Eisenstein's criterion, extending the coefficient ring from the integers to a  integral domain.
\end{rem}
\begin{rem}
	An irreducible polynomial $f(x)$ is defined as one that is neither a unit in $R[x]$ nor expressible as a product of two non-unit polynomials within $R[x]$. A polynomial is irreducible in $R[x,x^{-1}]$ if it is irreducible in $R[x]$, but a polynomial that is reducible in $R[x]$ may not be reducible in $R[x, x^{-1}]$. As an illustration, consider the polynomial $f(x) = yx+x^2=x(x+y)$.  It is reducible in the polynomial ring $\mathbb{F}_2[y,y^{-1}][x]$, yet it remains irreducible in the Laurent polynomial ring $\mathbb{F}_2[y,y^{-1}][x,x^{-1}]$ because $x$ is a unit in the latter.
\end{rem}
\begin{rem}
	When considering elements in Laurent polynomial rings $R[x,x^{-1}]$, it is always feasible to multiply by a factor $x^i$ to convert the expression into polynomial form.
\end{rem}
\begin{proof}
	Prove by contradiction. Assume the existence of a factorization $f(x) = b(x)c(x)$, where $b(x) = b_r x^r + \cdots + b_1x + b_0$ and $c(x)=c_s x^s + \cdots c_1 x + c_0$. Modulo $\mathfrak{p}$ on both sides given $\bar{a}_n x^n = (\bar{b}_r x^r+\cdots+\bar{b}_0)(\bar{c}_s x^s + \cdots +\bar{c}_0)$, where the bar denotes the equivalence class in $R/\mathfrak{p}$. Since the leading term gives $a_nx^n=b_rx^r\cdot c_s x^s$, implies that $b_0$ and $c_0$ are both in $\mathfrak{p}$, so $a_0=b_0c_0 \in \mathfrak{p}^2$. This contradicts the third condition, asserting that $f(x)$ cannot be factored as proposed.
\end{proof}

For BB codes under the toric layout, we can verify that $f =1+x+ x^{\overline{\alpha}}y^b \in \mathbb{F}_2[x,x^{-1}][y,y^{-1}]$ is irreducible by choosing the prime ideal $\mathfrak{p}=(x+1)$. Similarly, $g$ is also irreducible. Thus, if $\mathrm{Ann}_{R/(f)}(g)\ne 0$, the only possibility is $f = x^iy^jg$ for some monomial $x^iy^j$, given the parameters $(\overline{\alpha},\overline{\beta},a,b)=(0,0,1,1)$.  For all other parameter configurations, we have $\mathrm{Ann}_{R/(f)}(g)=0$, indicating that $(f,g)$ forms a regular sequence. Consequently, according to Corollary~\ref{thm:topo_condition}, we arrive at the core theorem.

\begin{thm}
	Following the convention in the main text, all BB codes are topological for arbitrary $\alpha,\beta\in\mathbb{Z}_{\geq0}$ and $a,b\in\mathbb{Z}$, with the single exception of $(\overline{\alpha}, \overline{\beta}, a, b) = (0,0,1,1)$.
\end{thm}

\section{Topological index and the BKK theorem}
The quantity $Q$ enumerates the zeros, accounting for their multiplicities, of the system $f=g=0$ over an algebraic closure of $\mathbb{F}_2$ in the framework of Laurent polynomial rings \cite{Haah13}. This can be analyzed effectively using the Bernstein-Khovanskii-Kushnirenko (BKK) theorem, which establishes that the number of solutions in Laurent polynomial rings is bounded by the mixed volume of certain convex polytopes. In this section, we (1) introduce the BKK theorem and its application in calculating the topological index and (2) derive explicit expressions for the topological index when the BKK conditions fail.

\subsection{BKK Theorem and BKK Condition}
Consider a Laurent polynomial $h$ in variables $x_1,\dots,x_d$, given by $h=\sum_{a\in\Lambda}c_a x^a$, where $\Lambda\subset \mathbb{R}^d$ and $x^a$ stands for monomial $\prod_{i=1}^d x_i^{a_i}$. The set $\Lambda$ is referred to as the $\textit{support}$ of $h$, denoted as $\mathop{supp}(h)$. The \textit{Newton polytope} of $h$, denoted $\mathop{NP}(h)$,  is the convex hull of the support of $h$ in $\mathbb{R}^d$. For example, the Newton polytope of $f=1+x+x^{\overline{\alpha}}y^b$ is a triangle with vertices at $(0,0)$, $(1,0)$ and $(-\alpha,b)$, while the Newton polytope of $g=1+y+x^a y^{\overline{\beta}}$ is a triangle with vertices at $(0,0)$, $(0,1)$ and $(a,-\beta)$.

Let $\mathcal{P}_1,\dots,\mathcal{P}_n$ be convex polytopes in $\mathbb{R}^d$. The \textit{Minkowski sum} of these polytopes is defined as
\begin{equation}
	\mathcal{P} = \sum_{i\in I}\mathcal{P}_i = \left\{\sum_{i\in I} v_i \big| v_i \in \mathcal{P}_i\right\} \subset \mathbb{R}^d.
\end{equation}

It can be shown that $\mathcal{P}$ is a convex set. For each surface $\mathcal{S}$ of $\mathcal{P}$, there are unique faces $\mathcal{S}_i$ of $\mathcal{P}_i$ such that $\mathcal{S} = \sum_{i=1}^n\mathcal{S}_j$. Let $h_{i,\mathcal{S}}=\sum_{a\in\mathcal{S}_i}c_{i,a}x^a$ be the component of $h_i$ supported at $\mathcal{S}_i$.

The \textit{mixed volume} of the polytopes is given by
\begin{equation}
	\mathop{MV}(\mathcal{P}_1,\dots,\mathcal{P}_n) = \sum_{I \subset \{1,2,\dots,n\}}(-1)^{n-|I|}\mathop{Vol_d}\left(\sum_{i\in I}P_i\right).
\end{equation}
Here, $\mathop{Vol_d}(\cdot)$ denotes the Lebesgue measure in $\mathbb{R}^d$. Specifically, for $n=2$ and $d=2$, we obtain
\begin{equation}\label{eq:mixed_volume}
	\mathop{MV}(\mathcal{P}_1,\mathcal{P}_2) = \mathop{Area}(\mathcal{P}_1+\mathcal{P}_2) - \mathop{Area}(\mathcal{P}_1)-\mathop{Area}(\mathcal{P}_2).
\end{equation}

\begin{thm}[BKK theorem]
	Let $\mathbb{K}$ be an algebraically closed field, and let $h_i \in \mathbb{K}[x_1,x_1^{-1}, \dots, x_d,x_d^{-1}]$ for $i=1,2,\dots,n$. Suppose that $Q = \dim_{\mathbb{K}}\frac{\mathbb{K}[x_1,x_1^{-1}, \dots, x_d,x_d^{-1}]}{(h_1,\dots,h_n)}$ is finite. Then
	\begin{equation}\label{eq:BBK}
		Q \leq \mathop{MV}\left(\mathop{NP}(h_1), \dots,\mathop{NP}(h_n)\right).
	\end{equation}
	Moreover, if the right-hand side is not zero, this bound is saturated if and only if, for each face $\mathcal{S}$ of dimension less than $n$, the Laurent polynomials $h_{i,\mathcal{S}}$ for $i=1,\dots,n$  have no common root. 
\end{thm}

For the proof, we refer interested readers to Ref.~\cite{mondal21} for interested readers. We examine a specific case where $a,b \geq 0$ and $ab>(\alpha+1)(\beta+1)$ to demonstrate the utility of the BKK theorem in calculating the topological index, as illustrated in Table~\ref{table:newton_polytope}. The red triangles represent the Newton polytopes of $f$, and the green triangles represent the Newton polytopes of $g$. The hexagon $ABCDEF$ represents the Minkowski sum of these polytopes, by definition. The mixed volume is derived from the hexagon substrates two triangles $ABC$ and $AEF$ (Eq.~\eqref{eq:mixed_volume}), yielding the area of the parallelogram $ACDF$, which is determined as $|\overrightarrow{DC}\times \overrightarrow{DE}|=ab-\alpha\beta$. It can be confirmed that the conditions of the BKK theorem are satisfied in this particular case.

When the condition $ab=(\alpha+1)(\beta+1)$ holds, the lines $BA$ and $AF$ in Table~\ref{table:newton_polytope} coincide. Under these parameters, the mixed volume remains $ab-\alpha\beta$; however, the equality in Eq.~\eqref{eq:BBK} may not hold. Let us investigate the BKK condition on the one-dimensional surface $BF$, whose defining components are $f_{BF}=x+x^{-\alpha}y^b$ and $g_{BF} = y+x^a y^{-\beta}$. Depending on the specific values of $a,b,\alpha$, and $\beta$, the Laurent polynomials $f_{BF}$ and $g_{BF}$ may share common roots.
\begin{table}
	\begin{tabular}{cccc}
		\toprule[1.5pt]
		\multicolumn{4}{c}{$a\geq 0$, $b\geq 0$} \\
		\hline
		$ab>(\alpha+1)(\beta+1)$ & $ab=(\alpha+1)(\beta+1)$ & $\alpha\beta<ab<(\alpha+1)(\beta+1)$ & $ab=\alpha\beta$ \\
		$\includegraphics[width=.21\columnwidth]{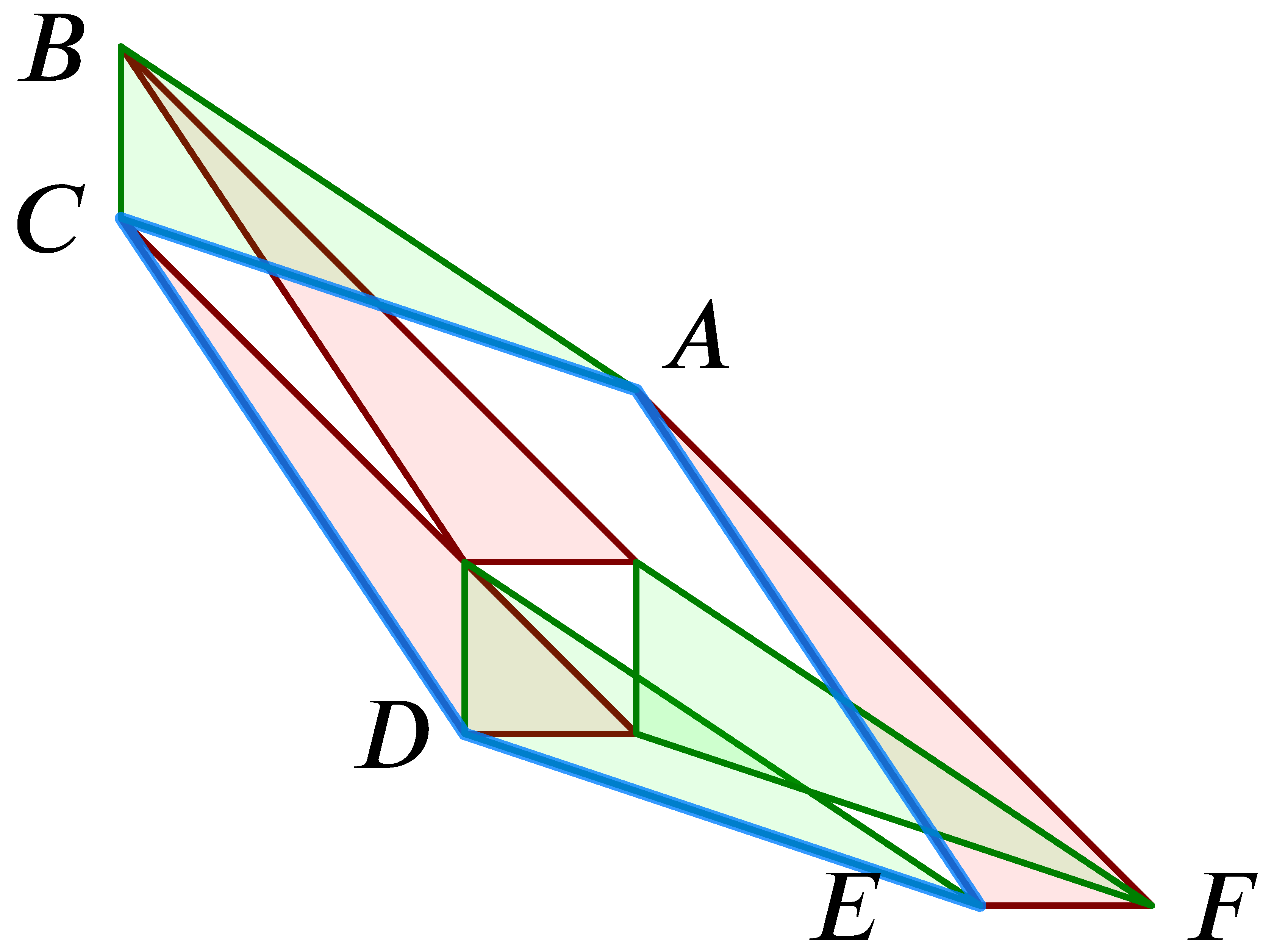}$ &
		$\includegraphics[width=.24\columnwidth]{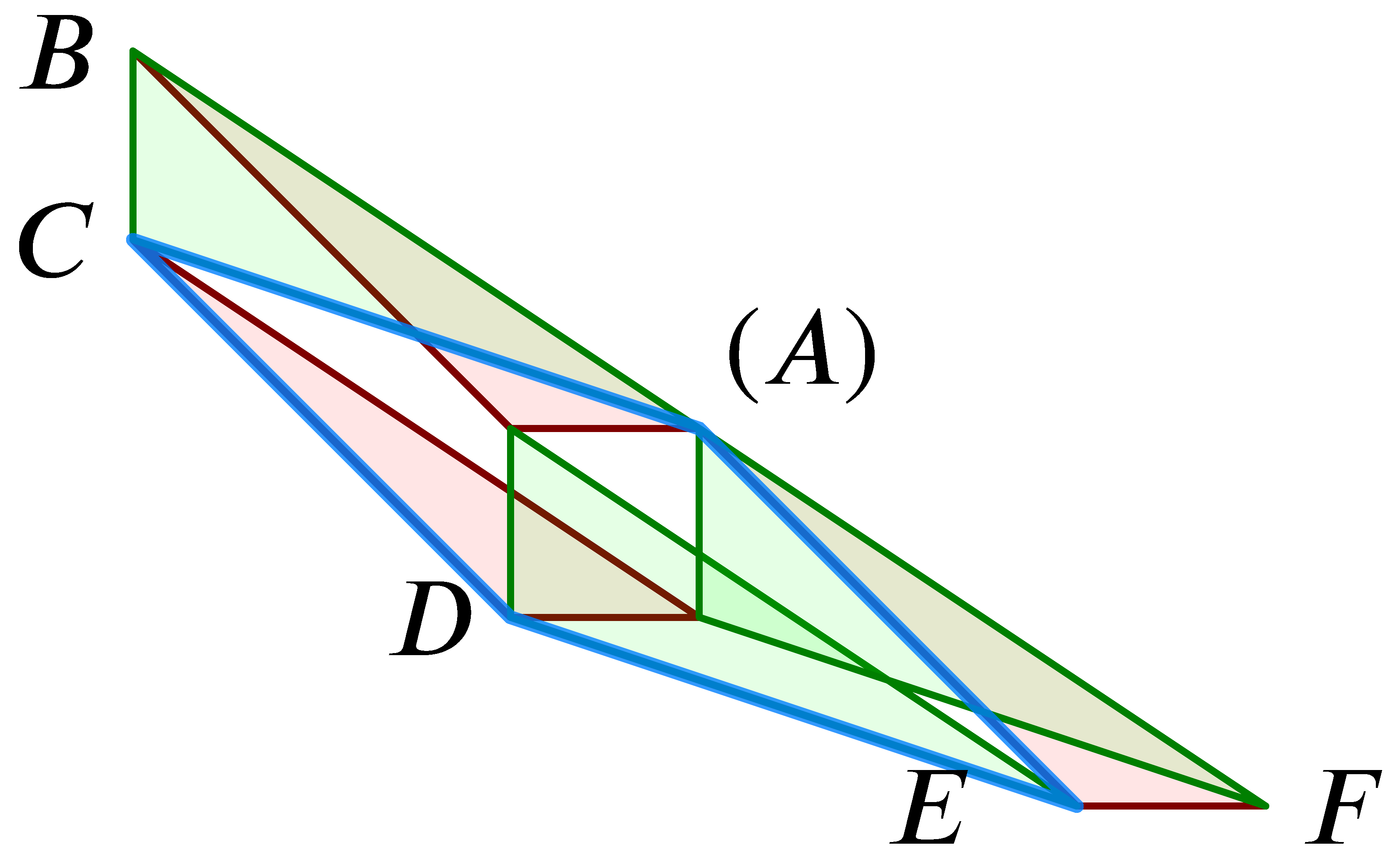}$ &
		$\includegraphics[width=.21\columnwidth]{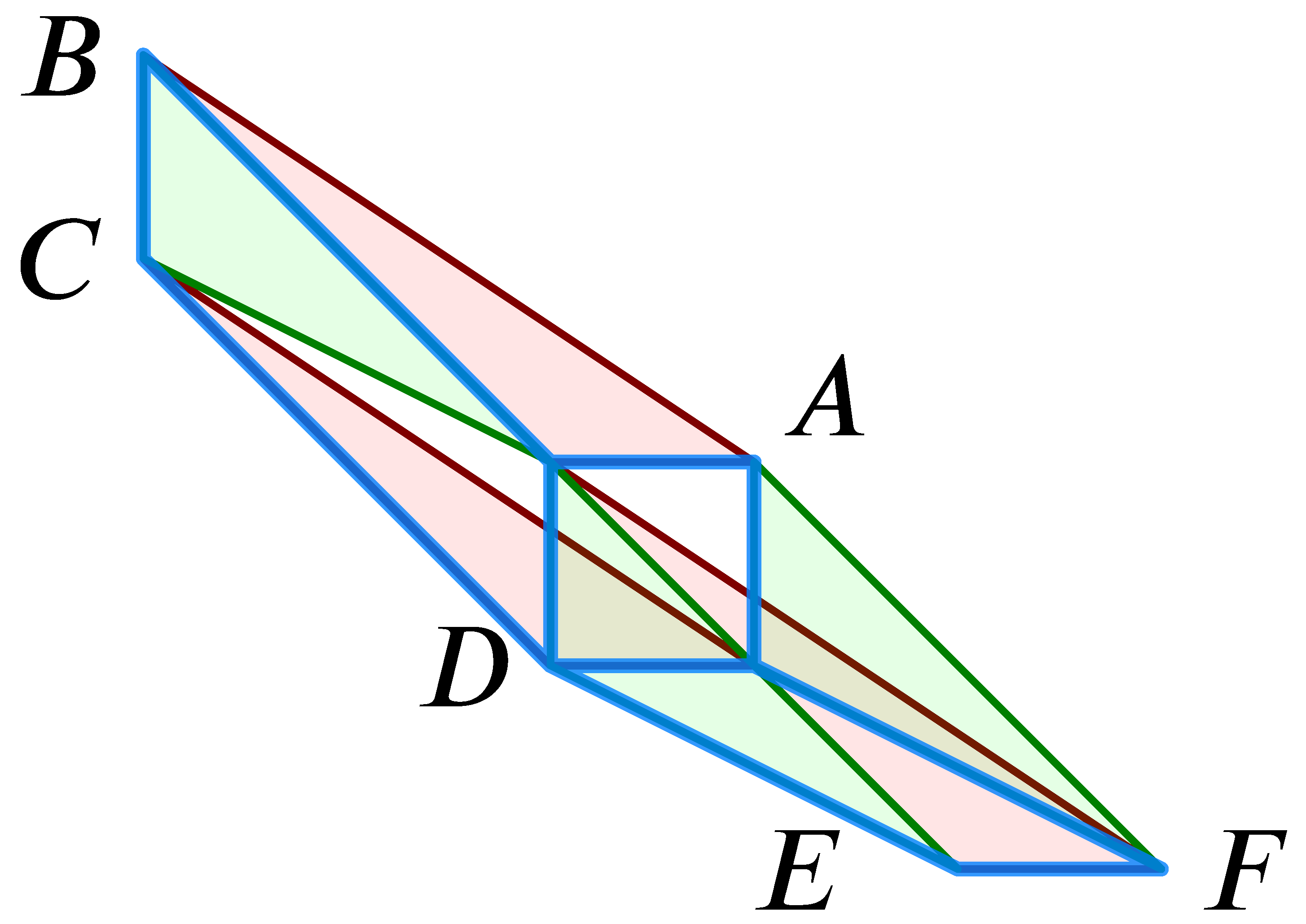}$ &
		$\includegraphics[width=.21\columnwidth]{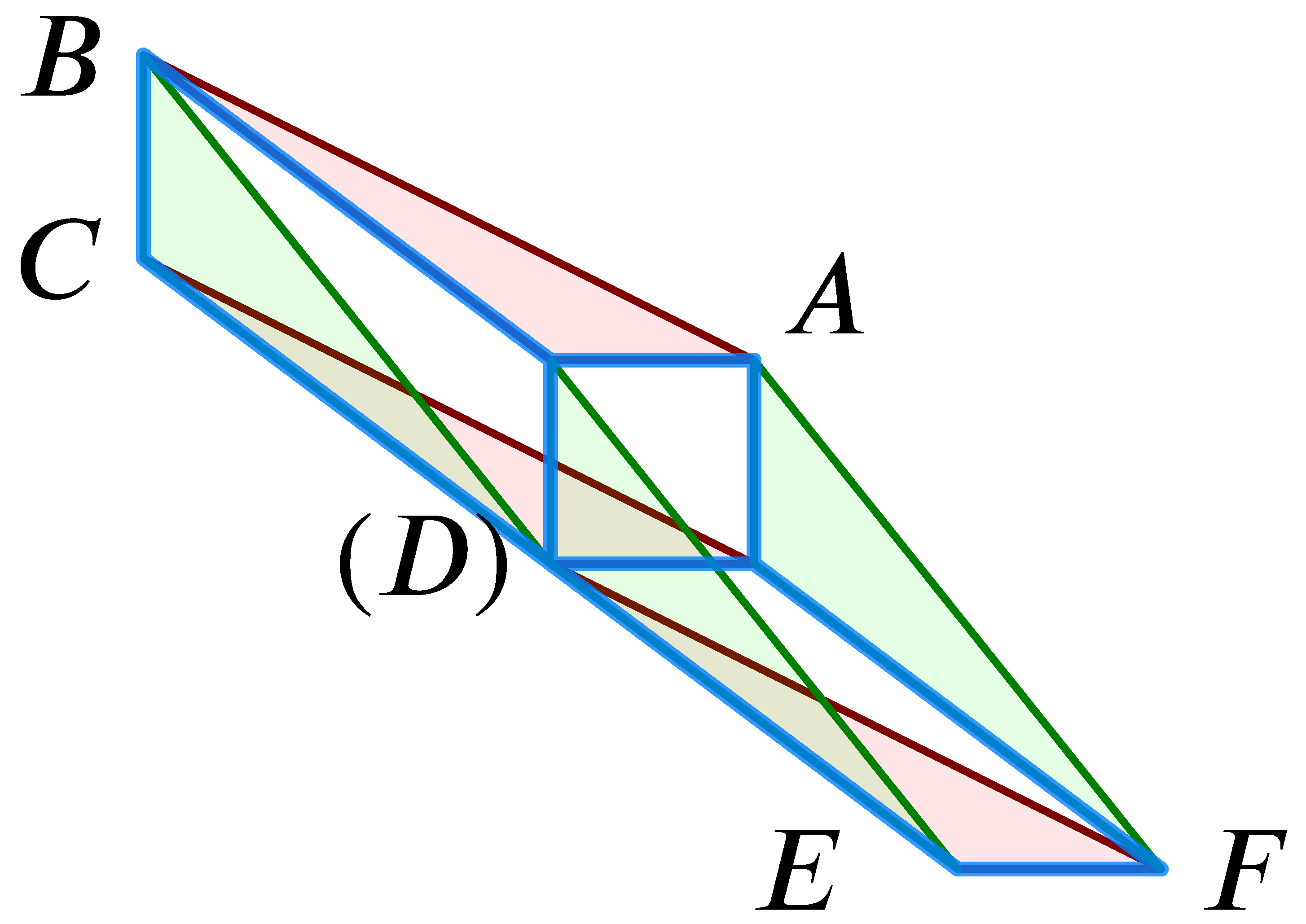}$ \\
		$ab-\alpha\beta$ & $\text{Deg}\left\{\varphi_{\alpha+1,a}(u)+(1+u)^\alpha\varphi_{\beta+1,b}(u)\right\}$ & $\alpha+\beta+1$ & $\text{Deg}\left\{\varphi_{\alpha,a}(u)+(1+u)^{\alpha+1}\varphi_{\beta,b}(u)\right\}$ \\
		\hline\hline
		$a\geq 0,\ b\geq0$ & $a>0, b<0$ & $a<0, b>0$ & $a \leq 0,\ b\leq 0$ \\
		\hline
		$ab<\alpha\beta$ & $ab<0$ & $ab<0$ & $ab<(\alpha+1)\beta$ \\
		$\includegraphics[width=.22\columnwidth]{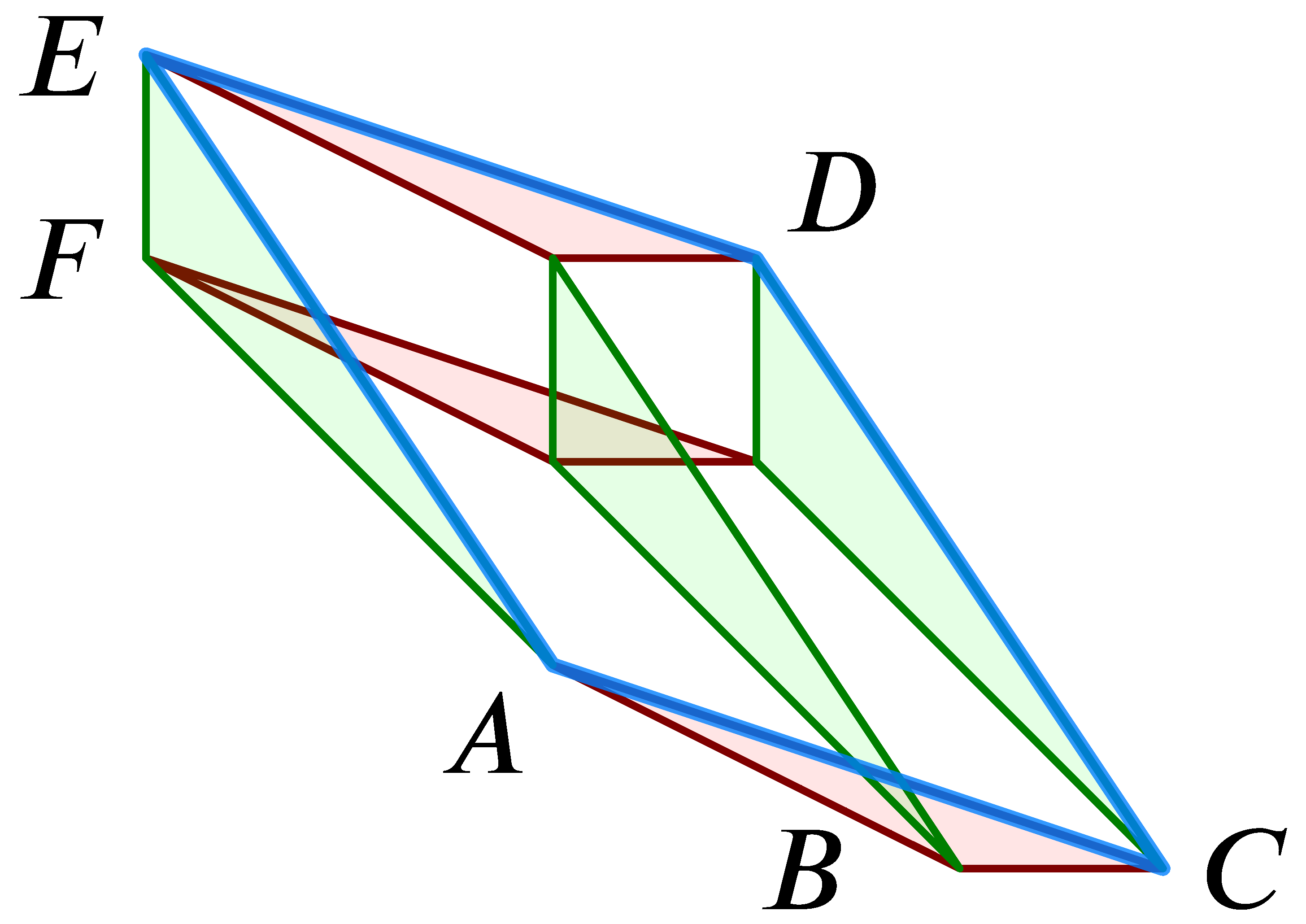}$ &
		$\includegraphics[width=.22\columnwidth]{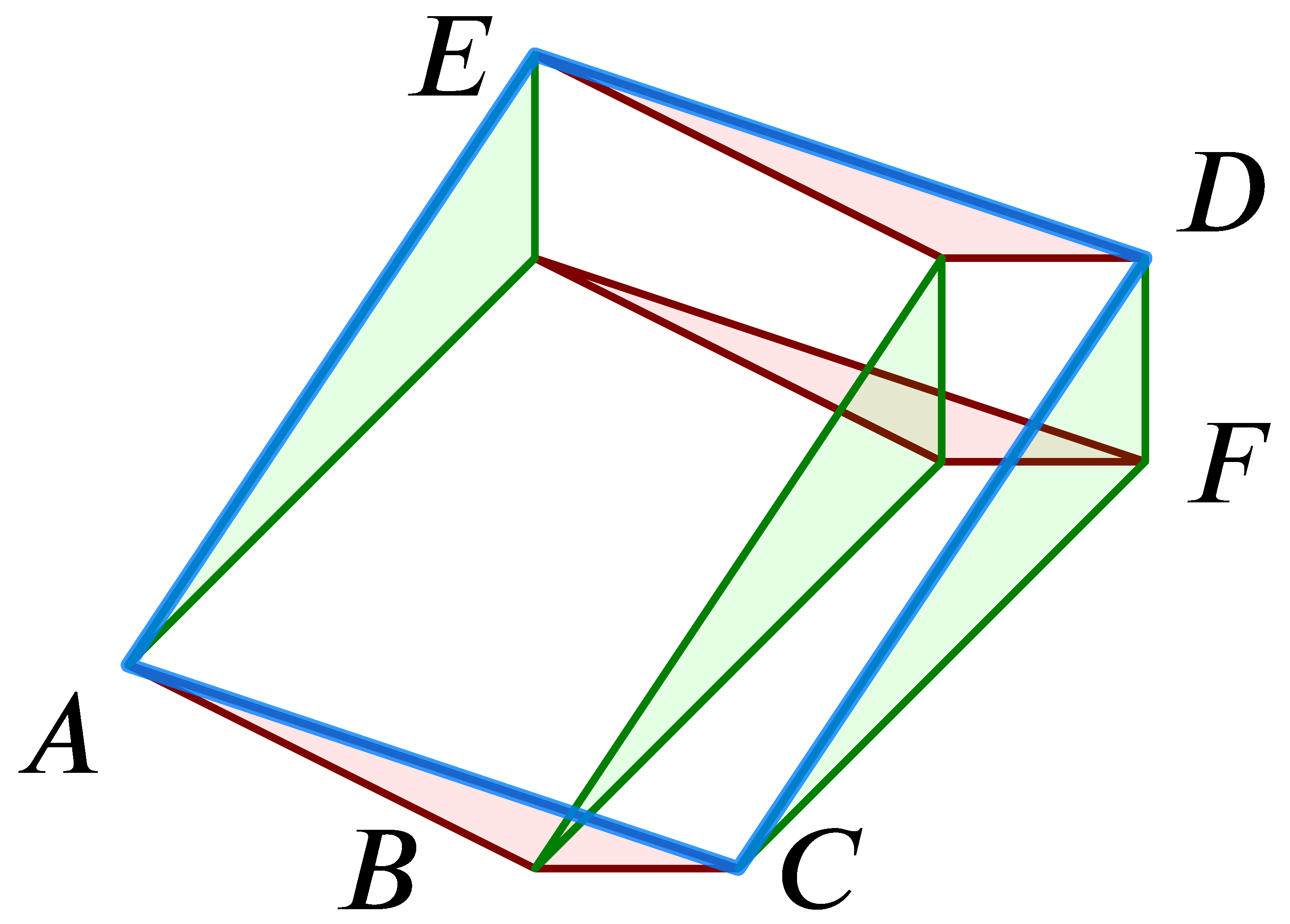}$ &
		$\includegraphics[width=.22\columnwidth]{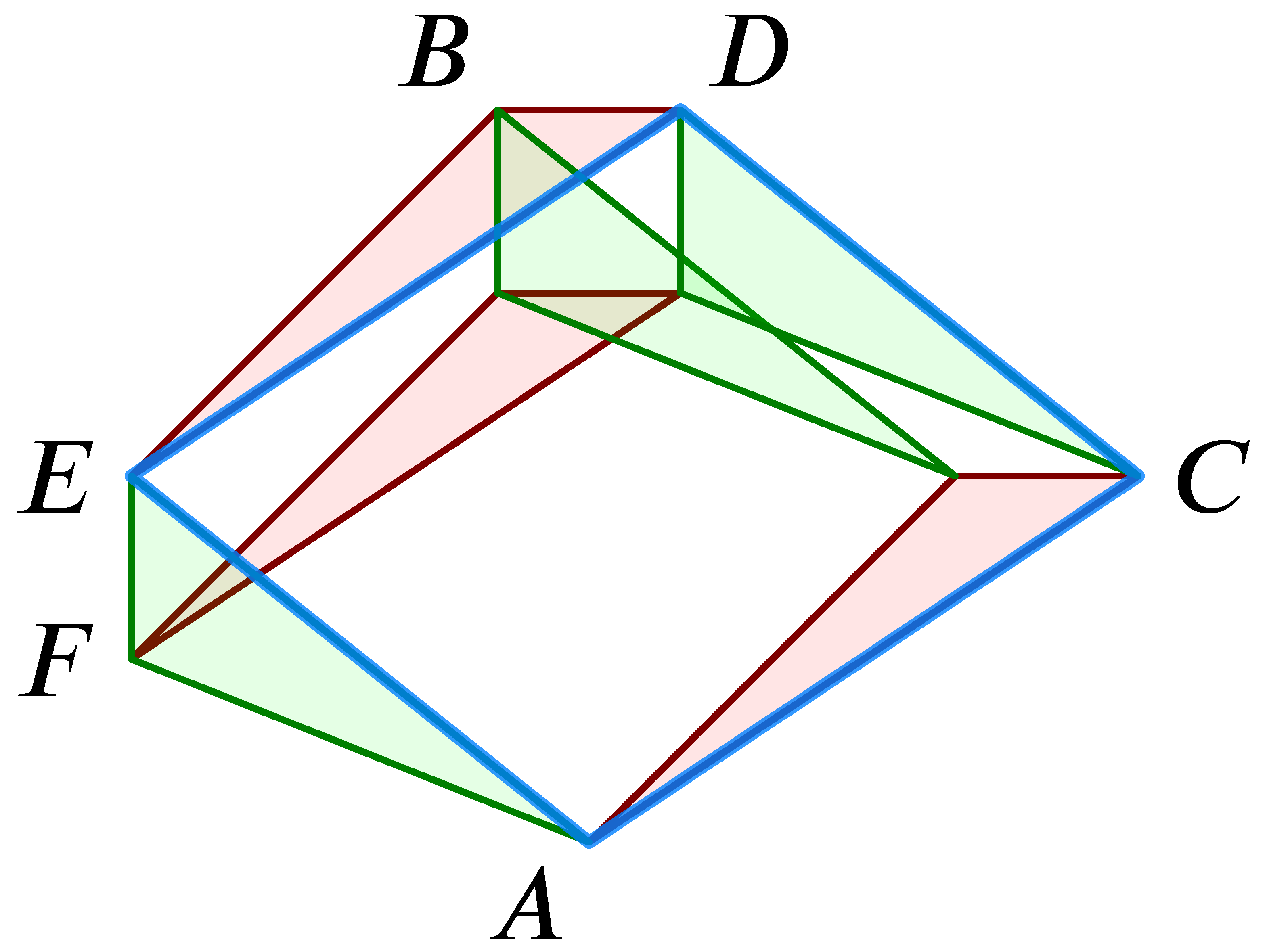}$ &
		$\includegraphics[width=.20\columnwidth]{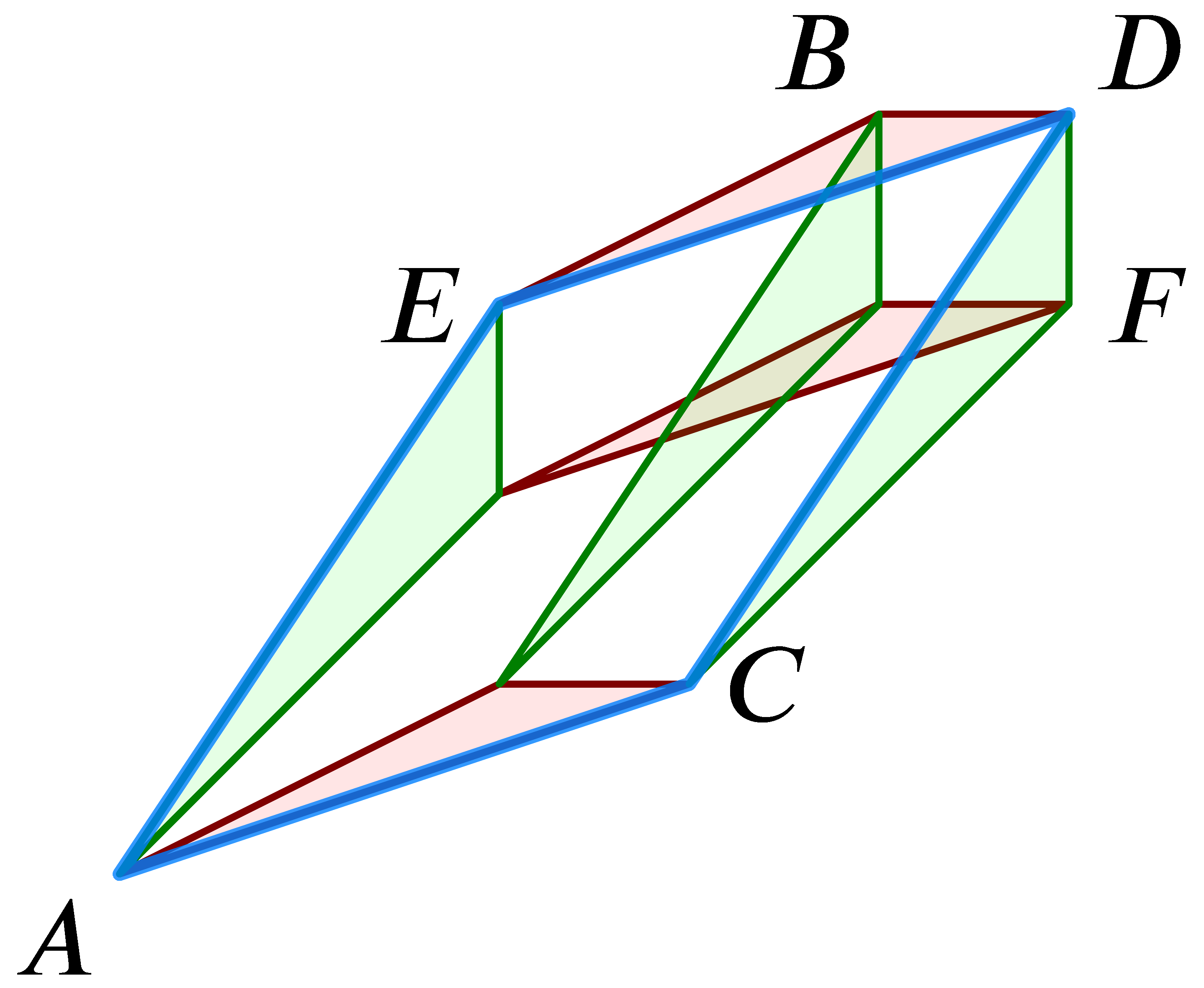}$ \\
		\multicolumn{4}{c}{$(\alpha+1)(\beta+1)-ab$} \\
		\hline\hline
		\multicolumn{4}{c}{$a\leq 0$, $b\leq 0$} \\
		\hline
		$ab>\alpha(\beta+1)$ & $ab=\alpha(\beta+1)$ & $(\alpha+1)\beta<ab<\alpha(\beta+1)$ & $ab=(\alpha+1)\beta$ \\
		$\includegraphics[width=.22\columnwidth]{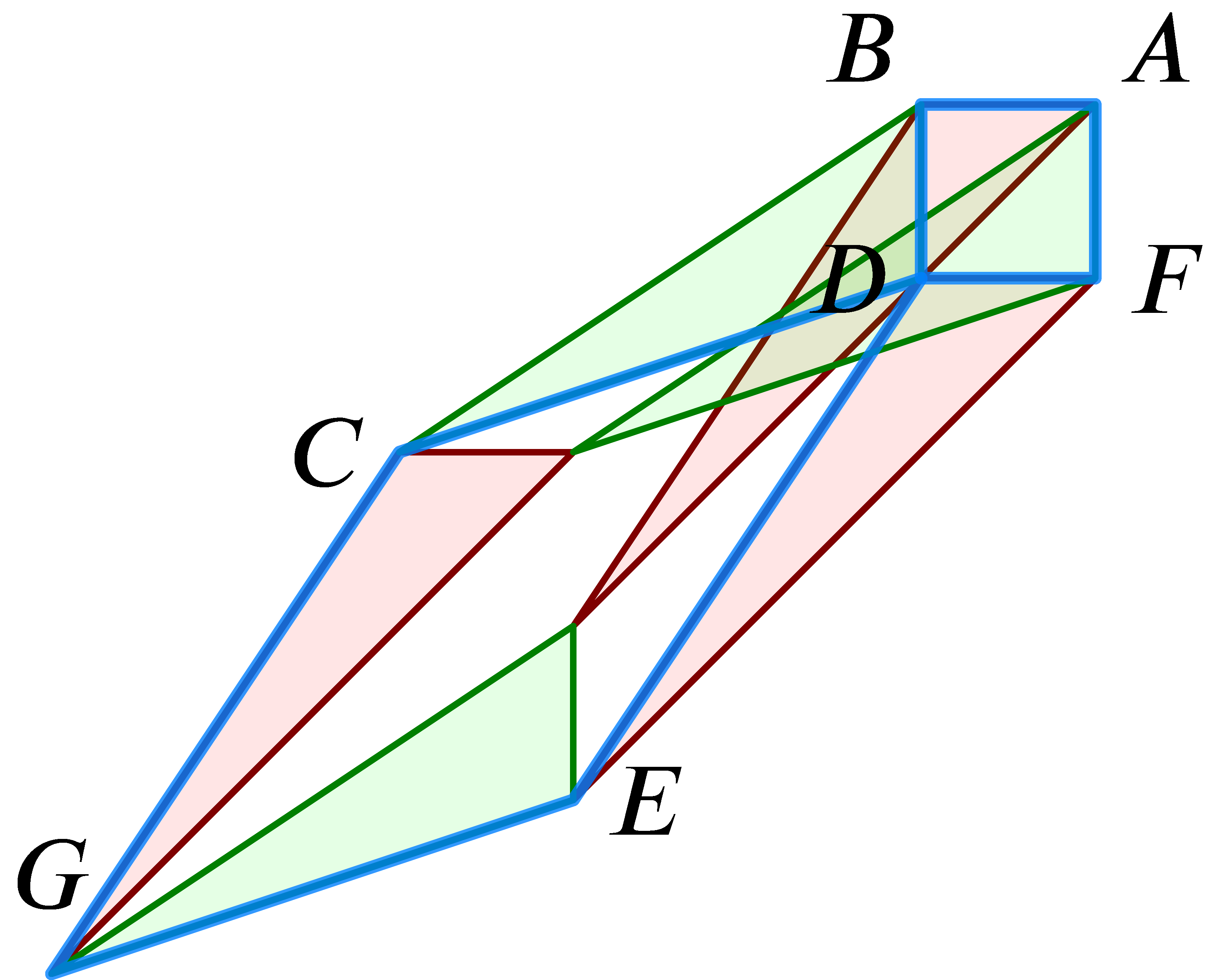}$ &
		$\includegraphics[width=.22\columnwidth]{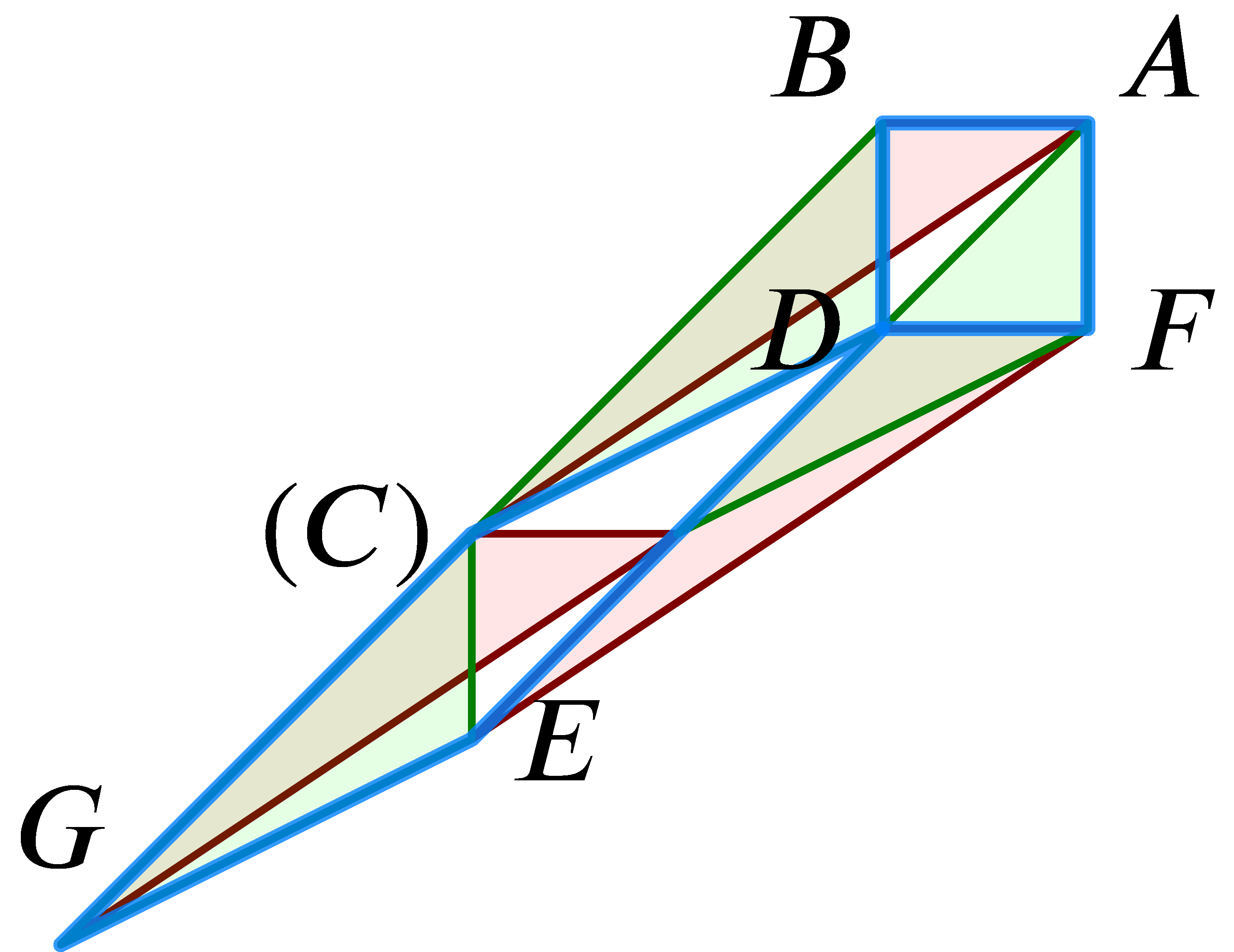}$ &
		$\includegraphics[width=.22\columnwidth]{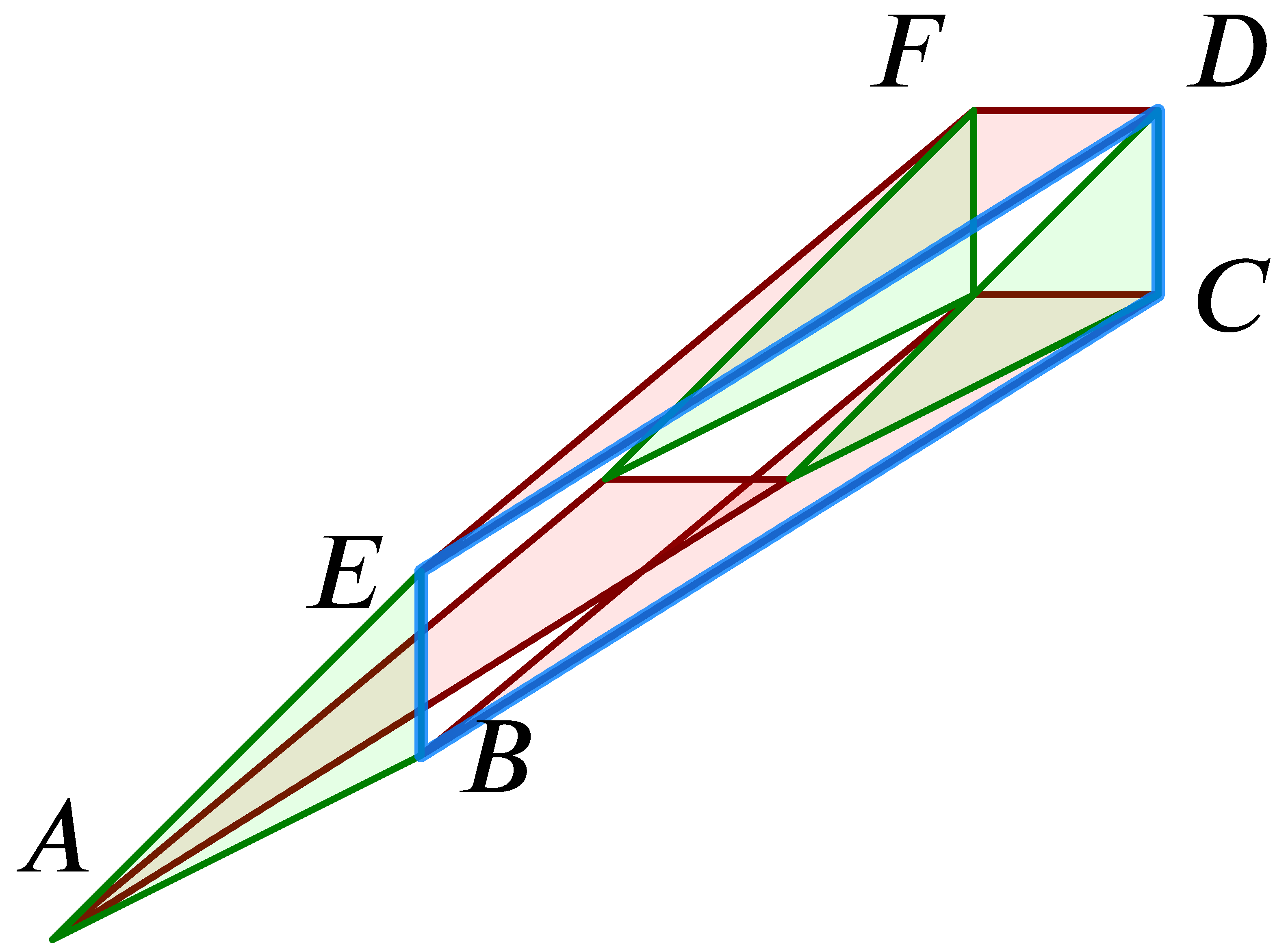}$ &
		$\includegraphics[width=.24\columnwidth]{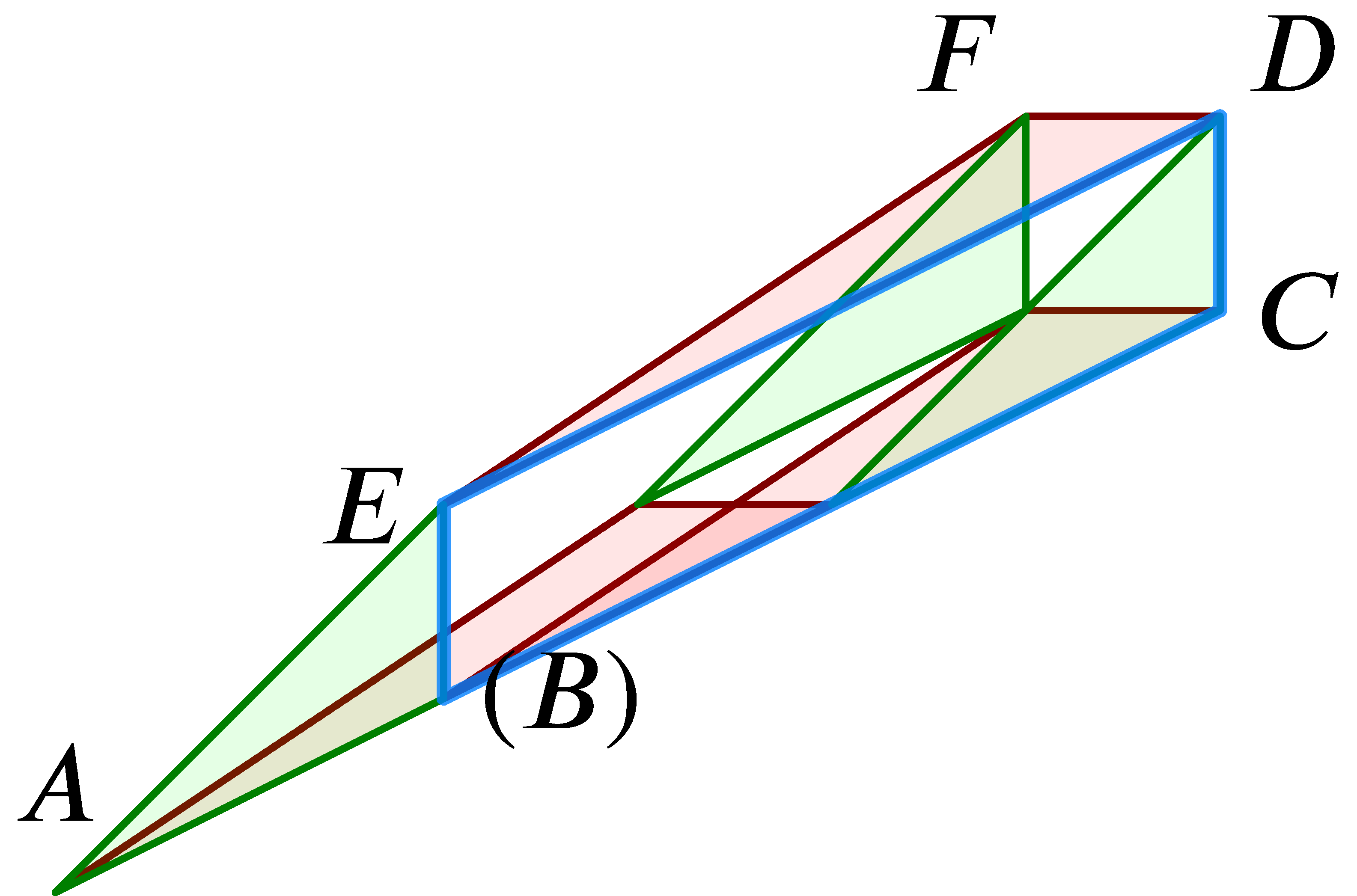}$ \\
		$ab-\alpha\beta+1$ & $\text{Deg}\left\{(1+u)\varphi_{\alpha+1,a}(u)+\varphi_{\beta,b}(u)\right\}$ & $\alpha+1$ & $\text{Deg}\left\{\varphi_{\alpha+1,a}(u)+(1+u)\varphi_{\beta,b}(u)\right\}$ \\
		\toprule[1.5pt]
	\end{tabular}
	\caption{The BKK computation of the topological index $Q$ for the BB code across various parameter choices. The red triangles represent the Newton polytopes of $f$, while the green triangles correspond to those of $g$. The regions enclosed by the blue lines delineate the mixed volume. Letters in parentheses denote points where two boundary lines are collinear, signaling a failure of the BKK condition. Each figure is accompanied by its associated topological index, displayed below it.}\label{table:newton_polytope}
\end{table}

\subsection{Topological Index on Hyperbolas}
We derive explicit expressions for the topological index under the condition that the parameters lie on hyperbolas. 

For demonstration, we focus on the calculation for the hyperbola $ab=(\alpha+1)(\beta+1)$, which implies $\frac{\alpha+1}{b} = \frac{a}{\beta+1}$. This common ratio can be expressed as follows: 
\begin{equation}
	\frac{\alpha+1}{a} = \frac{b}{\beta+1} := \frac{p}{q}
\end{equation}
for some positive integers $p$ and $q$ satisfying $\gcd(p,q) = 1$. Furthermore, it follows that there exist positive integers $r$ and $s$ such that $(\alpha+1,b)=p\cdot(r,s)$ and $(a, \beta+1)=q\cdot(r,s)$. Then
\begin{align}
	Q &= \dim_{\mathbb{F}_2}\frac{\mathbb{F}_2[x^{\pm}, y^{\pm}, z^{\pm}]}{\left(x^{-1}+1+z^{-p}, y^{-1}+1+z^q, z+x^ry^{-s}\right)} \\
	&= \dim_{\mathbb{F}_2}\frac{\mathbb{F}_2[x^{\pm}, y^{\pm}, z^{\pm}]}{\left(x^{-1}+1+z^{-p}, y^{-1}+1+z^q, (1+z^p)^r+z^{pr-1}(1+z^q)^s\right)}\label{eq:hyperbola2} \\
	&= \deg \mathcal{N}_{z^p+1}\mathcal{N}_{z^q+1}\left((1+z^p)^r+z^{pr-1}(1+z^q)^s\right).\label{eq:hyperbola3}
\end{align}
Let $\mathcal{N}_a(b)$ denote the largest factor of $b$ that is coprime to $a$. Since the topological index $Q$ counts the number of zeros over the algebraic closure of $\mathbb{F}_2$, we note from Eq.~\eqref{eq:hyperbola2} that the number of zeros is intrinsically linked to the roots of $(1+z^p)^r+z^{pr-1}(1+z^q)^s$, as the variables $x^\pm$ and $y^\pm$ can be represented in terms of $z$.  However, due to the invertibility of $x$, $y$, and $z$, we must exclude the zeros originating from $z^p+1$, $z^q+1$ as well as $z$, which leads us to the formula in Eq.~\eqref{eq:hyperbola3}. Clearly, $\mathcal{N}_a(b) = \mathcal{N}_{\gcd(a,b)}(b)$ and $\mathcal{N}_{ab}(c)=\mathcal{N}_a\mathcal{N}_b(c)=\mathcal{N}_b\mathcal{N}_a(c)$. To further simplify the expression, we proceed as follows:
\begin{align}
	\mathcal{N}_{z^p+1}\left((1+z^p)^r+z^{pr-1}(1+z^q)^s\right) &= \mathcal{N}_{\gcd(z^p+1, (1+z^p)^r+z^{pr-1}(1+z^q)^s)}\left((1+z^p)^r+z^{pr-1}(1+z^q)^s\right)\\
	&=\mathcal{N}_{z^p+1, (z^q+1)^s}\left((1+z^p)^r+z^{pr-1}(1+z^q)^s\right) \\
	&= \mathcal{N}_{z^{\gcd(p,q)+1}}\left((1+z^p)^r+z^{pr-1}(1+z^q)^s\right) \\
	&= \mathcal{N}_{z+1}\left((1+z^p)^r+z^{pr-1}(1+z^q)^s\right).
\end{align}
We can replace $\mathcal{N}_{z^q+1}$ with $\mathcal{N}_{z+1}(\cdot)$ for the same reason, obtaining
\begin{align}
	Q &= \deg\mathcal{N}_{z+1}\left((1+z^p)^r+z^{pr-1}(1+z^q)^s\right) \\
	&= \text{Deg} \left((1+(1+u)^p)^r + (1+u)^{pr-1}(1+(1+u)^q)^s\right) \quad \text{with}\quad z \mapsto u+1.
\end{align}
We use the notation $\text{Deg}$ to denote the degree difference between the highest and the lowest terms of Laurent polynomials. We introduce an auxiliary function defined as
\begin{equation}
	\varphi_{\gamma,\theta}(u) \coloneqq \left(1+(1+u)^{\frac{\gamma}{\gcd(\gamma, |\theta|)}}\right)^{\gcd(\gamma, |\theta|)}.
\end{equation}
Thus, the final simplified result is:
\begin{equation}
	Q = \text{Deg}\left\{\varphi_{\alpha+1,a}(u)+(1+u)^\alpha\varphi_{\beta+1,b}(u)\right\}.
\end{equation}
As a special case, consider the parameters $(\overline{\alpha}, \overline{\beta}, a, b) = (\overline{c},\overline{c},c+1,c+1)$. We derive that
\begin{equation}
	Q(\overline{c},\overline{c},c+1,c+1)=\text{Deg}\left((u+1)^c+1\right) = c - p_2(c)
\end{equation}
where $p_2(u)$ is the greatest factor of $u$ that is a power of 2; for instance, $p_2(4)=4$ and $p_2(5)=1$. 
Note $Q(\overline{c},\overline{c},c+1,c+1)=0$ if $c$ is a power of 2.

Expressions for $Q$ corresponding to other hyperbolas can be derived similarly and are summarized in Table~\ref{table:newton_polytope}.

\section{Anyon periods and mobility sublattice}\label{app:anyon_periods}

For a topological order with lattice translation symmetry, the translation group $\Lambda=\{x^iy^j | i,j\in\mathbb{Z}\}$ may permute anyon types. We denote the group action by 
\begin{equation}
	\Lambda \times \mathscr{C} \to \mathscr{C}:\  (\lambda, \mathfrak{a}) \mapsto \lambda \mathfrak{a},
\end{equation}
where $\mathscr{C}$ denotes the set of anyon types. 

Here, we are interested in finding the sublattice (a subgroup of $\Lambda$) that preserves all anyon types, namely,
\begin{equation}
	\Lambda_\mathscr{C}=\{\lambda\in\Lambda|\ \lambda \mathfrak{a}=\mathfrak{a}, \forall \mathfrak{a}\in\mathscr{C}\}
\end{equation}
We refer to $\Lambda_{\mathscr{C}}$ as the \emph{mobility sublattice} for $\mathscr{C}$.

For BB codes, since $\mathscr{C}$ is generated by the elementary excitations ${1 \choose 0}$ and ${0 \choose 1}$, we only need ${\lambda \choose 0} \sim {1 \choose 0}$ and ${0 \choose \lambda} \sim {0 \choose 1}$, which requires $\lambda\sim 1$ modulo the ideal $(f,g)_R$. Therefore,
\begin{equation}
	\Lambda_\mathscr{C} =\{\lambda\in\Lambda |\ \lambda-1\in (f,g)_R\}.
\end{equation}
We use $\ell_o$ and $m_o$ to denote the \emph{anyon periods}---namely, the lattice spacings of $\Lambda_{\mathscr{C}}$---in the $x$ and $y$ directions, respectively.

In the following, we introduce an efficient algorithm for computing anyon periods.

For demonstration, we show how to compute $\ell_o$ for $\mathbb{BB}(\overline{1},\overline{1},3,\overline{3})$.  First, compute the Gr\"obner basis of $\mathcal{I}=(1+x+\bar{x}\bar{y}^3,\ 1+y+\bar{y}x^3,\ x\bar{x}+1,\ y\bar{y}+1)$ in the lexicographic order $\bar{y}>\bar{x}>y>x$. The result is:
\begin{equation}
	\text{GB}=\left\{x^{11} + x^{10} + x^9 + x^7 + x^6 + x^4 + 1,\ yx^2 + yx + y + x^8 + x^7 + x^5 + x^4,\ y^2 + y + x^3,\ \bar{x}+y^3x+y^3,\ \bar{y}+\bar{x}^3y+\bar{x}^3\right\}
\end{equation}
Let $h=x^{11} + x^{10} + x^9 + x^7 + x^6 + x^4 + 1$, then 
\begin{equation}
	\ell_o = \text{ord}_{\mathcal{I}}(x) = \text{ord}_h(x).
\end{equation}
Here and below, we use $\text{ord}_\mathcal{I}(x)$ to denote the smallest $\ell>0$ satisfying $x^\ell-1\in \mathcal{I}$, and use $\text{ord}_{h}(x)$ to denote the smallest $\ell>0$ satisfying $h | (x^\ell-1)$. In the literature, $\text{ord}_{h}(x)$ is referred to as the \textit{period} of polynomial $h$ \cite{ber84}.

Factoring $h$ into the product of irreducible polynomials, we have 
\begin{equation}\label{eq:factor}
	h = (x^2+x+1)^2(x^7+x^6+x^4+x^2+1)
\end{equation}
According to the theory of polynomials over finite fields, if a polynomial has degree $n$, its period must be a factor of $2^n-1$ \cite{ber84}. So we can factor the integer $2^n-1$ to its unique production of primes and check each factors. The period of each irreducible polynomial
\begin{align}\label{eq:period}
	\mathrm{ord}_{x^2+x+1}(x) &= 3, \\
	\mathrm{ord}_{x^7+x^6+x^4+x^2+1}(x) &= 127.
\end{align}
\begin{thm}{\rm \cite[Theorem 6.21]{ber84}}\label{thm:period}
	If $f(x)=\prod_i[f_i(x)]^{m_i}$, where $f_i(x)$ are irreducible polynomials of periods $n_i$ over $\mathbb{F}_2$, then the period of $f(x)$ is the least common multiple of $n_i$ times the least power of $2$ which is not less then any of the $m_i$.
\end{thm}
Using the theorem above, we have
\begin{equation}
	\ell_o = 2 \times \mathrm{lcm}(3,127)=762.
\end{equation}
Similarly, the computation of $m_0$ can be performed using an analogous method, with the modification of adopting the lexicographic order $\bar{x}>\bar{y}>x>y$. We get
\begin{equation}
	\text{GB}' = \{y^{11}+y^{10}+y^6+y^4+y^3+y+1, y^2x+yx+x+y^{10}+y^5+y^2+y, x^2+x+y^{10}+y^7+y^5+y^4, \overline{y}+y^{10}+y^9+y^5+y^3+y^2+1, \overline{x}+x+y^5+y^4+y^3+1\}
\end{equation}
Now let $h'=y^{11}+y^{10}+y^6+y^4+y^3+y+1$ and factoring
\begin{equation}
	h' = (y^2+y+1)^2(y^7+y^6+y^5+y^4+y^2+y+1)
\end{equation}
Compute the period of each factor
\begin{align}\label{eq:period}
	\mathrm{ord}_{y^2+y+1}(x) &= 3 \\
	\mathrm{ord}_{y^7+y^6+y^5+y^4+y^2+y+1}(x) &= 127
\end{align}
Using Theorem~\ref{thm:period} again, we have 
\begin{equation}
	m_o = 2 \times \mathrm{lcm}(3,127) = 762
\end{equation}
The factorization of polynomials and the computation of periods for irreducible polynomials are supported by efficient algorithms. 

The pair $(\ell_o,m_o)$ does not inherently generate the entire mobility sublattice since $\langle x^{\ell_o}, y^{m_o}\rangle$ may not be the unit cell. In order to determine the generators for $\Lambda_{\mathscr{C}}$, we adopt the advanced computational algebra methods proposed in \cite{kreuzer24} to find the binomial generators, specifically, polynomials of the form $x^ly^m-1$. The primary methodology involves the decomposition of the quotient ring $R/\mathcal{I}$ into separable and nilpotent components. This is formally expressed as
\begin{equation}
	(R/\mathcal{I})^\times \cong (1+\mathfrak{R})\times \mathfrak{S}^\times 
\end{equation}
Here $\mathfrak{R}$ is the nilradical and $\mathfrak{S}$ is the separate part. This decomposition facilitates the computation of sublattices corresponding to each component individually.  To determine the mobility sublattices associated with the separable components, localize $R/\mathcal{I}$ at each maximum ideal and compute each mobility sublattice on residue fields. To compute the mobility sublattices related to nilpotent parts, it can be demonstrated that $1+\mathfrak{R}$  forms an abelian group. This property allows for the computation of mobility sublattices using the structural properties of abelian groups.  Combining the results, we obtain complete information regarding the mobility sublattices.

We calculate the mobility sublattices for the BB code with $\alpha, \beta$ in $\{0,1,2\}$ and $a, b$ ranging from $-3$ to $3$. Our algorithm enables the computation of billion-sized periods in several seconds on a personal computer. The results are displayed in Table~\ref{tabel:00}$\sim$\ref{tabel:22}.

\begin{table}[h]
	\begin{tabular}{
			|>{\centering\arraybackslash}p{.03\textwidth}|
			>{\centering\arraybackslash}p{.09\textwidth}|
			>{\centering\arraybackslash}p{.09\textwidth}|
			>{\centering\arraybackslash}p{.09\textwidth}|
			>{\centering\arraybackslash}p{.09\textwidth}|
			>{\centering\arraybackslash}p{.09\textwidth}|
			>{\centering\arraybackslash}p{.09\textwidth}|
			>{\centering\arraybackslash}p{.09\textwidth}|
		}
		\hline
		$a \backslash b$ &  -3 & -2 & -1 & 0 & 1 & 2 & 3 \\
		\hline\hline
		\multirow{3}{*}{$3$} & $Q=9$ &  6 & 3& 0& 0& 4&7\\
		& $\langle x^{511},x^{174}y\rangle$ & $\langle x^{63},x^{29}y\rangle$ & $\langle x^{7},x^{5}y\rangle$ & \multirow{2}{*}{*} &\multirow{2}{*}{*} & $\langle x^{5},x^{3}y^3\rangle$ & $\langle x^{105},x^{41}y\rangle$\\ 
		& $\langle y^{511}, xy^{279}\rangle$ & $\langle y^{63},xy^{50}\rangle$ & $\langle y^{7},xy^{3}\rangle$ &  & & $\langle y^{15},xy^{6}\rangle$ & $\langle y^{105},xy^{41}\rangle$\\ 
		\hline
		\multirow{3}{*}{$2$} & 6  & 4 & 2 & 0& 0& 2&\\
		& \multirow{2}{*}{$\langle x^7,y^7 \rangle$} & $\langle x^{15},x^{6}y\rangle$ & $\langle x^{3},x^{2}y\rangle$ & \multirow{2}{*}{*} &  \multirow{2}{*}{*}& $\langle x^{3},x^{2}y\rangle$& \\ 
		& & $\langle y^{5},x^3y^{3}\rangle$ & $\langle y^{3},xy^{2}\rangle$ &  &  & $\langle y^{3},xy^{2}\rangle$ &\\ 
		\hline
		\multirow{3}{*}{$1$}&0  & 2 & 0 & 0& $\infty$& &\\
		&\multirow{2}{*}{*} & $\langle x^{3},xy\rangle$ &\multirow{2}{*}{*}  & \multirow{2}{*}{*} & \multirow{2}{*}{*}& & \\ 
		& & $\langle y^{3},xy\rangle$ &  &  &  &  &\\ 
		\hline
		\multirow{3}{*}{$0$}& 0  & 0 & 0 & 0& & &\\
		&\multirow{2}{*}{*} & \multirow{2}{*}{*}&\multirow{2}{*}{*} & \multirow{2}{*}{*}& & & \\ 
		& & & &  &  &  &\\ 
		\hline
		\multirow{3}{*}{$-1$}& 4  & 3 & 2 & & & &\\
		&$\langle x^{15},x^{4}y\rangle$ & $\langle x^{7},x^{5}y\rangle$ & $\langle x^{3},x^{2}y\rangle$ &  & & & \\ 
		& $\langle y^{15},xy^{4}\rangle$ & $\langle y^{7},xy^{3}\rangle$ & $\langle y^{3},xy^{2}\rangle$ &  & & &\\ 
		\hline
		\multirow{3}{*}{$-2$}& 7  & 5 &  & & & &\\
		& $\langle x^{127},x^{91}y\rangle$ &   $\langle x^{21},x^{13}y\rangle$ & & & & & \\ 
		& $\langle y^{127},xy^{67}\rangle$ &   $\langle y^{21},xy^{13}\rangle$ & & & & &\\ 
		\hline
		\multirow{3}{*}{$-3$}&10  &  &  & & & &\\
		& $\langle x^{315},x^{102}y^3\rangle$ & & & & & & \\ 
		& $\langle y^{315},x^3y^{102}\rangle$ & & & & & &\\ 
		\hline
	\end{tabular}
	\caption{Topological index and mobility sublattices for the codes specified by $f=1+x+y^b$ and $g=1+y+x^a$. 
		Each pair $\left\langle \cdot,\cdot\right\rangle $ provides a set of translation generators for the corresponding mobility sublattice. We omit the value in right-bottom due to the symmetry between $a$ and $b$.}\label{tabel:00}
\end{table}

\begin{table}[h]
	\begin{tabular}{
			|>{\centering\arraybackslash}p{.03\textwidth}|
			>{\centering\arraybackslash}p{.09\textwidth}|
			>{\centering\arraybackslash}p{.09\textwidth}|
			>{\centering\arraybackslash}p{.09\textwidth}|
			>{\centering\arraybackslash}p{.09\textwidth}|
			>{\centering\arraybackslash}p{.09\textwidth}|
			>{\centering\arraybackslash}p{.09\textwidth}|
			>{\centering\arraybackslash}p{.09\textwidth}|
		}
		\hline
		$a \backslash b$ &  -3 & -2 & -1 & 0 & 1 & 2 & 3 \\
		\hline\hline
		\multirow{3}{*}{$3$} & $Q=10$ &  7 & 4& 0& 2& 5&8\\
		& $\langle x^{1023},x^{989}y\rangle$ & $\langle x^{93},x^{22}y\rangle$ & $\langle x^{6},x^{4}y^2\rangle$ & \multirow{2}{*}{*} & $\langle x^{3},xy\rangle$& $\langle x^{21},x^{8}y\rangle$ & $\langle x^{217},x^{192}y\rangle$\\ 
		& $\langle y^{1023}, xy^{692}\rangle$ & $\langle y^{93},xy^{55}\rangle$ & $\langle y^{6},x^2y^{4}\rangle$ &  &$\langle y^{3},xy\rangle$ & $\langle y^{21},xy^{8}\rangle$ & $\langle y^{1217},xy^{26}\rangle$\\ 
		\hline
		\multirow{3}{*}{$2$} & 7  & 5 & 3 & 0& 0& 3&5\\
		& $\langle x^{127},x^{38}y\rangle$ & $\langle x^{31},x^{7}y\rangle$ & $\langle x^{7},x^{3}y\rangle$ &\multirow{2}{*}{*}  & \multirow{2}{*}{*} & $\langle x^{7},x^{2}y\rangle$& $\langle x^{31},x^{25}y\rangle$\\ 
		& $\langle y^{127},xy^{117}\rangle$ & $\langle y^{31},xy^{9}\rangle$ & $\langle y^{7},xy^{3}\rangle$ &  &  & $\langle y^{7},xy^{2}\rangle$ &$\langle y^{31},xy^{5}\rangle$\\ 
		\hline
		\multirow{3}{*}{$1$}&3  & 3 & 0 & 0& 0& 0&0\\
		& $\langle x^{7},xy\rangle$& $\langle x^{7},x^5y\rangle$ &\multirow{2}{*}{*}  &\multirow{2}{*}{*}  & \multirow{2}{*}{*}& \multirow{2}{*}{*}&\multirow{2}{*}{*} \\ 
		& $\langle y^{7},xy\rangle$& $\langle y^{7},xy^3\rangle$ &  &  &  &  &\\ 
		\hline
		\multirow{3}{*}{$0$}& 0  & 2 & 2 & 0& 2& 2&0\\
		&\multirow{2}{*}{*} &$\langle x^{3},xy\rangle$ & $\langle x^{3},x^2y\rangle$& \multirow{2}{*}{*}&$\langle x^{3},xy\rangle$ &$\langle x^{3},x^2y\rangle$ & \multirow{2}{*}{*}\\ 
		& &$\langle y^{3},xy\rangle$ & $\langle y^{3},xy^2\rangle$ &  & $\langle y^{3},xy\rangle$ & $\langle y^{3},xy^2\rangle$ &\\ 
		\hline
		\multirow{3}{*}{$-1$}& 4  & 3 & 0 & 0& 3& 4&5\\
		&$\langle x^{15},x^{9}y\rangle$ & $\langle x^{7},x^{6}y\rangle$ & \multirow{2}{*}{*} & \multirow{2}{*}{*} & $\langle x^{7},x^{2}y\rangle$& $\langle x^{15},x^{9}y\rangle$& $\langle x^{31},x^{14}y\rangle$\\ 
		& $\langle y^{5},x^3y^{2}\rangle$ & $\langle y^{7},xy^{6}\rangle$ &  &  & $\langle y^{7},xy^{4}\rangle$& $\langle y^{5},x^3y^{2}\rangle$&$\langle y^{31},xy^{20}\rangle$\\ 
		\hline
		\multirow{3}{*}{$-2$}& 7  & 5 & 3 & 0& 4&6 &8\\
		& $\langle x^{105},x^{74}y\rangle$ &   $\langle x^{31},x^{9}y\rangle$ & $\langle x^{7},x^{5}y\rangle$&\multirow{2}{*}{*} & $\langle x^{15},x^{3}y\rangle$& $\langle x^{63},x^{34}y\rangle$& $\langle x^{51},x^{46}y^5\rangle$\\ 
		& $\langle y^{105},xy^{44}\rangle$ &   $\langle y^{31},xy^{7}\rangle$ & $\langle y^{7},xy^{3}\rangle$& & $\langle y^{5},x^3y\rangle$& $\langle y^{63},xy^{13}\rangle$&$\langle y^{255},xy^{50}\rangle$\\ 
		\hline
		\multirow{3}{*}{$-3$}&10  & 7 & 4 & 0& 5& 8&11\\
		& $\langle x^{1023},x^{28}y\rangle$ &$\langle x^{93},x^{25}y\rangle$ &$\langle x^{6},x^{4}y^2\rangle$ &\multirow{2}{*}{*}& $\langle x^{21},x^{4}y\rangle$&$\langle x^{63},x^{35}y\rangle$ & $\langle x^{119},x^{6}y^{15}\rangle$\\ 
		& $\langle y^{1023},xy^{475}\rangle$ &$\langle y^{93},xy^{67}\rangle$ &$\langle y^{6},x^2y^{4}\rangle$ & &$\langle y^{121},xy^{16}\rangle$ &$\langle y^{9},x^7y^{2}\rangle$ &$\langle y^{1785},xy^{300}\rangle$\\ 
		\hline
	\end{tabular}
	\caption{Topological index and mobility sublattices for the codes specified by $f=1+x+y^b$ and $g=1+y+y^{-1}x^a$.}\label{tabel:01}
\end{table}

\begin{table}[h]
	\begin{tabular}{
			|>{\centering\arraybackslash}p{.03\textwidth}|
			>{\centering\arraybackslash}p{.09\textwidth}|
			>{\centering\arraybackslash}p{.09\textwidth}|
			>{\centering\arraybackslash}p{.09\textwidth}|
			>{\centering\arraybackslash}p{.09\textwidth}|
			>{\centering\arraybackslash}p{.09\textwidth}|
			>{\centering\arraybackslash}p{.09\textwidth}|
			>{\centering\arraybackslash}p{.09\textwidth}|
		}
		\hline\hline
		$a \backslash b$ &  -3 & -2 & -1 & 0 & 1 & 2 & 3 \\
		\hline
		\multirow{3}{*}{$3$} & $Q=11$ &  8 & 5& 0& 0& 5&8\\
		& $\langle x^{2047},x^{647}y\rangle$ & $\langle x^{217},x^{180}y\rangle$ & $\langle x^{31},x^{20}y\rangle$ & \multirow{2}{*}{*} &\multirow{2}{*}{*} & $\langle x^{31},x^{22}y\rangle$ & $\langle x^{255},x^{178}y\rangle$\\ 
		& $\langle y^{2047}, xy^{560}\rangle$ & $\langle y^{217},xy^{129}\rangle$ & $\langle y^{31},xy^{129}\rangle$ &  & & $\langle y^{31},xy^{24}\rangle$ & $\langle y^{255},xy^{202}\rangle$\\ 
		\hline
		\multirow{3}{*}{$2$} & 8  & 6 & 4 & 0& 2& 3&5\\
		& $\langle x^{255},x^{132}y\rangle$ & $\langle x^{15},x^{6}y^3\rangle$ & $\langle x^{5},x^{4}y^3\rangle$ & \multirow{2}{*}{*} & $\langle x^{3},xy\rangle$ & $\langle x^{7},xy\rangle$&$\langle x^{31},x^4y\rangle$ \\ 
		& $\langle y^{85},x^3y^{29}\rangle$ & $\langle y^{15},x^3y^{9}\rangle$ & $\langle y^{15},xy^{12}\rangle$ &  & $\langle y^{3},xy\rangle$ & $\langle y^{7},xy\rangle$ &$\langle y^{31},xy^8\rangle$\\ 
		\hline
		\multirow{3}{*}{$1$}&3  & 4 & 2 & 0& 0& 2&0\\
		& $\langle x^{7},x^4y\rangle$& $\langle x^{5},x^2y^3\rangle$ & $\langle x^{3},x^2y\rangle$ & \multirow{2}{*}{*} & \multirow{2}{*}{*}& $\langle x^{3},x^2y\rangle$& \multirow{2}{*}{*}\\ 
		& $\langle y^{7},xy^2\rangle$& $\langle y^{15},xy^9\rangle$ & $\langle y^{3},xy^2\rangle$ &  &  & $\langle y^{3},xy^2\rangle$ &\\ 
		\hline
		\multirow{3}{*}{$0$}& 3  & 3 & 3 & 0& 3& 3&3\\
		&$\langle x^{7},xy\rangle$ & $\langle x^{7},x^6y\rangle$&$\langle x^{7},x^5y\rangle$ & \multirow{2}{*}{*}& $\langle x^{7},x^4y\rangle$&$\langle x^{7},x^2y\rangle$ &$\langle x^{7},x^3y\rangle$\\ 
		&$\langle y^{7},xy\rangle$ & $\langle y^{7},xy^6\rangle$ &$\langle y^{7},xy^3\rangle$  &  &$\langle y^{7},xy^2\rangle$ & $\langle y^{7},xy^4\rangle$&$\langle y^{7},xy^5\rangle$\\ 
		\hline
		\multirow{3}{*}{$-1$}& 4  & 2 & 3 &0 & 4& 5&6\\
		&$\langle x^{15},x^{14}y\rangle$ & $\langle x^{3},xy\rangle$ & $\langle x^{7},x^3y\rangle$ & \multirow{2}{*}{*} & $\langle x^{6},x^2y^2\rangle$& $\langle x^{31},x^{24}y\rangle$& $\langle x^{63},x^{19}y\rangle$\\ 
		& $\langle y^{15},xy^{14}\rangle$ & $\langle y^{3},xy\rangle$ & $\langle y^{7},xy^5\rangle$ &  & $\langle y^{6},x^2y^2\rangle$& $\langle y^{31},xy^{22}\rangle$&$\langle y^{63},xy^{10}\rangle$\\ 
		\hline
		\multirow{3}{*}{$-2$}& 7 & 5 & 2 &0 &5 & 7&9\\
		& $\langle x^{127},x^{7}y\rangle$ &   $\langle x^{31},x^{6}y\rangle$ &$\langle x^{3},x^{2}y\rangle$ & \multirow{2}{*}{*}& $\langle x^{31},x^{17}y\rangle$&$\langle x^{93},x^{71}y\rangle$ & $\langle x^{63},x^{28}y^7\rangle$\\ 
		& $\langle y^{127},xy^{109}\rangle$ &   $\langle y^{31},xy^{26}\rangle$ &$\langle y^{3},xy^{2}\rangle$ & & $\langle y^{31},xy^{11}\rangle$&$\langle y^{93},xy^{38}\rangle$ &$\langle y^{63},x^7y^{49}\rangle$\\ 
		\hline
		\multirow{3}{*}{$-3$}&10  & 7 &  4& 0& 6&9 &12\\
		& $\langle x^{1023},x^{499}y\rangle$ & $\langle x^{127},x^{107}y\rangle$& $\langle x^{5},x^{4}y^3\rangle$& \multirow{2}{*}{*}& $\langle x^{14},x^{4}y^2\rangle$& $\langle x^{465},x^{118}y\rangle$& $\langle x^{3937},x^{3199}y\rangle$\\ 
		& $\langle y^{1023},xy^{982}\rangle$ &$\langle y^{127},xy^{19}\rangle$ &$\langle y^{15},xy^{12}\rangle$ & &$\langle y^{14},x^2y^{8}\rangle$ & $\langle y^{465},xy^{67}\rangle$&$\langle y^{3937},xy^{2630}\rangle$\\ 
		\hline
	\end{tabular}
	\caption{Topological index and mobility sublattices for the codes specified by $f=1+x+y^b$ and $g=1+y+y^{-2}x^a$.}\label{tabel:02}
\end{table}

\begin{table}[h]
	\begin{tabular}{
			|>{\centering\arraybackslash}p{.03\textwidth}|
			>{\centering\arraybackslash}p{.09\textwidth}|
			>{\centering\arraybackslash}p{.09\textwidth}|
			>{\centering\arraybackslash}p{.09\textwidth}|
			>{\centering\arraybackslash}p{.09\textwidth}|
			>{\centering\arraybackslash}p{.09\textwidth}|
			>{\centering\arraybackslash}p{.09\textwidth}|
			>{\centering\arraybackslash}p{.09\textwidth}|
		}
		\hline
		$a \backslash b$ &  -3 & -2 & -1 & 0 & 1 & 2 & 3 \\
		\hline\hline
		\multirow{3}{*}{$3$} & $Q=13$ &  10 & 7& 4& 3& 5&8\\
		& $\langle x^{762},x^{216}y^6\rangle$ & $\langle x^{889},x^{45}y\rangle$ & $\langle x^{127},x^{115}y\rangle$ & \multirow{2}{*}{$\langle x^3,y^3 \rangle$} & $\langle x^{7},xy\rangle$ & $\langle x^{31},x^{8}y\rangle$ & \multirow{2}{*}{$\langle x^{12},y^{12} \rangle$}\\ 
		& $\langle y^{762}, x^{6}y^{360}\rangle$ & $\langle y^{889},xy^{810}\rangle$ & $\langle y^{127},xy^{74}\rangle$ &  & $\langle y^{7},xy\rangle$ & $\langle y^{31},xy^4\rangle$ &\\ 
		\hline
		\multirow{3}{*}{$2$} & 10  &  8 & 6& 4& 3& 0&\\
		& $\langle x^{341},x^{41}y^3\rangle$ & $\langle x^{217},x^{129}y\rangle$ & $\langle x^{21},x^{14}y^3\rangle$ & $\langle x^{3},xy^5\rangle$ & $\langle x^{7},x^3y\rangle$ & \multirow{2}{*}{*}& \\ 
		&$\langle y^{1023}, xy^{549}\rangle$ & $\langle y^{217},xy^{67}\rangle$ & $\langle y^{9},x^7y^{6}\rangle$ & $\langle y^{15},xy^{5}\rangle$ & $\langle y^{7},xy^{5}\rangle$ &  &\\ 
		\hline
		\multirow{3}{*}{$1$}&7  & 6 & 5 & 4& 0& &\\
		&$\langle x^{127},x^{92}y\rangle$ & $\langle x^{21},x^{7}y^3\rangle$ & $\langle x^{31},x^{13}y\rangle$ & $\langle x^{3},x^2y^5\rangle$ & \multirow{2}{*}{*}& & \\ 
		& $\langle y^{127}, xy^{29}\rangle$ & $\langle y^{9},x^7y^{3}\rangle$ & $\langle y^{31},xy^{12}\rangle$ & $\langle y^{15},xy^{10}\rangle$ &  &  &\\ 
		\hline
		\multirow{3}{*}{$0$}& 4  & 4 & 4 & 4& & &\\
		& \multirow{2}{*}{$\langle x^3,y^3 \rangle$} & $\langle x^{15},x^{10}y\rangle$ & $\langle x^{15},x^{5}y\rangle$ & \multirow{2}{*}{$\langle x^3,y^3 \rangle$} & & & \\ 
		& & $\langle y^{3},x^5y^{2}\rangle$ & $\langle y^{3},x^5y\rangle$ &  &  &  &\\ 
		\hline
		\multirow{3}{*}{$-1$}& 3  & 0 & 3 & & & &\\
		&$\langle x^{7},x^{2}y\rangle$ & \multirow{2}{*}{*} & $\langle x^{7},x^{6}y\rangle$ &  & & & \\ 
		& $\langle y^{7},xy^{4}\rangle$ &  & $\langle y^{7},xy^{6}\rangle$ &  & & &\\ 
		\hline
		\multirow{3}{*}{$-2$}& 6  & 4 &  & & & &\\
		& $\langle x^{63},x^{24}y\rangle$ &   $\langle x^{15},x^{14}y\rangle$ & & & & & \\ 
		& $\langle y^{21},x^3y^{8}\rangle$ &   $\langle y^{15},xy^{14}\rangle$ & & & & &\\ 
		\hline
		\multirow{3}{*}{$-3$}&9  &  &  & & & &\\
		& $\langle x^{42},x^{36}y^6\rangle$ & & & & & & \\ 
		& $\langle y^{42},x^6y^{36}\rangle$ & & & & & &\\ 
		\hline
	\end{tabular}
	\caption{Topological index and mobility sublattices for the codes specified by $f=1+x+x^{-1}y^b$ and $g=1+y+y^{-1}x^a$. We omit the value in right-bottom due to the symmetry between $a$ and $b$.}\label{tabel:11}
\end{table}

\begin{table}[h]
	\begin{tabular}{
			|>{\centering\arraybackslash}p{.03\textwidth}|
			>{\centering\arraybackslash}p{.1\textwidth}|
			>{\centering\arraybackslash}p{.1\textwidth}|
			>{\centering\arraybackslash}p{.09\textwidth}|
			>{\centering\arraybackslash}p{.09\textwidth}|
			>{\centering\arraybackslash}p{.09\textwidth}|
			>{\centering\arraybackslash}p{.09\textwidth}|
			>{\centering\arraybackslash}p{.11\textwidth}|
		}
		\hline
		$a \backslash b$ &  -3 & -2 & -1 & 0 & 1 & 2 & 3 \\
		\hline\hline
		\multirow{3}{*}{$3$} & $Q=15$ &  12 & 9& 6& 4& 3&7\\
		& $\langle x^{4681},x^{2270}y^7\rangle$ & $\langle x^{1365},x^{764}y^3\rangle$ & $\langle x^{511},x^{186}y\rangle$ & \multirow{2}{*}{$\langle x^{3},y^7\rangle$} &$\langle x^{15},x^{2}y\rangle$ & $\langle x^{7},x^{5}y\rangle$ & $\langle x^{127},x^{97}y\rangle$\\ 
		& $\langle y^{32767}, xy^{30443}\rangle$ & $\langle y^{4095},xy^{402}\rangle$ & $\langle y^{511},xy^{261}\rangle$ &  & $\langle y^{15},xy^{8}\rangle$& $\langle y^{7},xy^{3}\rangle$ & $\langle y^{127},xy^{55}\rangle$\\ 
		\hline
		\multirow{3}{*}{$2$} & 12  & 10 & 8 & 6& 3& 4&2\\
		& $\langle x^{33},x^{10}y^{31}\rangle$ & $\langle x^{341},x^{326}y\rangle$ & $\langle x^{255},x^{116}y\rangle$ & \multirow{2}{*}{$\langle x^{3},y^{15}\rangle$} & $\langle x^{7},xy\rangle$ & $\langle x^{5},x^4y^3\rangle$&$\langle x^{3},xy\rangle$ \\ 
		& $\langle y^{1023},xy^{310}\rangle$ & $\langle y^{341},xy^{250}\rangle$ & $\langle y^{255},xy^{11}\rangle$ &  & $\langle y^{7},xy\rangle$ & $\langle y^{15},xy^{12}\rangle$ &$\langle y^{3},xy\rangle$\\ 
		\hline
		\multirow{3}{*}{$1$}&9  & 8 & 7 & 6& 5& 3&4\\
		& $\langle x^{186},x^{34}y^2\rangle$& $\langle x^{255},x^{58}y\rangle$ & $\langle x^{127},x^{74}y\rangle$ & \multirow{2}{*}{$\langle x^{3},y^{15}\rangle$} & $\langle x^{31},x^{11}y\rangle$& $\langle x^{7},x^{4}y\rangle$& $\langle x^{6},x^{4}y^2\rangle$\\ 
		& $\langle y^{186},x^2y^{22}\rangle$& $\langle y^{255},xy^{11}\rangle$ & $\langle y^{127},xy^{115}\rangle$ &  & $\langle y^{31},xy^{17}\rangle$ & $\langle y^{7},xy^{2}\rangle$ &$\langle y^{6},x^2y^{4}\rangle$\\ 
		\hline
		\multirow{3}{*}{$0$}& 6  & 6 & 6 & 6& 6& 6&6\\
		&$\langle x^{63},x^{45}y\rangle$ & \multirow{2}{*}{$\langle x^{7},y^{7}\rangle$}&\multirow{2}{*}{$\langle x^{7},y^{7}\rangle$} &\multirow{2}{*}{$\langle x^{3},y^{7}\rangle$} &$\langle x^{63},x^{54}y\rangle$ & $\langle x^{63},x^{27}y\rangle$&\multirow{2}{*}{$\langle x^{7},y^{7}\rangle$}\\ 
		&$\langle y^{7},x^9y^3\rangle$ &  &  &  & $\langle y^{7},x^9y^6\rangle$&$\langle y^{7},x^9y^5\rangle$ &\\ 
		\hline
		\multirow{3}{*}{$-1$}& 2  & 4 & 5 & 6&7 & 8&9\\
		&$\langle x^{3},xy\rangle$ & $\langle x^{15},x^{14}y\rangle$ & $\langle x^{31},x^{14}y\rangle$ & \multirow{2}{*}{$\langle x^{3},y^{15}\rangle$} & $\langle x^{35},x^{23}y^3\rangle$&$\langle x^{51},x^{41}y^{5}\rangle$ & $\langle x^{381},x^{205}y\rangle$\\ 
		& $\langle y^{3},xy\rangle$ & $\langle y^{15},xy^{14}\rangle$ & $\langle y^{31},xy^{20}\rangle$ &  &$\langle y^{105},xy^{96}\rangle$ & $\langle y^{255},xy^{25}\rangle$&$\langle y^{381},xy^{184}\rangle$\\ 
		\hline
		\multirow{3}{*}{$-2$}& 5  & 0 & 4 & 6& 8&10 &12\\
		& $\langle x^{21},x^{20}y\rangle$ &\multirow{2}{*}{*} &$\langle x^{15},x^{13}y\rangle$ & \multirow{2}{*}{$\langle x^{3},y^{15}\rangle$}& $\langle x^{51},x^{31}y^5\rangle$& $\langle x^{1023},x^{757}y\rangle$& $\langle x^{93},x^{62}y^{31}\rangle$\\ 
		& $\langle y^{21},xy^{20}\rangle$ & & $\langle y^{15},xy^{7}\rangle$& &$\langle y^{255},xy^{140}\rangle$ & $\langle y^{1023},xy^{973}\rangle$&$\langle y^{93},x^{31}y^{62}\rangle$\\ 
		\hline
		\multirow{3}{*}{$-3$}&8  & 5 & 0 & 6& 9&12 &15\\
		& $\langle x^{217},x^{149}y\rangle$ &$\langle x^{31},x^{10}y\rangle$ &\multirow{2}{*}{*} &\multirow{2}{*}{$\langle x^{3},y^{7}\rangle$} & $\langle x^{511},x^{195}y\rangle$& $\langle x^{255},x^{45}y^{15}\rangle$& $\langle x^{32767},x^{17772}y\rangle$\\ 
		& $\langle y^{217},xy^{67}\rangle$ & $\langle y^{31},xy^{28}\rangle$& & &$\langle y^{511},xy^{76}\rangle$ &$\langle y^{255},x^{15}y^{90}\rangle$ &$\langle y^{32767},xy^{24590}\rangle$\\ 
		\hline
	\end{tabular}
	\caption{Topological index and mobility sublattices for the codes specified by $f=1+x+x^{-1}y^b$ and $g=1+y+y^{-2}x^a$.}\label{tabel:12}
\end{table}

\begin{table}[h]
	\begin{tabular}{
			|>{\centering\arraybackslash}p{.03\textwidth}|
			>{\centering\arraybackslash}p{.11\textwidth}|
			>{\centering\arraybackslash}p{.11\textwidth}|
			>{\centering\arraybackslash}p{.09\textwidth}|
			>{\centering\arraybackslash}p{.09\textwidth}|
			>{\centering\arraybackslash}p{.09\textwidth}|
			>{\centering\arraybackslash}p{.09\textwidth}|
			>{\centering\arraybackslash}p{.09\textwidth}|
		}
		\hline
		$a \backslash b$ &  -3 & -2 & -1 & 0 & 1 & 2 & 3 \\
		\hline\hline
		\multirow{3}{*}{$3$} & $Q=18$ &  15 & 12& 9& 6& 5&0\\
		& $\langle x^{81915},x^{45186}y^3\rangle$ & $\langle x^{32767},x^{297}y\rangle$ & $\langle x^{511},x^{322}y\rangle$ & $\langle x^{7},x^{4}y^{73}\rangle$ &$\langle x^{63},x^{22}y\rangle$ & $\langle x^{31},x^{20}y\rangle$ & \multirow{2}{*}{*}\\ 
		& $\langle y^{81915}, x^3y^{52944}\rangle$ & $\langle y^{32767},xy^{27361}\rangle$ & $\langle y^{511},xy^{189}\rangle$ & $\langle y^{511},xy^{146}\rangle$ & $\langle y^{63},xy^{43}\rangle$& $\langle y^{31},xy^{14}\rangle$ & \\ 
		\hline
		\multirow{3}{*}{$2$} & 15  & 13 & 11 & 9& 7& 2&\\
		& $\langle x^{30705},x^{15992}y\rangle$ & $\langle x^{8191},x^{357}y\rangle$ & $\langle x^{1533},x^{1117}y\rangle$ & \multirow{2}{*}{$\langle x^{7},y^{63}\rangle$} & $\langle x^{127},x^{111}y\rangle$ & $\langle x^{3},xy\rangle$& \\ 
		& $\langle y^{30705},xy^{6818}\rangle$ & $\langle y^{8191},xy^{2868}\rangle$ & $\langle y^{1533},xy^{667}\rangle$ &  & $\langle y^{127},xy^{119}\rangle$ & $\langle y^{3},xy\rangle$ &\\ 
		\hline
		\multirow{3}{*}{$1$}& 12 & 11 & 10 & 9& 8& &\\
		& $\langle x^{3937},x^{2043}y\rangle$& $\langle x^{1533},x^{262}y\rangle$ & $\langle x^{1023},x^{262}y\rangle$ & \multirow{2}{*}{$\langle x^{7},y^{63}\rangle$} & $\langle x^{63},x^{39}y^3\rangle$& & \\ 
		& $\langle y^{3937},xy^{1374}\rangle$& $\langle y^{1533},xy^{1334}\rangle$ & $\langle y^{1023},xy^{82}\rangle$ &  & $\langle y^{63},x^3y^{39}\rangle$ &  &\\ 
		\hline
		\multirow{3}{*}{$0$}& 9  & 9 & 9 & 9& & &\\
		&\multirow{2}{*}{$\langle x^{63},y^{7}\rangle$} & $\langle x^{511},x^{292}y\rangle$&$\langle x^{511},x^{73}y\rangle$ &\multirow{2}{*}{$\langle x^{7},y^{7}\rangle$} & & &\\ 
		& & $\langle y^{7},x^{73}y^{2}\rangle$ & $\langle y^{7},x^{73}y\rangle$ &  & & &\\ 
		\hline
		\multirow{3}{*}{$-1$}& 6  & 7 & 8 & & & &\\
		&$\langle x^{63},x^{30}y\rangle$ & $\langle x^{127},x^{119}y\rangle$ & $\langle x^{30},x^{24}y^6\rangle$ &  & & & \\ 
		& $\langle y^{21},x^3y^{19}\rangle$ & $\langle y^{127},xy^{111}\rangle$ & $\langle y^{30},x^6y^{24}\rangle$ &  & & &\\ 
		\hline
		\multirow{3}{*}{$-2$}& 0  & 5 &  & & & &\\
		& \multirow{2}{*}{*}& $\langle x^{21},x^{20}y\rangle$ & & & & & \\ 
		& & $\langle y^{21},xy^{20}\rangle$ & & & & &\\ 
		\hline
		\multirow{3}{*}{$-3$}&6  &  &  & & & &\\
		& $\langle x^{63},x^{62}y\rangle$ & & & & & & \\ 
		& $\langle y^{63},xy^{62}\rangle$ & & & & & &\\ 
		\hline
	\end{tabular}
	\caption{Topological index and mobility sublattices for the codes specified by $f=1+x+x^{-2}y^b$ and $g=1+y+y^{-2}x^a$. We omit the value in right-bottom due to the symmetry between $a$ and $b$.}\label{tabel:22}
\end{table}

\section{Logical qubit count and anyon periods}\label{app:logical_qubit}

Here, we prove the identity presented in Eq.~\eqref{eq:k_relation} of the main text,
which is restated below as Eq.~\eqref{eq:klm}. This identity
asserts that $k(\ell,m)$ is fully determined by the greatest common
divisors $\gcd\left(\ell,\ell_{o}\right)$ and $\gcd\left(m,m_{o}\right)$,
where $\ell$ and $m$ are the linear sizes of the system, and $\ell_{o}$
and $m_{o}$ denote the anyon periods.

This result highlights the
interplay between system sizes and anyon periodicity.

\begin{thm}
	Let $R\coloneqq\mathbb{F}_{2}\left[x^{\pm},y^{\pm}\right]$, and $f,g\in R$.
	Suppose that $\left|R/\left(f,g\right)\right|<\infty$. Let $\ell_{o}$
	and $m_{o}$ be the periods for all anyons along the $x$ and $y$
	directions, respectively. Let 
	\begin{equation}
		k\left(\ell,m\right)\coloneqq\dim_{\mathbb{F}_{2}}\frac{R}{\left(f,g,x^{\ell}-1,y^{m}-1\right)}.
	\end{equation}
	Then we have the identity 
	\begin{equation}
		k\left(\ell,m\right)=k\left(\gcd\left(\ell,\ell_{0}\right),\gcd\left(m,m_{o}\right)\right).\label{eq:klm}
	\end{equation}
\end{thm}

\begin{proof}
	For clarity, the ideal of $R$ generated by $f_{1},f_{2},\cdots,f_{n}$
	is denoted $\left(f_{1},f_{2},\cdots,f_{n}\right)_{R}$. 
	
	By the definition of anyon periods, we have $x^{\ell_{o}}-1,y^{m_{o}}-1\in\left(f,g\right)_{R}$.
	Consequently, $x^{\ell_{o}}-1,y^{m_{o}}-1\in(f,g,x^{\ell}-1,y^{m}-1)_{R}$
	as well. It follows that
	\begin{equation}
		k\left(\ell,m\right)=\dim\frac{R}{\left(f,g,x^{\ell}-1,y^{m}-1,x^{\ell_{o}}-1,y^{m_{o}}-1\right)}=\dim\frac{R}{\left(f,g\right)+\left(x^{\ell}-1,x^{\ell_{o}}-1\right)+\left(y^{m}-1,y^{m_{o}}-1\right)}.
	\end{equation}
	Using Euclidean algorithm, we have $(x^{\ell}-1,x^{\ell_{o}}-1)_{R}=(x^{\gcd\left(\ell,\ell_{o}\right)}-1)_{R}$
	and $(y^{m}-1,y^{m_{o}}-1)_{R}=(y^{\gcd(m,m_{o})}-1)_{R}$. Therefore,
	\begin{equation}
		k\left(\ell,m\right)=\dim_{\mathbb{F}_{2}}\frac{R}{\left(f,g\right)+(x^{\gcd\left(\ell,\ell_{o}\right)}-1)+(y^{\gcd(m,m_{o})}-1)}=k\left(\gcd\left(\ell,\ell_{o}\right),\gcd\left(m,m_{o}\right)\right).
	\end{equation}
	This completes the proof.
\end{proof}

\section{Koszul complex and anyon-logic duality}

The relation between anyons and logical operators in the BB codes
can be understood as a direct consequence of the following duality
property of Koszul complexes (Theorem~\ref{thm:Koszul}).

\begin{thm} \label{thm:Koszul}
	Let $R$ be a commutative ring, and let $a,b,c,d\in R$. Suppose that
	the Koszul complexes
	\begin{equation}
		R\xrightarrow[\phantom{(d\ c)}]{\partial_{2}={b \choose -a}}R\oplus R\xrightarrow{\partial_{1}=(a\ b)}R\qquad\text{and}\qquad R\xrightarrow[\phantom{(d\ c)}]{\sigma_{2}={d \choose -c}}R\oplus R\xrightarrow{\sigma_{1}=(c\ d)}R
	\end{equation}
	are exact. Consider the modified Koszul complexes obtained by replacing
	$R$ with the quotients $\frac{R}{\left(c,d\right)}$ and $\frac{R}{\left(a,b\right)}$,
	respectively:
	\begin{equation}
		\frac{R}{\left(c,d\right)}\xrightarrow[\phantom{(d\ c)}]{\tilde{\partial}_{2}={b \choose -a}}\frac{R}{\left(c,d\right)}\oplus\frac{R}{\left(c,d\right)}\xrightarrow{\tilde{\partial}_{1}=(a\ b)}\frac{R}{\left(c,d\right)}\qquad\text{and}\qquad\frac{R}{\left(a,b\right)}\xrightarrow[\phantom{(d\ c)}]{\tilde{\sigma}_{2}={d \choose -c}}\frac{R}{\left(a,b\right)}\oplus\frac{R}{\left(a,b\right)}\xrightarrow{\tilde{\sigma}_{1}=(c\ d)}\frac{R}{\left(a,b\right)}.
	\end{equation}
	Then, the homologies of these modified complexes, namely,
	\begin{equation}
		H\left(a,b;c,d\right)\coloneqq\ker\tilde{\partial}_{1}/\mathrm{im}\,\tilde{\partial}_{2}\quad\text{and}\quad H\left(c,d;a,b\right)\coloneqq\ker\tilde{\sigma}_{1}/\mathrm{im}\,\tilde{\sigma}_{2}
	\end{equation}
	are isomorphic as $R$-modules.
\end{thm}
\begin{rem}
	The maps $\tilde{\partial}_{1}$, $\tilde{\partial}_{2}$, $\tilde{\sigma}_{1}$,
	and $\tilde{\sigma}_{2}$ have the same matrix representations as
	$\partial_{1}$, $\partial_{2}$, $\sigma_{1}$, and $\sigma_{2}$,
	respectively, but are defined over different domains.
\end{rem}

\begin{proof}
	Straightforwardly from the definitions of the maps, $\ker\tilde{\partial}_{1}=\frac{\partial_{1}^{-1}\left(\text{im }\sigma_{1}\right)}{\text{im}\,\sigma_{1}\oplus\text{im}\,\sigma_{1}}$
	and $\text{im}\,\tilde{\partial}_{2}=\frac{\left(\text{im }\sigma_{1}\oplus\text{im}\,\sigma_{1}\right)+\text{im}\,\partial_{2}}{\text{im}\,\sigma_{1}\oplus\text{im}\,\sigma_{1}}$.
	Thus, 
	\begin{align}
		H\left(a,b;c,d\right) & =\frac{\ker\tilde{\partial}_{1}}{\text{im}\,\tilde{\partial}_{2}}\cong\frac{\partial_{1}^{-1}\left(\text{im }\sigma_{1}\right)}{\left(\text{im}\,\sigma_{1}\oplus\text{im}\,\sigma_{1}\right)+\text{im}\,\partial_{2}}=\frac{\sigma_{1}^{-1}\left(\text{im}\,\sigma_{1}\right)}{\left(\text{im}\,\sigma_{1}\oplus\text{im}\,\sigma_{1}\right)+\ker\partial_{1}},\label{eq:K1-1}
	\end{align}
	where in the last equality we use $\text{im}\,\partial_{2}=\ker\partial_{1}$
	due to the exactness assumption. Similarly, we have 
	\begin{equation}
		H\left(c,d;a,b\right)\cong\frac{\sigma_{1}^{-1}\left(\text{im}\,\partial_{1}\right)}{\left(\text{im}\,\partial_{1}\oplus\text{im}\,\partial_{1}\right)+\ker\sigma_{1}}.
	\end{equation}
	
	Let $W\coloneqq\text{im }\sigma_{1}\cap\text{im }\partial_{1}$. Then
	$\partial_{1}^{-1}\left(\text{im}\,\sigma_{1}\right)=\sigma_{1}^{-1}\left(W\right)$
	and $\sigma_{1}^{-1}\left(\text{im}\,\partial_{1}\right)=\sigma_{1}^{-1}\left(W\right)$.
	Therefore,
	\begin{equation}
		H\left(a,b;c,d\right)\cong\frac{\partial_{1}^{-1}\left(W\right)}{\left(\text{im}\,\sigma_{1}\oplus\text{im}\,\sigma_{1}\right)+\ker\partial_{1}},\qquad H\left(c,d;a,b\right)\cong\frac{\sigma_{1}^{-1}\left(W\right)}{\left(\text{im}\,\partial_{1}\oplus\text{im}\,\partial_{1}\right)+\ker\sigma_{1}}.
	\end{equation}
	
	Let $W_{0}\coloneqq\partial_{1}\left(\left(\text{im}\,\sigma_{1}\oplus\text{im}\,\sigma_{1}\right)+\ker\partial_{1}\right)=\partial_{1}\left(\text{im}\,\sigma_{1}\oplus\text{im}\,\sigma_{1}\right)=\left(ac,ad,bc,bd\right)_{R}$.
	Let $\partial_{W}:\partial_{1}^{-1}\left(W\right)\rightarrow W/W_{0}$
	be the composition of the $R$-module homomorphisms
	\begin{equation}
		\partial_{W}:\quad\partial_{1}^{-1}\left(W\right)\xrightarrow{\partial_{1}}W\rightarrow\frac{W}{W_{0}}.
	\end{equation}
	By construction, $\partial_{W}$ is surjective and $\ker\partial_{W}=\left(\text{im}\,\sigma_{1}\oplus\text{im}\,\sigma_{1}\right)+\ker\partial_{1}$.
	Thus, 
	\begin{equation}
		H\left(a,b;c,d\right)\cong\frac{W}{W_{0}}.
	\end{equation}
	
	Note $\sigma_{1}\left(\left(\text{im }\partial_{1}\oplus\text{im }\partial_{1}\right)+\ker\sigma_{1}\right)=W_{0}$.
	Thus, analogously, 
	\begin{equation}
		\sigma_{W}:\quad\sigma_{1}^{-1}\left(W\right)\xrightarrow{\sigma_{1}}W\rightarrow\frac{W}{W_{0}}
	\end{equation}
	induces an $R$-module isomorphism
	\begin{equation}
		H\left(c,d;a,b\right)\cong\frac{W}{W_{0}}. 
	\end{equation}
	
	Therefore, $H\left(a,b;c,d\right)\cong H\left(c,d;a,b\right)$ as
	$R$-modules. 
\end{proof}

\begin{rem}
	We emphasize that this is an $R$-module isomorphism. Physically, this implies that not only are the vector dimensions preserved, but also the structure of the translation symmetry is respected.
\end{rem}

The duality property of Koszul complexes proves the generalized connection
between anyons and logical operators claimed in the main text based
on physical arguments. In the discussion of BB codes, $\left(a,b\right)=(x^{\ell}-1,y^{m}-1)$
corresponds to periodic boundary conditions and $\left(c,d\right)=\left(g,f\right)$
describes the check operators.

\section{product construction}

Product constructions are instrumental tools for discovering good QLDPC codes and exploring exotic quantum phases of matter~\cite{Panteleev22,Panteleev22b,tan2023fracton}. We first review the hypergraph product (HGP) and balanced product (BP) constructions of quantum CSS
codes in the language of chain complexes. Furthermore, we will demonstrate the construction of GB and BB codes using these product approaches.

A classical error-correcting code can be described as a two-term chain complex:
\begin{equation}\label{classical_chain}
	V_1 \overset{h^\dagger}{\longrightarrow} V_0.
\end{equation}
In this formulation, $V_1=\mathbb{F}_2^m$ and $V_0=\mathbb{F}_2^n$ denote vector spaces over the finite field $\mathbb{F}_2$ corresponding to checks and bits, respectively. The map $h\in \mathbb{F}_2^{n\times m}$ represents the parity check matrix, where an element $h_{ij}=1$ indicates that the $j$-th bit is supported by the $i$-th check. To connect to quantum CSS codes, we extend to the three-term chain complex and associate the $X$-stabilizers, qubits, and $Z$-stabilizers to $\mathbb{F}_2$-vector spaces $V_2, V_1, V_0$ respectively.
\begin{equation}\label{CSS_chain}
	V_2 \overset{h_X}{\longrightarrow} V_1 \overset{h_Z^\dagger}{\longrightarrow} V_0
\end{equation}  
The requirement $h_Z^T\circ h_X=0$ ensures that $X$-stabilizers and $Z$-stabilizers mutually intersect an even number of qubits, thereby maintaining the commutation relations between the $Z$ and $X$ stabilizers.

The hypergraph product and balanced product of two classical codes, within the chain complex framework, can be formulated as tensor products of two chain complexes.
Consider two classical codes $\mathcal{C}$ and $\mathcal{D}$, each exhibiting a symmetry under the action of a group $G$. We assume $G$ to be abelian,  although this can be generalized to non-abelian cases.
The representations of $\mathcal{C}$ and $\mathcal{D}$ are given as follows:
\begin{align}
	\mathcal{C} &: C_1 \overset{f}{\longrightarrow} C_0 \\
	\mathcal{D} &: D_1 \overset{g}{\longrightarrow} D_0.
\end{align}
The tensor product of two chain complexes over $G$ is given by:
\begin{equation}\label{BP_chain}
	\mathcal{C}\otimes_G\mathcal{D}: C_1\otimes_G D_1 \xrightarrow{h_X} (C_0\otimes_G D_1)\oplus
	(C_1\otimes_G D_0)\xrightarrow{h_Z^\dagger} C_0\otimes_G D_0.
\end{equation}
where the boundary maps are defined as:
\begin{align}
	h_X &= f \otimes_G \text{id} + \text{id} \otimes_G g,\\
	h_Z^\dagger &=\text{id} \otimes_G g - f \otimes_G \text{id},
\end{align}
the notation $\otimes_G$ is equivalent to $\otimes_R$, where $R=\mathbb{F}_2G$ is a group algebra. The CSS code described by the above chain complex is referred to as the balanced product of $\mathcal{C}$ and $\mathcal{D}$. The hypergraph product is a particular instance of the balanced product when the group $G$ is trivial. Moreover, if $G$ acts freely on each of $C_i$ and $D_i$, the balanced product reduces to the lifted product~\cite{Breuckmann21a}.
\begin{ex}
	When $f$ is a univariate polynomial in $x$ and $g$ is a univariate polynomial in $y$. The BB code is the HGP of two classical cyclic codes of lengths $\ell$ and $m$, characterized by generating polynomials $f(x)$ and $g(y)$, respectively. This relation is described as follows
	\begin{equation}
		R_{\ell,m}\xrightarrow{\; h_{X}={f \choose g}\;}R_{\ell,m}^{2}\xrightarrow{h_{Z}^{\dagger}=\left(g\ f\right)}R_{\ell,m},
	\end{equation} 
	where $R_{\ell,m} \coloneqq \frac{\mathbb{F}_2[x]}{(x^\ell-1)} \otimes \frac{\mathbb{F}_2[y]}{(y^m-1)} = \frac{\mathbb{F}_2[x,y]}{(x^{\ell}-1, y^m-1)}$.
	Here we use the polynomial formalism to represent the codes. With $x=S_\ell$ and $y=S_m$, where $S_n$ is a cyclic shift matrix of size $n\times n$, this provides a binary representation equivalent to what is discussed in Eq.~\eqref{classical_chain} and Eq.~\eqref{CSS_chain}.
\end{ex}
\begin{ex}\cite{Eberhardt24}
	Consider $G=\mathbb{Z}_\ell \otimes \mathbb{Z}_m$ and $R_{\ell,m} = \mathbb{F}_2 G \cong \mathbb{F}_2[x,y]/(x^\ell -1, y^m-1)$. Two classical codes with $G$-symmetry are given by
	\begin{align}
		\mathcal{C}&: R_{\ell,m} \xrightarrow{f(x,y)} R_{\ell,m}, \\
		\mathcal{D}&: R_{\ell,m} \xrightarrow{g(x,y)} R_{\ell,m}.
	\end{align}
	Note that $R_{\ell,m}\otimes_G R_{\ell,m} \cong R_{\ell,m}$. By modding out the diagonal $G$ action, we obtain BB($f$, $g$), which corresponds to $\mathcal{C} \otimes_G \mathcal{D}$. Since $G$ acts freely on $R_{\ell,m}$, BB codes are also lifted product codes.
\end{ex}
\begin{rem}
	A generalized bicycle (GB) code of length $2 n$ is equivalent to a rotated HGP code with periodicity vectors $\vec{L}_1$ and $\vec{L}_2$ such that $|\vec{L}_1 \times \vec{L}_2|=n$ \cite{Wang22}. A BB code becomes a GB code if $\gcd(\ell, m)=1$.
\end{rem}

\section{Subsystem symmetry structure of BB codes}

Fracton (or fracton-like) physics can often be theoretically connected to subsystem symmetries~\cite{Vijay16} or broader classes of spatially modulated symmetries~\cite{Sala22,Delfino2024}. Here, we investigate this perspective for BB codes, motivated by their quasi-fractonic features. Specifically, we show that BB codes can be interpreted as lattice gauge theories associated with subsystem symmetries.

We explicitly ``ungauge'' BB codes to specify the subsystem symmetry
structure. This is done by employing the algebraic formulation developed
in Ref.~\cite{Vijay16}, which generalizes the Ising--toric code
correspondence~\cite{Vijay16,Canossa24}. The ungauging result for BB code $\mathbb{BB}\left(f,g\right)$
is a generalized Ising model specified by the algebraic map
\[
R^{2}\xrightarrow{h_{Z}^{\dagger}=(f,\,g)}R,
\]
where spins are assigned to vertices of the lattice, and the polynomials
$f$ and $g$ define two types of spin interactions. The dual map
$h_{Z}$ governs how spin-flip operations generate excitations in
this model.

The kernel \( \ker h_Z \) thus encodes the space of subsystem spin-flip symmetries. Each element
\[
\xi = \sum_{i,j} \varphi_{i,j} x^i y^j \in \ker h_Z
\]
selects a (possibly infinite) set of lattice sites whose simultaneous flip creates no excitation. The coefficients \( \varphi_{i,j} \) can be viewed as a lattice function: \(\varphi_{i,j}=1\) marks a spin to flip, while \(\varphi_{i,j}=0\) leaves it untouched.

\end{document}